\pdfoutput=1 \pdfminorversion=5 \pdfobjcompresslevel=3 \pdfcompresslevel=9

\PassOptionsToPackage{backref=page,bookmarksopen=true}{hyperref}

\documentclass[11pt,letterpaper]{article}
\usepackage[utf8]{inputenc}
\usepackage{fullpage}

\usepackage{caption}
\usepackage{subcaption}
\usepackage{soul}
\usepackage{amsmath}
\usepackage{amsthm}
\usepackage{amsfonts}
\usepackage{mathtools}

\usepackage{thmtools}
\usepackage{thm-restate}

\usepackage{enumitem}
\usepackage{float}
\usepackage{adjustbox}
\usepackage{graphicx}
\usepackage[svgnames]{xcolor}
\usepackage{tikz}
\usetikzlibrary{calc,positioning,shapes,patterns}
\pgfdeclarelayer{background}
\pgfsetlayers{background,main}

\usepackage{titlesec}
\titlespacing{\paragraph}{0pt}{0.35\baselineskip plus 0.25\baselineskip minus 0.1\baselineskip}{1em}

\usepackage[textsize=scriptsize,textwidth=2.125cm,disable]{todonotes}

\usepackage{hyperref}

\newcommand{\OO}{\mathcal{O}}

\newcommand{\bag}{\mathcal{B}}
\newcommand{\children}{\operatorname{children}}
\newcommand{\poly}{\operatorname{poly}}

\DeclarePairedDelimiter{\abs}{\lvert}{\rvert}

\DeclarePairedDelimiter{\floor}{\lfloor}{\rfloor}

\DeclarePairedDelimiter{\set}{\lbrace}{\rbrace}
\newcommand{\suchthat}{\mathrel{}\mathclose{}\ifnum\currentgrouptype=16\middle\fi\vert\mathopen{}\mathrel{}}

\DeclarePairedDelimiter{\paren}{\lparen}{\rparen}

\newcommand{\emb}{\phi}
\newcommand{\Emb}{\operatorname{Emb}}
\newcommand{\EmbGood}{\Emb^\star}
\newcommand{\EmbSolid}{\Emb^{+}}
\newcommand{\flipdist}{\operatorname{dist}}
\newcommand{\flipparity}{\operatorname{flip-parity}}
\newcommand{\gluecorner}{\operatorname{glue-corner}}

\declaretheorem[style=plain,name={Theorem}]{theorem}
\declaretheorem[style=plain,name={Lemma},sibling=theorem]{lemma}
\declaretheorem[style=plain,name={Corollary},sibling=theorem]{corollary}

\declaretheorem[style=plain,name={Observation},sibling=theorem]{observation}
\declaretheorem[style=definition,name={Definition},sibling=theorem]{definition}

\usepackage{authblk}
\usepackage{hyperref}
\title{Worst-Case Polylog Incremental SPQR-trees: \\
Embeddings, Planarity, and Triconnectivity.}
\author[1]{Jacob Holm\thanks{Partially supported by the VILLUM Foundation grant 16582, Basic Algorithms Research Copenhagen (BARC).}}
\author[2]{Eva Rotenberg\thanks{Partially supported by Independent Research Fund Denmark grant "AlgoGraph" 2018-2021 (8021-00249B).}}
\affil[1]{University of Copenhagen \hspace{1em}{\small \href{mailto:jaho@di.ku.dk}{jaho@di.ku.dk}}}
\affil[2]{Technical University of Denmark \hspace{1em}{\small \href{mailto:eva@rotenberg.dk}{erot@dtu.dk}}}
\date{}

\begin{document}
\thispagestyle{empty}	
\maketitle
\begin{abstract}
  We show that every labelled planar graph $G$ can be assigned a
  canonical embedding $\emb(G)$, such that for any planar $G'$
  that differs from $G$ by the insertion or deletion of one edge, the number of
  local changes to the combinatorial embedding needed to get from $\emb(G)$ to $\emb(G')$
  is $\OO(\log n)$. 

In contrast, there exist embedded graphs where $\Omega(n)$ changes are necessary to accommodate one inserted edge. 
We provide a matching lower bound of $\Omega(\log n)$ local changes, and although our upper bound is worst-case, our lower bound hold in the amortized case as well.

Our proof is based on BC trees and SPQR trees, and we develop
\emph{pre-split} variants of these for general graphs, based on a
novel biased heavy-path decomposition, where the structural changes
corresponding to edge insertions and deletions in the underlying
graph consist of at most $\OO(\log n)$ basic operations of a
particularly simple form.

  As a secondary result, we show how to maintain the pre-split trees
  under edge insertions in the underlying graph deterministically in
  worst case $\OO(\log^3 n)$ time.  Using this, we obtain
  deterministic data structures for incremental planarity testing,
  incremental planar embedding, and incremental triconnectivity, that
  each have worst case $O(\log^3 n)$ update and query time,
  answering an open question by La Poutré and Westbrook from 1998.

\end{abstract}

\thispagestyle{empty}
\newpage
\setcounter{page}{1}

\section{Introduction}
The motivation behind dynamic data structures such as dynamic graphs is that local changes, such as the insertion of an edge, should only have limited influence on the global properties of the graph, and thus, after a local change, one needs not process the entire graph again in order to have substantial information about properties of the graph.
An example of a class of well-studied questions is that of connectivity and $k$-edge/vertex-connectivity: upon the insertion or deletion of an edge, 
a representation of the graph is quickly updated, such that later \emph{queries} to whether an at query-time specified pair of vertices are connected or $k$-edge/vertex-connected, can be answered promptly. 

For planarity, the most well-studied query is that of whether a given edge can be added without violating planarity, or without violating the planar embedding. Other natural queries include questions about whether a component of the graph is itself planar. 
In this paper, our main focus of study is not the question of \emph{whether} the graph is planar, but \emph{how} the graph is embedded in the plane. Here, and throughout the paper, we view embeddings from a combinatorial point of view: a graph is embedded if each edge, for each of its endpoints, can compute its right and left neighbours in the circular ordering. An embedding is planar if said circular ordering around the endpoints is realisable as a planar drawing of the graph (see \cite{KleinMozesOnline}). 

In this setting, \emph{queries} to the neighbours around an endpoint of an edge are well-defined, and local changes obtain a new and interesting meaning: for each edge in the graph, we can say that it is \emph{affected} by an update if its set of neighbours of either endpoint has changed. A \emph{local change} to the graph \emph{or its embedding} is any change such that only a constant number of edges are affected (where the constant depends only on the type of operation changing the graph).
The local changes we will consider are: the insertion of an edge across a face, or the deletion of an edge, as well as the \emph{flip} operations that keeps the graph fixed, but changes the embedding by taking a subgraph that is separated from the rest of the graph by either a separation pair (a \emph{separation flip}), or an articulation point (an \emph{articulation flip}), and either \emph{reflecting} it, or \emph{sliding} it to a new position in the edge order of the separating vertices (see Figure~\ref{fig:flip}).

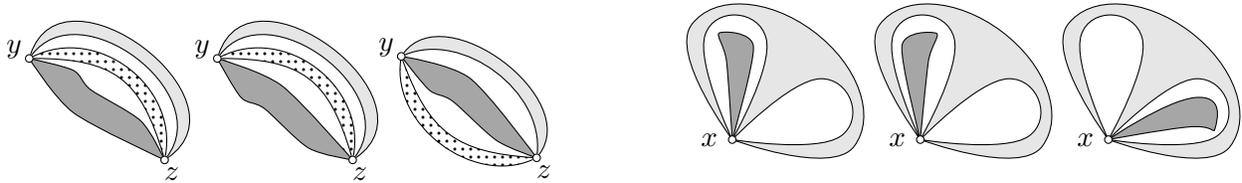
\begin{figure}[htb]
	\begin{minipage}[t]{.45\textwidth}
		\centering
		\begin{tikzpicture}[
    vertex/.style={
      draw,
      circle,
      fill=white,
      minimum size=1mm,
      inner sep=0pt,
    },
  ]
  \begin{scope}[shift={(2.5,0)},rotate=-36.87,scale=.45]
    \path[use as bounding box] (0,-1.5) rectangle (4,2);
    \node[vertex,label={[label distance=-2mm]105:$y$}] (x) at (0,0) {};
    \node[vertex,label={[label distance=-2mm]-15:$z$}] (y) at (5,0) {};

    \coordinate (x1) at (0,2.5);
    \coordinate (y1) at (5,2.5);
    \coordinate (x2) at (-1,3);
    \coordinate (y2) at (6,3);

    \coordinate (c1) at (2,-.35) {};
    \coordinate (x3) at (1.25,-.5) {};
    \coordinate (y3) at (3.75,-.5) {};
    \coordinate (c2) at (2,.5) {};

    \coordinate (x4) at (2,1.5);
    \coordinate (y4) at (3,1.5);
    \coordinate (x5) at (1,2);
    \coordinate (y5) at (4,2);

    \begin{pgfonlayer}{background}
      \draw[fill=gray!20] (x) .. controls (x1) and (y1) .. (y)
      .. controls (y2) and (x2) .. (x);

      \draw[fill=gray!40,pattern=dots] (x) .. controls (x4) and (y4) .. (y)
      .. controls (y5) and (x5) .. (x);

      \draw[fill=gray!70] plot[smooth] coordinates {(x) (x3) (c1) (y3) (y)}
      plot [smooth] coordinates {(y) (c2) (x)};
    \end{pgfonlayer}
  \end{scope}

  \begin{scope}[shift={(0,0)},rotate=-36.87,scale=.45]
    \path[use as bounding box] (0,-1.5) rectangle (4,2);
    \node[vertex,label={[label distance=-2mm]105:$y$}] (x) at (0,0) {};
    \node[vertex,label={[label distance=-2mm]-15:$z$}] (y) at (5,0) {};

    \coordinate (x1) at (0,2.5);
    \coordinate (y1) at (5,2.5);
    \coordinate (x2) at (-1,3);
    \coordinate (y2) at (6,3);

    \coordinate (c1) at (2,.35) {};
    \coordinate (x3) at (1.25,.5) {};
    \coordinate (y3) at (3.75,.5) {};
    \coordinate (c2) at (2,-.5) {};

    \coordinate (x4) at (2,1.5);
    \coordinate (y4) at (3,1.5);
    \coordinate (x5) at (1,2);
    \coordinate (y5) at (4,2);

    \begin{pgfonlayer}{background}
      \draw[fill=gray!20] (x) .. controls (x1) and (y1) .. (y)
      .. controls (y2) and (x2) .. (x);

      \draw[fill=gray!40,pattern=dots] (x) .. controls (x4) and (y4) .. (y)
      .. controls (y5) and (x5) .. (x);

      \draw[fill=gray!70] plot[smooth] coordinates {(x) (x3) (c1) (y3) (y)}
      plot [smooth] coordinates {(y) (c2) (x)};
    \end{pgfonlayer}

  \end{scope}

  \begin{scope}[shift={(5,0.1)},rotate=-36.87,scale=.45]
    \path[use as bounding box] (0,-1.7) rectangle (4,1.8);
    \node[vertex,label={[label distance=-2mm]105:$y$}] (x) at (0,-.2) {};
    \node[vertex,label={[label distance=-2mm]-15:$z$}] (y) at (5,-.2) {};

    \coordinate (x1) at (0,2.5);
    \coordinate (y1) at (5,2.5);
    \coordinate (x2) at (-1,3);
    \coordinate (y2) at (6,3);

    \coordinate (c1) at (2,-.35) {};
    \coordinate (x3) at (1.25,-.5) {};
    \coordinate (y3) at (3.75,-.5) {};
    \coordinate (c2) at (2,.5) {};

    \coordinate (x4) at (1,1.5);
    \coordinate (y4) at (4,1.5);
    \coordinate (x5) at (0,2);
    \coordinate (y5) at (5,2);

    \coordinate (x6) at (2,-1.5);
    \coordinate (y6) at (3,-1.5);
    \coordinate (x7) at (1,-2);
    \coordinate (y7) at (4,-2);

    \begin{pgfonlayer}{background}
      \draw[fill=gray!20] (x) .. controls (x4) and (y4) .. (y)
      .. controls (y5) and (x5) .. (x);

      \draw[fill=gray!40,pattern=dots] (x) .. controls (x6) and (y6) .. (y)
      .. controls (y7) and (x7) .. (x);

      \draw[fill=gray!70] plot[smooth] coordinates {(x) (x3) (c1) (y3) (y)}
      plot [smooth] coordinates {(y) (c2) (x)};
    \end{pgfonlayer}
  \end{scope}

\end{tikzpicture}
 		\vspace{.1em}	
		\subcaption{Separation flips: reflect and  slide.}\label{fig:separation-flip}
	\end{minipage}  
	\hfill
	\begin{minipage}[t]{.45\textwidth}
		\centering
		\begin{tikzpicture}[
    vertex/.style={
      draw,
      circle,
      fill=white,
      minimum size=1mm,
      inner sep=0pt,
    },
  ]
  \begin{scope}[shift={(2.5,0)},rotate=-36.87, scale=.45]
    \path[use as bounding box] (-2,0) rectangle (2,4);
    \node[vertex,label={180:$x$}] (x) at (0,0) {};

    \coordinate (a1) at (-3,2);
    \coordinate (a2) at (0,4);
    \coordinate (a3) at (3,2);
    \coordinate (b1) at (2.5,2);
    \coordinate (b2) at (1,3);
    \coordinate (c1) at (-1,3);
    \coordinate (c2) at (-2.5,2);

    \coordinate (d1) at (-1.2,2);
    \coordinate (d2) at (-1.5,2.8);
    \coordinate (d3) at (-2.1,1.9);

    \begin{pgfonlayer}{background}
      \draw[fill=gray!20]
      plot[smooth,tension=1] coordinates {
        (x) (a1) (a2) (a3) (x)
      }
      plot[smooth,tension=1] coordinates {
        (x) (b1) (b2) (x)
      }
      plot[smooth,tension=1] coordinates {
        (x) (c1) (c2) (x)
      };

      \draw[fill=gray!70]
      plot[smooth,tension=.5] coordinates {
        (x) (d1) (d2) (d3) (x)
      };
    \end{pgfonlayer}
  \end{scope}

  \begin{scope}[shift={(0,0)},rotate=-36.87, scale=.45]
    \path[use as bounding box] (-2,0) rectangle (2,4);
    \node[vertex,label={180:$x$}] (x) at (0,0) {};

    \coordinate (a1) at (-3,2);
    \coordinate (a2) at (0,4);
    \coordinate (a3) at (3,2);
    \coordinate (b1) at (2.5,2);
    \coordinate (b2) at (1,3);
    \coordinate (c1) at (-1,3);
    \coordinate (c2) at (-2.5,2);

    \begin{scope}[rotate=72,xscale=-1]
    \coordinate (d1) at (-1.2,2);
    \coordinate (d2) at (-1.5,2.8);
    \coordinate (d3) at (-2.1,1.9);
    \end{scope}

    \begin{pgfonlayer}{background}
      \draw[fill=gray!20]
      plot[smooth,tension=1] coordinates {
        (x) (a1) (a2) (a3) (x)
      }
      plot[smooth,tension=1] coordinates {
        (x) (b1) (b2) (x)
      }
      plot[smooth,tension=1] coordinates {
        (x) (c1) (c2) (x)
      };

      \draw[fill=gray!70]
      plot[smooth,tension=.5] coordinates {
        (x) (d1) (d2) (d3) (x)
      };
    \end{pgfonlayer}
  \end{scope}

  \begin{scope}[shift={(5,0)},rotate=-36.87,scale=.45]
    \path[use as bounding box] (-2,0) rectangle (2,4);
    \node[vertex,label={180:$x$}] (x) at (0,0) {};

    \coordinate (a1) at (-3,2);
    \coordinate (a2) at (0,4);
    \coordinate (a3) at (3,2);
    \coordinate (b1) at (2.5,2);
    \coordinate (b2) at (1,3);
    \coordinate (c1) at (-1,3);
    \coordinate (c2) at (-2.5,2);

    \begin{scope}[rotate=-75]
    \coordinate (d1) at (-1.2,2);
    \coordinate (d2) at (-1.5,2.8);
    \coordinate (d3) at (-2.1,1.9);
    \end{scope}

    \begin{pgfonlayer}{background}
      \draw[fill=gray!20]
      plot[smooth,tension=1] coordinates {
        (x) (a1) (a2) (a3) (x)
      }
      plot[smooth,tension=1] coordinates {
        (x) (b1) (b2) (x)
      }
      plot[smooth,tension=1] coordinates {
        (x) (c1) (c2) (x)
      };

      \draw[fill=gray!70]
      plot[smooth,tension=1] coordinates {
        (x) (d1) (d2) (d3) (x)
      };
    \end{pgfonlayer}
  \end{scope}
\end{tikzpicture}
 		\vspace{.1em}	
		\subcaption{Articulation flips: reflect and slide.}\label{fig:articulation-flip}
	\end{minipage}
	\caption{Local changes to the embedding of a graph.\label{fig:flip}}
\end{figure}

In other words, one can say that there are two types of local changes: The ones that keep the embedding fixed but changes the graph by deleting an edge or inserting an edge across a face -- i.e. \emph{updates to the graph} --, and the ones that keep the graph fixed but change the embedding as in Figure~\ref{fig:flip} above -- i.e. \emph{updating the embedding}. 
We show an interesting connection between these two types of local changes, namely a class of embeddings where only worst-case $O(\log n)$ updates to the embedding are necessary to accommodate any update to the graph. This is asymptotically tight, matching an $\Omega(\log n)$ lower bound.

We give an analysis of changes to the embedding that can accommodate fully-dynamic  changes to the graph, that is, both edge insertions and edge deletions. As long as the resulting graph can be embedded in the plane, we will maintain an implicit representation of its embedding. Our analysis goes only via understanding the edge-insertion case. In terms of counting changes to the embedding, the deletion case is symmetric: to delete $e$, find the ``canonical" embedding of $G-e$, add $e$, and record the necessary changes to the embedding; to get from $G$ to $G-e$ perform their opposites. Thus, computation time aside, in terms of counting changes to the embedding, it is enough to study the incremental case. 

This notion of an analysis that counts only the \emph{changes of heart} is akin to the field of \emph{online algorithms with recourse} introduced by Imase and Waxman~\cite{DBLP:journals/siamdm/ImaseW91} for the problem of Dynamic Steiner Trees, and since studied e.g. for Packing, Covering, Edge-Orientations, and Perfect Matchings~\cite{DBLP:journals/ipl/AvitabileMP13, DBLP:conf/wads/BrodalF99, DBLP:conf/wads/GroveKKV95}.

As a welcome side effect, our results entail improved computation times for a series of related problems in the incremental setting. Specifically, since our bound for the changes to the embedding is worst-case, we obtain new \emph{worst-case} polylog incremental algorithms for planarity and $3$-connectivity, problems that were previously only efficiently solved in terms of amortised analysis. This is interesting on two accounts: One is from a practical point of view, where it is an asset to guarantee fast update-times, another is from a theoretical point of view, namely that this allows for unlimited undo.

In terms of technical contribution, our path towards understanding how embeddings are related, and edges are accommodated, goes via \emph{BC-trees and SPQR-trees}. Intuitively, for every connected graph, its BC-tree describes all its $2$-connected components, its articulation points, and how these are related. Similarly, the SPQR-tree describes the $3$-connected components and their relations via separation pairs. It is well-known that the SPQR-trees and the BC-tree can be used to count the number of embeddings -- here, we use them to dynamically maintain a specification of an embedding. 

An edge insertion in the graph becomes a path contraction in the BC-tree, and several path contractions in SPQR-trees; one for each block on said BC-path. To obtain our $O(\log n)$ bound, we need to ensure that these paths, though they may be long, only lead to limited changes to the embedding. They should be, so to speak, almost everywhere ready-to-contract. 
For solving this, our main idea is to maintain heavy path decompositions over the BC-tree and the SPQR-trees. Briefly speaking, we assign each edge to be either heavy or light in a way that has the property that the heavy edges form a forest of paths, and every pair of vertices are at most $O(\log n)$ light edges apart. Once we have made sure that the heavy paths are ready-to-contract, we only need to accommodate the $O(\log n)$ light edges, leading to the $O(\log n)$ total changes to the embedding. Since a single edge insertion in a connected component can lead to contractions of paths in the BC-tree and in several SPQR-trees, we introduce a weighting of nodes that ensures that the total number of light edges that are affected is still bounded.

\subsection{Our results}
In this paper, we present the following theorems:
\begin{theorem}\label{thm:flips}
For every planar graph $G$ there exists a \emph{canonical embedding} $\emb(G)$, such that for any edge $e$ in $G$, the number of local changes between $\emb(G-e)$ and $\emb(G)-e$ is bounded by $O(\log n)$.
\end{theorem}
A matching lower bound is derived from the following construction:
\begin{theorem}\label{thm:intro-lowerbound}
There exists a family of planar graphs $\set{G_h}_{h\in \mathbb{N}}$ where $G_h$ has $O(2^h)$ vertices, and any embedding $\emb$ of $G_h$ has an edge $e\notin G_h$ such that inserting $e$ requires $h$ flips to $\emb$. 
\end{theorem}
Thus, by repeatedly inserting an edge and deleting it again, we get a lower bound that matches Theorem~\ref{thm:flips}:
\begin{corollary}\label{cor:intro-lowerbound}
For every $n\in \mathbb{N}$, we can initiate a dynamic graph with $n$ vertices, and provide an adaptive sequence (of arbitrary length) of alternating edge insertions and edge deletions,
where every edge insertion requires $\Omega(\log n)$ flips.
\end{corollary}
We provide algorithmic results for incremental graphs with worst-case guarantees:
\begin{theorem}\label{thm:incr3con}
There exists a deterministic data structure for incremental triconnectivity that handles edge insertions in $O(\log^3 n)$ worst-case time, and answers queries in $O(\log^3 n)$ worst-case time.
\end{theorem}

\noindent Note that this (Theorem~\ref{thm:incr3con}) holds for general graphs, not only planar graphs.

\begin{theorem}\label{thm:incrplanarity}
There exists a deterministic data structure for incremental planarity that handles edge insertions in $O(\log^3 n)$ worst-case time, and answers queries to whether an edge can be inserted, and to the neighbors of a given existing edge in the current embedding, in $O(\log n)$ worst-case time.
\end{theorem}

\subsection{Previous work}
\paragraph{Biconnectivity.} In 1998~\cite{DBLP:journals/siamcomp/PoutreW98} La Poutré and Westbrook gave algorithms for maintaining $2$-edge and $2$-vertex connnectivity (here called \emph{biconnectivity}) in worst case $\OO(\log n)$ time per operation, under edge insertions with \emph{backtracking}.  They used a number of the same techniques that we do, and stated as an open problem whether they could be extended to triconnectivity and planarity testing.

\paragraph{SPQR-trees.} Understanding $3$-vertex connectivity (here called \emph{triconnectivity}) in a graph via a structure over its triconnected components, its separation pairs, and their relations, dates back to Saunders Mac Lane~\cite{maclane1937}, who also in his work discussed the connections between embeddings and triconnectivity via said structure. This was the mathematical foundation of the algorithms by Hopcroft and Tarjan~\cite{DBLP:journals/siamcomp/HopcroftT73,DBLP:journals/jacm/HopcroftT74} for calculating the triconnected components and, respectively, for determining whether a graph is planar, in linear time. Later, this structure was dubbed SPQR-tree by Di~Battista and Tamassia~\cite{DBLP:conf/focs/BattistaT89, DBLP:conf/icalp/BattistaT90}, who devised the first efficient incremental algorithm for planarity testing, handling updates to the graph in amortised $O(\log n)$ time. 

\paragraph{Planarity testing and triconnectivity.} 
\begin{sloppypar} 
The amortized update time of $O(\log n)$ by Di~Battista and Tamassia
was soon improved to expected $O(k\alpha(k,n))$ total time for $k$
operations by Westbrook~\cite{DBLP:conf/icalp/Westbrook92}, where $\alpha$ denotes the inverse of Ackermann's function. The
optimal total time of deterministic $O(k\alpha(k,n))$ for $k$
operations is due to La~Poutré~\cite{DBLP:conf/stoc/Poutre94}.
Further $\alpha$-time algorithms for incremental triconnectivity were
given by Di~Battista and
Tamassia~\cite{DBLP:journals/algorithmica/BattistaT96}, specifically,
their algorithm spends $\OO(k\alpha(k,n))$ total time for $k$
insertions and queries when the initial graph is biconnected, and
$\OO(n\log n+k)$ for general graphs.
\end{sloppypar}

The incremental algorithms referenced above achieve optimal or nearly optimal amortised running time, but may take up to linear worst-case time for an insertion. This is contrasted by the work by Galil et~al.\cite{DBLP:journals/jacm/GalilIS99}, who sought worst-case sublinear algorithms for this and other problems. Indeed, they give a data structure that handles any insertion or deletion of an edge in deterministic $O(n^{2/3})$ time while maintaining whether the graph is planar, and facilitating queries to triconnectivity between vertices in $O(n^{2/3})$ time. 

\paragraph{Biased search trees and dynamic trees} Our new technique goes via maintaining a heavy path decomposition of the SPQR-tree, as our graph changes dynamically. The heavy path decomposition was invented by Sleator and Tarjan~\cite{DBLP:journals/jcss/SleatorT83}, as a tool for handling dynamic forests
subject to links and cuts in worst case $\OO(\log n)$ time per operation.

This data structure for dynamic trees in turn uses the biased search
trees by Bent~et~al.~\cite{DBLP:journals/siamcomp/BentST85} to
represent each heavy path.  In a biased search tree $T$, each node $v$
has an associated weight $w(v)$, and the depth of $v$ is
$\OO(1+\log\frac{w(T)}{w(v)})$. This is in contrast to other balanced binary search trees, where the depth of a node is typically $\OO(\log n)$.

\paragraph{Dynamic embeddings}
The question about maintaining a dynamic embedded graph can be asked in many ways. Tamassia~\cite{DBLP:journals/jal/Tamassia96} gives an algorithm for a fixed embedding of a graph that handles edge-insertion across a face and ``undo" in 
amortized $O(\log n)$ time per operation. Italiano~et~al.~\cite{DBLP:conf/esa/ItalianoPR93} give an algorithm that maintains an embedded graph subject to deletions of edges and insertions across a face in $\OO(\log^2 n)$ time per update.

While the previous papers only allowed for edge insertions across a face, Eppstein~\cite{DBLP:conf/soda/Eppstein03} takes a different approach: The edge may be inserted in any way specified, possibly increasing the genus of the embedding, but as long as the genus $g$ is low, the update time for minimum spanning tree, the fundamental group, and orientability  of the surface is fast: $O(\log n + \poly\log g)$. The data structure allows for changes in the embedding corresponding to the flips in Figure~\ref{fig:flip}, and also allows for contraction of edges and splits of vertices.

The construction by Eppstein~\cite{DBLP:conf/soda/Eppstein03} uses a primal-dual decomposition that is updated dynamically. Combining this idea with ideas from Italiano~et~al.~\cite{DBLP:conf/esa/ItalianoPR93} and using top-trees~\cite{DBLP:journals/talg/AlstrupHLT05}, Holm and Rotenberg~\cite{DBLP:journals/mst/HolmR17} give a data structure for maintaining a dynamic planar embedding that allows for edge-deletions, insertions across a face, and flips (Figure~\ref{fig:flip}) in worst case $\OO(\log^2 n)$ time. The data structure in \cite{DBLP:journals/mst/HolmR17} also facilitates queries to whether a pair of vertices would be linkable after only one flip operation, to which it responds with the specific flip operation in the affirmative case.

\subsection{Overview of techniques}

For the purpose of our algorithmic results, we cannot afford to maintain constructions such as the SPQR-tree explicitly; the insertion or deletion of just one edge can lead to $\Theta(n)$ changes, which is challenging to handle in worst-case time. We show, however, that we can maintain an implicit representation of the forest of BC-trees and the SPQR-trees of blocks. It turns out that this implicit representation also simplifies the presentation of the structural result stated in Theorem~\ref{thm:flips}.
In this section, we point out some of the challenges with maintaining embeddings and triconnected components of incremental graphs, and give a high level overview of how we solve them.

\subsubsection{For biconnected graphs}~\label{sec:sketchbicon}
Once a graph is triconnected, its planar embedding is fixed. On the other hand, as soon as the graph has a separation pair, there is flexibility corresponding to flipping in said pair. 
Assuming first we have a biconnected graph, 
our approach goes via keeping track of the triconnected components and the separation cuts they have in common. 

The SPQR-tree is a well-known structure that maintains information about triconnectivity in the graph. The triconnected components appear as nodes in the SPQR-tree, so called R-nodes, and the separation pairs that separate them appear as edges. However, because of the complex nature of triconnectivity, these are not the only gadgets in the tree: The fact that one separation pair may split the graph into more than two parts is reflected in the SPQR-tree by having P-nodes represent such \emph{parallel} splits, and the fact that one may find more than two vertices arranged around a cycle such that any pair of them form a separation pair, is reflected in the SPQR-tree by having S-nodes represent such \emph{series} splits.

It is well-known that the SPQR-tree maintains information about the possible embeddings of the biconnected graph. 
In fact, given a rooted SPQR-tree, any embedding can be encoded by annotating edges and vertices of the SPQR tree. The edges can be annotated  with a bit telling whether the subgraph separated from the root by said pair should be flipped or not, and the P-nodes of degree $k$ can given an annotation specifying one of the $(k-1)!$ possible different orderings of their children.
When we want to make the embedding such that it is 
nearly ready
for an edge insertion, this means that most of these annotations on the edges and nodes should be set in a way that need not change when the adversarial edge insertion happens. 
Specifically, we want to maintain the invariant that only $O(\log n)$ annotations change upon any edge-insertion, and that each change to a $P$-node annotation only corresponds to a \emph{slide} operation (see Figure~\ref{fig:flip}).

Our idea is simple: Given the rooted SPQR-tree, calculate the heavy path decomposition, and make sure that the annotations on the heavy edges and on those $P$-nodes that are internal on the heavy path, are set in a favourable way:
Assume the endpoints of an edge-to-be-inserted are represented by SPQR-nodes, and consider the path in the SPQR-tree connecting them. Then, we can afford to perform $O(1)$ flips for each light edge, since there are only $O(\log n)$ of those. We can also afford to perform $O(1)$ flips for each distinct heavy path along the way, since there are only $O(\log n)$ of those. But we can not afford to perform a number of flips proportional to the length of the heavy paths. Thus, we need to ensure that all the maximal segments of heavy paths that can be ``covered" by an edge-insertion without violating planarity, are already embedded in a way such that said insertion is possible.

However, once the edge actually is inserted, the SPQR-tree may change drastically; up to $\Theta(n)$ SPQR-nodes may disappear, and up to $\Theta(n)$ new SPQR-nodes may arise. Providing each of the new SPQR-nodes with a new heavy child without breaking the invariant could cause up to $\Theta(n)$ flips. Fortunately, the nature of the new vertices is very predictable. The changes caused by an edge insertion (see Figure~\ref{fig:SPQR}) all happen along a path: first, each S- and P-node on the path can give rise to up to two new SPQR-nodes of the same type, inheriting a subset of the original node's children, and then, the path is contracted. Again, we use the heavy path decomposition to overcome this challenge: By
 pre-splitting most S- and P-nodes that have a heavy child (see Figure~\ref{fig:pre-splitting}), we make sure that only $O(\log n)$ end vertices of heavy paths may need to be split in order for an edge to be inserted.

Note that we can not just require all such nodes to be pre-split, as in some cases that would require $\Theta(n)$ nodes to be pre-split.  In fact, finding a deterministic, history-independent rule for when to pre-split that needs only $\OO(\log n)$ flips in the worst case, was one of the main technical challenges.
By carefully choosing which nodes to pre-split, we become able to prove that the act of pre-splitting and choosing a heavy child does not cascade.  

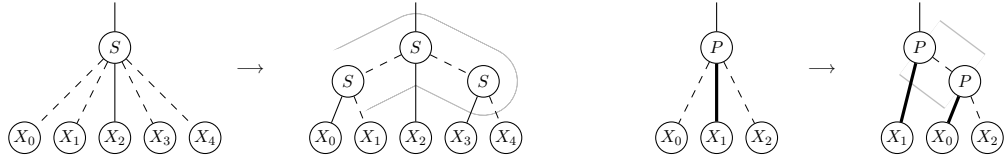
\begin{figure}[H]
\center
\begin{tikzpicture}[
    spqr-node/.style={
      draw,
      circle,
      fill=white,
      minimum size=7mm,
      inner sep=0pt,
    },
    spqr-edge/.style={
      draw,
    },
    dashed spqr-edge/.style={
      spqr-edge,
      very thin,
      dashed
    },
    solid spqr-edge/.style={
      spqr-edge,
      very thick,
    }
    every node/.style={scale=0.6}
    scale=.7
  ]
  \begin{scope}[scale=0.6, every node/.style={scale=0.6}]
    \coordinate (p) at (2,3);
    \node[spqr-node] (s0) at (2,2) {$S$};
    \node[spqr-node] (l0) at (0,0) {$X_0$};
    \node[spqr-node] (l1) at (1,0) {$X_1$};
    \node[spqr-node] (l2) at (2,0) {$X_2$};
    \node[spqr-node] (l3) at (3,0) {$X_3$};
    \node[spqr-node] (l4) at (4,0) {$X_4$};
    \draw[spqr-edge] (p) -- (s0);
    \draw[dashed spqr-edge] (s0) -- (l0);
    \draw[dashed spqr-edge] (s0) -- (l1);
    \draw[solid spqr-edge] (s0) -- (l2);
    \draw[dashed spqr-edge] (s0) -- (l3);
    \draw[dashed spqr-edge] (s0) -- (l4);
  \end{scope}
  \begin{scope}[shift={(3,0)},scale=0.6, every node/.style={scale=0.6}]
    \node at (0,1.5) {$\longrightarrow$};
  \end{scope}
  \begin{scope}[shift={(4,0)},scale=0.6, every node/.style={scale=0.6}]
    \coordinate (p) at (2,3);
    \node[spqr-node] (s0) at (2,2) {$S$};
    \node[spqr-node] (s1) at (0.5,1.25) {$S$};
    \node[spqr-node] (s2) at (3.5,1.25) {$S$};
    \node[spqr-node] (l0) at (0,0) {$X_0$};
    \node[spqr-node] (l1) at (1,0) {$X_1$};
    \node[spqr-node] (l2) at (2,0) {$X_2$};
    \node[spqr-node] (l3) at (3,0) {$X_3$};
    \node[spqr-node] (l4) at (4,0) {$X_4$};
    \draw[spqr-edge] (p) -- (s0);
    \draw[dashed spqr-edge] (s0) -- (s1);
    \draw[dashed spqr-edge] (s0) -- (s2);
    \draw[solid spqr-edge] (s1) -- (l0);
    \draw[dashed spqr-edge] (s1) -- (l1);
    \draw[solid spqr-edge] (s0) -- (l2);
    \draw[solid spqr-edge] (s2) -- (l3);
    \draw[dashed spqr-edge] (s2) -- (l4);
    \begin{pgfonlayer}{background}
      \filldraw[line width=26,line join=round,gray!50](s0.center)--(s1.center)--(s0.center)--(s2.center)--cycle;
      \filldraw[line width=25,line join=round,white](s0.center)--(s1.center)--(s0.center)--(s2.center)--cycle;
    \end{pgfonlayer}
  \end{scope}
\end{tikzpicture}
 \hfil \begin{tikzpicture}[
    spqr-node/.style={
      draw,
      circle,
      fill=white,
      minimum size=7mm,
      inner sep=0pt,
    },
    spqr-edge/.style={
      draw,
    },
    dashed spqr-edge/.style={
      spqr-edge,
      very thin,
      dashed
    },
    solid spqr-edge/.style={
      spqr-edge,
      very thick,
    }
  ]
  \begin{scope}[scale=0.6, every node/.style={scale=0.6}]
    \coordinate (p) at (1,3);
    \node[spqr-node] (p0) at (1,2) {$P$};
    \node[spqr-node] (l0) at (0,0) {$X_0$};
    \node[spqr-node] (l1) at (1,0) {$X_1$};
    \node[spqr-node] (l2) at (2,0) {$X_2$};
    \draw[spqr-edge] (p) -- (p0);
    \draw[dashed spqr-edge] (p0) -- (l0);
    \draw[solid spqr-edge] (p0) -- (l1);
    \draw[dashed spqr-edge] (p0) -- (l2);
  \end{scope}
  \begin{scope}[shift={(2,0)},scale=0.6, every node/.style={scale=0.6}]
    \node at (0,1.5) {$\longrightarrow$};
  \end{scope}
  \begin{scope}[shift={(3,0)},scale=0.6, every node/.style={scale=0.6}]
    \coordinate (p) at (.5,3);
    \node[spqr-node] (p0) at (.5,2) {$P$};
    \node[spqr-node] (p1) at (1.5,1.25) {$P$};
    \node[spqr-node] (l1) at (0,0) {$X_1$};
    \node[spqr-node] (l0) at (1,0) {$X_0$};
    \node[spqr-node] (l2) at (2,0) {$X_2$};
    \draw[spqr-edge] (p) -- (p0);
    \draw[solid spqr-edge] (p0) -- (l1);
    \draw[dashed spqr-edge] (p0) -- (p1);
    \draw[solid spqr-edge] (p1) -- (l0);
    \draw[dashed spqr-edge] (p1) -- (l2);
    \begin{pgfonlayer}{background}
      \filldraw[line width=26,line join=round,gray!50](p0.center)--(p1.center)--cycle;
      \filldraw[line width=25,line join=round,white](p0.center)--(p1.center)--cycle;
    \end{pgfonlayer}
  \end{scope}
\end{tikzpicture}
 	\caption{In order to accommodate edge insertions quickly, we pre-split the S- and P-nodes that lie internal on a heavy path. \label{fig:pre-splitting}}
\end{figure}

\subsubsection{Connected graphs}

If the graph is not biconnected, then it 
can be described via the block-cut tree (or \emph{BC-tree}), as consisting of \emph{blocks} or \emph{biconnected components}, separated by cutvertices. The basic idea is to maintain an SPQR-tree for each block. 

The main challenge is almost exactly the same as in the previous section: Insertion of an edge causes changes along a path, consisting of splits of cut-vertices along the path followed by a contraction of the path. This similar problem has a similar solution: again, we decompose the tree into heavy paths, pre-split vertices along each heavy path, and represent each heavy path in an easy-to-contract way. The easy-to-contract representation of each heavy path is easy to understand; we simply insert a \emph{strut}, an edge from the head to the tail of the heavy path, thus covering the entire path and making it artificially biconnected. We then represent the solid path with an SPQR-tree, whose root is the S-node formed by adding the strut, and where the children of the root correspond to the blocks along the path. 

In terms of planarity, this approach appears to be problematic, since the strut may violate planarity. To overcome this, we define a best-effort planar embedding, such that if the actual graph is always plane, the embedding allows the insertion of each strut across a face whenever possible, and we always know which struts are compatible with the planarity of the graph.

To maintain this heavy path decomposed BC-tree dynamically, we need to augment our structure for SPQR-trees with some extra operations, that translate into link and cut in the tree. Namely, to concatenate two heavy paths, we want to go from two SPQR-trees, each rooted in an S-node, each having children corresponding to the blocks on the path, to one SPQR-tree rooted in one S-node, whose children correspond to all the blocks on the combined path. Similarly, cleaving a heavy path corresponds to splitting an S-node into two S-nodes, each inheriting the section of the children that corresponds to their particular sub-path. 

Furthermore, because the inserted strut is only virtual, and can thus be removed or changed, we need the SPQR-trees to handle \emph{undo} of edge-insertions. Loosely speaking, we maintain a dependency forest over the inserted edges, and allow the deletion of any edge that is the root of its tree in the forest. Vice versa, inserting an edge in the graph corresponds to inserting a new root in the dependency forest, linking it to some of the roots of the existing trees.

We have now sketched a structure for maintaining embeddings of a connected incremental graph. Namely, to insert an edge, we find the corresponding path in the BC-tree. If we apply the well-known trick of forcing this path to be heavy, then this corresponds to only $O(\log n)$ operations of heavy paths being severed or melded. Each of these sever and meld operations perform $O(1)$ updates to $O(1)$ SPQR-trees. Each of these updates are handled by our structure for the SPQR-trees with $O(\log n)$ changes. Thus, we have presented a sketch of how to embed an incremental connected graph with only $O(\log ^2 n)$ flips per insertion.  

To improve this from $O(\log^2 n)$ to $O(\log n)$, we need to apply a carefully chosen weighting of the SPQR-nodes in each SPQR-tree, and redefine the heavy paths and light edges accordingly. This ensures that for each edge insertion, the total number of involved light edges the BC-tree and in the SPQR-trees is still bounded by $O(\log n)$.

\subsubsection{General graphs}

When the incremental graph is not connected, we maintain an embedding of each component. When an edge-insertion connects two components, the idea is to root each BC-tree in the endpoint of the new edge, and then link those to a root block corresponding to the new edge. 
This calls for an \emph{evert operation}, forcing a node of the BC-tree to be root.

\subsubsection{Algorithmic challenges}
In addition to the combinatorial result about embeddings, we also give several data structures:

\paragraph{Biased dynamic trees}
To support the rest of our algorithmic results, we need to extend
Sleator and Tarjan's dynamic trees~\cite{DBLP:journals/jcss/SleatorT83} (which maintain a heavy path
decomposition of an unweighted tree) to handle vertex weights. While
the extension seems obvious after the fact, we have been unable to
find a prior description.  We call the resulting data structure
\emph{biased dynamic trees}, by analogy with the Biased search trees
by Bent~et~al.~\cite{DBLP:journals/siamcomp/BentST85} (which are used
as subroutines by both versions of dynamic trees).  

\paragraph{Incremental SPQR/BC trees}
We present our biased dynamic trees as a nicely wrapped data structure
with certain operations available, but for our actual use on pre-split
BC/SPQR trees we have to open the black box and extend it to e.g. handle
different kinds of node splitting.  Also, the weights we use turn out
to be hard to maintain explicitly.  This is not an issue for the structural result, but costs an $\OO(\log^2 n)$ factor in running time.

In relation to the description above, we need to argue that the contraction of a path in the SPQR-tree can be done in efficient worst-case time. To overcome this challenge, we use a representation of the SPQR-tree where the difference between a heavy path and a contracted heavy path is concisely encoded. The encoding has to be done in a clever way that allows a heavy path to be cleaved into two heavy paths, and, vice versa, allows adjacent heavy paths to be concatenated, when the heavy path decomposition undergoes dyhnamic changes. This is technically somewhat similar to, but more involved than, the corresponding data structure for biconnectivity~\cite{DBLP:journals/siamcomp/PoutreW98}.

\paragraph{Triconnectivity} 
In order to answer triconnectivity queries, we need to augment our implicit representation of an SPQR-tree with enough information for us to use it to answer queries to whether a pair of vertices are triconnected. 

The implicit representation of the SPQR-tree contains information about triconnectivity between vertices of the graph: to query for triconnectivity of a pair of vertices, all we need is to find an R- or P-node in the SPQR-tree that contains both of them.

\paragraph{Planarity testing}
The main argument for the bounded number of changes to the embedding has been sketched in Subsection~\ref{sec:sketchbicon}. 
Those ideas lead to a data structure for maintaining an implicit representation of the SPQR-tree of a biconnected graph subject to edge-insertions, that handles each update in worst-case polylog time. Together with a well-chosen encoding of its embedding, we obtain a scheme for maintaining an embedding such that only $O(\log n)$ \emph{flips} (see Figure~\ref{fig:flip}) are necessary to accommodate each new edge. Roughly speaking, this can be used to obtain an incremental worst-case $O(\log^3 n)$ algorithm for planarity testing, by using \cite{DBLP:journals/mst/HolmR17} as a subroutine; we have argued that only $O(\log n)$ flips are necessary, but each of them can be found in $O(\log^2 n)$ time using the algorithm from \cite{DBLP:journals/mst/HolmR17}. 

Note that given the data structure from~\cite{DBLP:journals/mst/HolmR17}, we can answer in only $O(\log ^2 n)$ time whether the query edge is already compatible with the current embedding. Only when the embedding needs to change, the full $O(\log^3 n)$ time is necessary.

\paragraph{Queries to the embedding}
Along with maintaining the canonical embedding dynamically, we are able to answer queries to the embedding. Given an edge incident to a specific vertex $v$, we can output its two neighbours in the circular ordering around $v$ in $O(1)$ time. If we need to distinguish between its right and left neighbour, we can use~\cite{DBLP:journals/mst/HolmR17} to find the correct order in  $O(\log n)$ time.

\subsection{Organisation of the paper}

We begin, in Section~\ref{sec:flip-dist}, with a formal definition of the distance between a pair of embedded graphs, which allows us to formally state our result about existence of ``good" embeddings, as well as prove the matching lower bounds, that is, Theorem~\ref{thm:intro-lowerbound} and its corollary. In Section~\ref{sec:biaseddyntree}, we show how to extend the dynamic trees of Sleator and Tarjan~\cite{DBLP:journals/jcss/SleatorT83} to \emph{biased dynamic trees}, which we need in order to handle graphs that are not necessarily biconnected. Our main technical contribution lies in Section~\ref{sec:dynSPQR}, in which we show how to efficiently maintain an implicit representation of the SPQR-tree of a block, and (in Section~\ref{sec:incr-spqr}) we give a data structure for incremental triconnectivity for biconnected graphs. Correspondingly, Section~\ref{sec:dynbc} shows how to maintain the forest of BC-trees of an incremental graph, and (in Section~\ref{sec:incr-bc}) we use this to extend the incremental triconnectivity data structure to general graphs.
In Section~\ref{sec:dynEmb}, we show how to annotate the implicit SPQR-tree to describe an implicit canonical embedding for biconnected graphs, and in Section~\ref{sec:dynEmb-general} we extend this to general graphs.
Finally, Section~\ref{sec:implementation} is dedicated to some deferred implementation details of the algorithmic results.

\section{Flip-distance for labelled plane multigraphs}\label{sec:flip-dist}

Unless otherwise noted, we will be working with undirected
planar loopless multigraphs, with distinct vertex labels and distinct edge labels, both from some totally ordered, countable set.

Let $\Emb$ denote the (countably infinite) set of all plane embedded
graphs that are embeddings of such graphs, and given a graph $G$, let
$\Emb(G)\subseteq\Emb$ denote the set of all plane embeddings of $G$.
We say that two plane embedded graphs $H, H'\in\Emb$ are
\emph{adjacent} if either:
\begin{itemize}
\item There is an edge $e\in H$ such that $H-e=H'$, or an edge
  $e'\in H'$ such that $H=H'-e'$, or
\item There is a single \emph{flip} (either an
  \emph{articulation-flip} or a \emph{separation-flip}, see Figure~\ref{fig:flip}) that
  transforms $H$ into $H'$.
\end{itemize}

Define the \emph{flip-graph} $\mathcal{G}$ as the (countably infinite) graph
with $\Emb$ as its vertices, and an edge for each pair of adjacent
embeddings.  For any two plane embedded graphs $H, H'\in\Emb$ we
define the \emph{flip-distance} $\flipdist(H,H')$ to be the length of any
shortest path between $H$ and $H'$ in $\mathcal{G}$. Note that for any planar graph $G$ with $n$ vertices, the flip-distance between any two plane embeddings of $G$ is $\OO(n)$. We will extend
this notion of distance to sets $A,B\subseteq\Emb$ by defining
$\flipdist(A,B)$ to be the \emph{Hausdorff distance} between $A$ and
$B$, that is:
\begin{align*}
  \flipdist(A,B) &:= \max\set*{
    \max_{a\in A}\min_{b\in B}\flipdist(a,b)
    ,\,
    \max_{b\in B}\min_{a\in A}\flipdist(a,b)
  }
\end{align*}

Our main results about embeddings can be stated as
\begin{theorem}\label{thm:embclass}
  We define a set $\EmbGood\subseteq\Emb$ of \emph{good} embeddings,
  and define $\EmbGood(G):=\EmbGood\cap\Emb(G)$, such that for any
  planar graph $G$ with $n$ vertices, $\EmbGood(G)\neq\emptyset$ and
  for any edge $e$ in $G$,
  \begin{align*}
    \flipdist(\EmbGood(G-e),\EmbGood(G))\in\OO(\log n)
  \end{align*}
\end{theorem}
\begin{proof}The proof for biconnected graphs is deferred to Section~\ref{sec:heavySPQRembedding}, see Theorem~\ref{thm:biconn-embgood-flipdist}, and for general graphs, the proof is in Section~\ref{sec:EmbGoodBiconn}, see Theorem~\ref{thm:good-emb-bc}.
\end{proof}
\begin{theorem}
  Furthermore, we can maintain a good embedding for a dynamic planar
  graph under edge insertions and deletions in worst case $\OO(\log
  n)$ flips per operation.

  The algorithm uses linear space and $\OO(\log^3 n)$
  worst case time per edge insertion, and worst
  case linear time per edge deletion.
\end{theorem}
\begin{proof}Deferred to section~\ref{sec:incrPlanrBicon}.
\end{proof}

Thus, we can restrict ourselves to the set of \emph{good} embeddings
and never need more than $\OO(\log n)$ flips per edge insertion or
deletion.  Note that the set of good embeddings is not necessarily
small.  There are graphs where $\abs{\EmbGood(G)}$ may
be as large as $2^{\Omega(n\log n)}$ (See Figure~\ref{fig:manygood}).

\begin{theorem}[Restatement of Theorem~\ref{thm:flips}]\label{thm:canonical-emb}
  For any planar graph $G$ with $n$ vertices we can define a 
  \emph{canonical} embedding $\emb(G)\in\EmbGood(G)$ such that for
  any edge $e$ in $G$,
  \begin{align*}
    \flipdist(\emb(G-e),\emb(G))\in\OO(\log n)
  \end{align*}

\end{theorem}
\begin{proof}The proof for biconnected graphs is deferred to Section~\ref{sec:canonical}, and for general graphs, the proof is in Section~\ref{sec:canonBC}.
\end{proof}
So in fact, we can restrict ourselves to a particular \emph{canonical}
embedding of each graph and still not need more than $\OO(\log n)$
flips per edge insertion or deletion.
The choice of canonical embedding is, however, not unique.

These results should be contrasted with the following easy lower bounds.
\begin{theorem}[Lower bound 1]
  For any $n$, there exists a planar graph $G$ with $O(n)$ vertices, and an
  edge $e\in G$ such that $\flipdist(\Emb(G-e),\Emb(G))\in\Omega(n)$.
\end{theorem}
\begin{proof}
  The wheel graph $W_n$ on $n\geq5$ vertices is such an example. For
  any edge $e$ on the outer rim, $W_n-e$ has an embedding that
  requires $n-4$ flips before $e$ can be added (see
  Figure~\ref{fig:lowerbound1}, left).
  \begin{figure}[H]
	\center
	\begin{tikzpicture}[
    vertex/.style={
      draw,
      circle,
      fill=white,
      minimum size=1mm,
      inner sep=0pt,
    },
    edge/.style={
      draw,
    },
    dashed edge/.style={
      edge,
      very thin,
      dashed
    },
    solid edge/.style={
      edge,
    }
  ]
  \begin{scope}[shift={(0,0)}]
    \node[vertex] (v0) at (0,0) {};
    \node[vertex] (v1) at (.5,-.125) {};
    \node[vertex] (v2) at (.75,.125) {};
    \node[vertex] (v3) at (.875,-.5) {};
    \node[vertex] (v4) at (1,.5) {};
    \node[vertex] (v5) at (1.125,-1) {};
    \node[vertex] (v6) at (1.25,1) {};
    \foreach \i [remember=\i as \lasti] in {1,...,6} {
      \draw[solid edge] (v0) -- (v\i) {};
      \ifnum\i>1
      \draw[solid edge] (v\lasti) -- (v\i) {};
      \fi
    }
    \draw[dashed edge, red] (v1) to[bend left=10] (v6);
  \end{scope}
  \begin{scope}[shift={(2,0)}]
    \node at (0,0) {$\longrightarrow$};
  \end{scope}
  \begin{scope}[shift={(3.75,0)}]
    \node[vertex] (v0) at (0,0) {};
    \foreach \i in {1,...,6} {
      \draw[solid edge] (v0) -- (60*\i:1) node[vertex] (v\i) {};
    }
    \foreach \i [remember=\i as \lasti] in {1,...,6} {
      \ifnum\i>1
      \draw[solid edge] (v\lasti) -- (v\i) {};
      \fi
    }
    \draw[dashed edge, green] (v1) -- (v6);
  \end{scope}
\end{tikzpicture}
         \hfil
	\begin{tikzpicture}[
    vertex/.style={
      draw,
      circle,
      fill=white,
      minimum size=1mm,
      inner sep=0pt,
    },
    edge/.style={
      draw,
    },
    dashed edge/.style={
      edge,
      very thin,
      dashed
    },
    solid edge/.style={
      edge,
    }
  ]
  \begin{scope}
    \node[vertex] (ab1) at (0,0) {};
    \node[vertex] (a2) at (-.25,.25) {};
    \node[vertex] (b2) at (-.25,-.25) {};
    \node[vertex] (a3) at (.25,.25) {};
    \node[vertex] (b3) at (.25,-.25) {};
    \node[vertex] (a4) at (-.5,.5) {};
    \node[vertex] (b4) at (-.5,-.5) {};
    \node[vertex] (a5) at (.5,.5) {};
    \node[vertex] (b5) at (.5,-.5) {};
    \node[vertex] (ab6) at (1,0) {};
    \foreach \i [remember=\i as \lasti] in {2,...,5} {
      \draw[solid edge] (a\i) -- (b\i) {};
      \ifnum\i>2
      \draw[solid edge] (a\lasti) -- (a\i) {};
      \draw[solid edge] (b\lasti) -- (b\i) {};
      \fi
    }
    \draw[solid edge] (ab1) -- (a2) {};
    \draw[solid edge] (ab1) -- (b2) {};
    \draw[solid edge] (a5) -- (ab6) {};
    \draw[solid edge] (b5) -- (ab6) {};
    \draw[dashed edge, red] (ab6) -- (ab1);
  \end{scope}
  \begin{scope}[shift={(2,0)}]
    \node at (0,0) {$\longrightarrow$};
  \end{scope}
  \begin{scope}[shift={(4,0)}]
    \node[vertex] (ab1) at (60:.75) {};
    \node[vertex] (ab6) at (0:.75) {};
    \foreach \i in {2,...,5} {
      \draw[solid edge] (60*\i:.5) node[vertex] (a\i) {} -- (60*\i:1) node[vertex] (b\i) {};
    }
    \foreach \i [remember=\i as \lasti] in {2,...,5} {
      \ifnum\i>2
      \draw[solid edge] (a\lasti) -- (a\i) {};
      \draw[solid edge] (b\lasti) -- (b\i) {};
      \fi
    }
      \draw[solid edge] (ab1) -- (a2) {};
      \draw[solid edge] (ab1) -- (b2) {};
      \draw[solid edge] (a5) -- (ab6) {};
      \draw[solid edge] (b5) -- (ab6) {};
    \draw[dashed edge, green] (ab6) -- (ab1);
  \end{scope}
\end{tikzpicture}
 	\caption{\label{fig:lowerbound1} $W_7-e$ and $G_6-e$ may each
          need $3$ separation flips to admit $e$.}
  \end{figure}
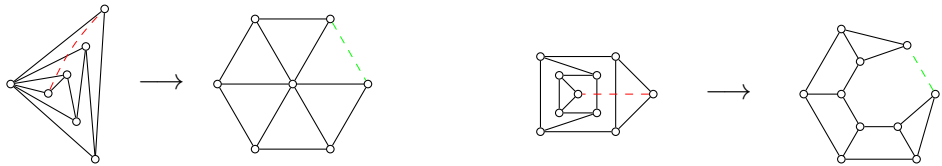

  The wheel graph has a vertex of high degree, but the result still
  holds for cubic graphs (see
  Figure~\ref{fig:lowerbound1}, right). Start with the circular ladder graph $\mathit{CL}_k$
  with $2k$ vertices, contract two consecutive rungs, remove one of the
  duplicate edges, and call the other edge $e$.  The resulting graph $G_k$
  has $n=2k-2$ vertices, and $G_k-e$ has an embedding that
  requires $k-3=n/2-2$ flips before $e$ can be added.
\end{proof}
In other words, if your current embedding is bad, adding an edge may
require $\Omega(n)$ flips.

\begin{proof}[Proof of Theorem~\ref{thm:intro-lowerbound}]
Let $h\geq 2$ be given. 
Take two identical complete binary trees of height $h$, and
  identify the corresponding leaves.  Let $L$ be the set of these
  joined leaves.  The resulting graph $G_h$ has $n=3\cdot 2^h-2$ vertices,
  and for any pair of distinct $x,y\in L$, $G_h\cup(x,y)$ is planar
  (See Figure~\ref{fig:lowerbound2}).

  \begin{sloppypar}
  However, in any plane embedding $H\in\Emb(G_h)$ there is a pair of
  leaves $x,y\in L$ such that adding $(x,y)$ to $H$ requires
  $h-1=\log_2((n+2)/3)-1$ flips.  In other words,
  $\min_{H'\in\Emb(G_h\cup(x,y))}\flipdist(H,H')=\log_2((n+2)/3)$.\qedhere
  \end{sloppypar}
\end{proof}
\vspace{-.2em}
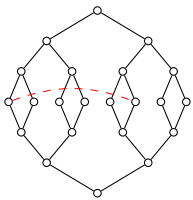
\begin{figure}[htb]
	\begin{minipage}[t]{.4\textwidth}
		\begin{figure}[H]
			\centering \begin{tikzpicture}[
    vertex/.style={
      draw,
      circle,
      fill=white,
      minimum size=1mm,
      inner sep=0pt,
    },
    edge/.style={
      draw,
    },
    dashed edge/.style={
      edge,
      very thin,
      dashed
    },
    solid edge/.style={
      edge,
    }
  ]
  \begin{scope}[scale=1.35]
    \node[vertex] (a1) at (0,.9) {};
    \node[vertex] (b1) at (0,-.9) {};
    \node[vertex] (a2) at (-.5,.6) {};
    \node[vertex] (b2) at (-.5,-.6) {};
    \node[vertex] (a3) at (.5,.6) {};
    \node[vertex] (b3) at (.5,-.6) {};
    \node[vertex] (a4) at (-.75,.3) {};
    \node[vertex] (b4) at (-.75,-.3) {};
    \node[vertex] (a5) at (-.25,.3) {};
    \node[vertex] (b5) at (-.25,-.3) {};
    \node[vertex] (a6) at (.25,.3) {};
    \node[vertex] (b6) at (.25,-.3) {};
    \node[vertex] (a7) at (.75,.3) {};
    \node[vertex] (b7) at (.75,-.3) {};
    \node[vertex] (l8) at (-.875,0) {};
    \node[vertex] (l9) at (-.625,0) {};
    \node[vertex] (l10) at (-.375,0) {};
    \node[vertex] (l11) at (-.125,0) {};
    \node[vertex] (l12) at (.125,0) {};
    \node[vertex] (l13) at (.375,0) {};
    \node[vertex] (l14) at (.625,0) {};
    \node[vertex] (l15) at (.875,0) {};

    \foreach \i in {1,...,7} {
      \pgfmathtruncatemacro{\ileft}{2*\i}
      \pgfmathtruncatemacro{\iright}{2*\i+1}
      \ifnum\i<4
      \draw[solid edge] (a\i) -- (a\ileft) {};
      \draw[solid edge] (a\i) -- (a\iright) {};
      \draw[solid edge] (b\i) -- (b\ileft) {};
      \draw[solid edge] (b\i) -- (b\iright) {};
      \else
      \draw[solid edge] (a\i) -- (l\ileft) {};
      \draw[solid edge] (a\i) -- (l\iright) {};
      \draw[solid edge] (b\i) -- (l\ileft) {};
      \draw[solid edge] (b\i) -- (l\iright) {};
      \fi
    }
    \draw[dashed edge, red] (l8) to[bend left=20] (l13);
  \end{scope}
\end{tikzpicture}
 			\caption{\label{fig:lowerbound2} $G_3$ may need $2$ separation
				flips to admit $e$.}
		\end{figure}
	\end{minipage}\hfill 
	\begin{minipage}[t]{.5\textwidth}
		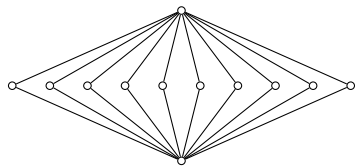
\begin{figure}[H]
			\centering
			\begin{tikzpicture}[
    vertex/.style={
      draw,
      circle,
      fill=white,
      minimum size=1mm,
      inner sep=0pt,
    },
    edge/.style={
      draw,
    },
    dashed edge/.style={
      edge,
      very thin,
      dashed
    },
    solid edge/.style={
      edge,
    }
  ]
  \begin{scope}
    \node[vertex] (a) at (0,1) {};
    \node[vertex] (b) at (0,-1) {};
    \foreach \i in {0,...,9}{
      \pgfmathsetmacro{\x}{.5*\i-2.25}
      \node[vertex] (v\i) at (\x,0) {};
      \draw[solid edge] (a) -- (v\i) -- (b);
    };
  \end{scope}
\end{tikzpicture}
 			\caption{\label{fig:manygood} The graph $K_{2,k}$ has $n=k+2$ vertices and $(k-1)!=(n-3)!\in 2^{\Omega(n\log n)}$ distinct embeddings, all of which are good.}
		\end{figure}
	\end{minipage} 
\end{figure}

In other words, there are planar graphs where for any plane embedding
there is some edge insertion that preserve the planarity of the
underlying graph but where $\Omega(\log n)$ flips are required to get
to an embedding that admits the edge.

\section{Biased dynamic trees}\label{sec:biaseddyntree}
This section is dedicated to extending existing data structures for dynamic trees such that they can handle vertex weights. We need these to exist in order to prove our combinatorial results, and we need to be able to maintain them efficiently in order to prove our algorithmic results. 

Let $T$ be a rooted tree, and suppose we mark each edge in $T$ as
either \emph{solid} or \emph{dashed} in such a way that each node has
at most one solid child edge.  This partitions the nodes of $T$ into
disjoint \emph{solid paths} (a node incident to no solid edges is a
path by itself), and we call such a collection of paths a
\emph{(naive) path decomposition} of $T$.

In their seminal paper on dynamic trees, Sleator and
Tarjan~\cite{DBLP:journals/jcss/SleatorT83} showed how to maintain a
so-called \emph{heavy path decomposition\footnote{A.k.a. heavy-light
    decomposition}} of a tree in worst
case $\OO(\log n)$ time per operation.
Here we extend their definition in the obvious way to trees where the
nodes have nonnegative integer weights (with some restrictions), and
the heavy paths are defined with respect to these weights.
We will support the following ``standard'' update operations:
\begin{description}
\item[link$(v,u)$:] Where $v$ is a root and $u$ is in some other
  tree. Link the two trees by adding the edge $(v,u)$, making $v$ a
  child of $u$.
\item[cut$(v)$:] Where $v$ is not a tree root. Delete the edge
  $(v,p(v))$ between $v$ and its parent, splitting the tree in two.
\item[evert$(v)$:] Change the root of the tree containing node $v$ to
  be $v$.
\item[reweight$(v,i)$:] Where $v$ is a node and $i$ is a nonnegative
  integer (with some restrictions on when it can be zero). Change the
  weight of $v$ to $i$.
\end{description}
Our goal is to have the operation times and the number of changes to
the path decomposition during each of these operations be
\emph{biased} with respect to the weights of the relevant nodes. Informally, we can describe the result as follows:

\begin{lemma}\label{lem:intro-biaseddyntree}
	There is a data structure for node-weighted dynamic trees that supports link$(v,u)$, cut$(v)$, evert$(v)$, and reweight$(v,w_{\textnormal{new}})$ in time proportional to $\log(W/\max\{w,1\})$, where 
	$w=\min\{w(v),w_{\textnormal{new}}\}$ is the (current or subsequent) weight of the node $v$, and $W = \max\{w(T),w(T_{\textnormal{new}})\}$ is the (current or subsequent) sum of node weights in the tree $T \ni v$.
\end{lemma}

However, in order to state our results precisely we need to dive a bit
deeper into how these dynamic trees work.

\subsection{Definitions}

Given a family of trees with nonnegative node weights, let $w(v)$ denote the
weight of node $v$, and let $w(T_v)$ denote the total weight of the
nodes in the subtree $T_v$ rooted in $v$ of the tree $T$ that contains
$v$.

\begin{definition}
  Given a child $c$ of a node $p$, we call $c$ and its parent edge
  $(c,p)$ \emph{heavy} if and only if
  $w(T_p)<2w(T_c)$, and \emph{light} otherwise.
\end{definition}
\begin{lemma}
  Each node has at most one heavy child.
\end{lemma}
\begin{proof}
  Suppose for contradiction that $h_1$ and $h_2$ are both heavy
  children of $p$, then
  \begin{align*}
    w(T_p) - w(p) 
    =
    \sum_{c\in\children(p)} w(T_c)
    \geq
    w(T_{h_1}) + w(T_{h_2})
    >
    2\cdot\frac{1}{2}w(T_p)
    = w(T_p)
  \end{align*}
  which is a contradition since $w(p)\ge 0$.
\end{proof}

A \emph{(proper) heavy path decomposition} is a path decomposition
where an edge is solid if and only if it is heavy.  It is sometimes
useful to deviate slightly from this exact decomposition, in
particular during certain updates.  One particularly useful extension
that warrants its own name is the \emph{$u$-exposed heavy path
  decomposition}, which for any node $u$ is defined as the unique path
decomposition such that an edge $e=(v,p(v))$ is solid if and only if
either
\begin{itemize}
\item $e\in r\cdots u$, or
\item $p(v)\not\in r\cdots u$ and $e$ is heavy.
\end{itemize}
Note that a proper heavy path decomposition is always a $u$-exposed
heavy path decomposition for some choice of $u$, but not vice versa.

Given any path decomposition, for each node $v$, let $c(v)$ denote its
solid child (or $\bot$ if $v$ has no solid child), and let $h(v)$
denote its heavy child (or $\bot$ if $v$ has no heavy child).

Define the \emph{light depth} $\ell(v)$ of a node $v$ as the number of light
edges on the $r\cdots v$ path.  A core property of the heavy/light
definition is the following lemma.
\begin{lemma}[Light depth]\label{lem:dyntree-lightdepth}
  In any tree $T_r$ with nonnegative node weights, a node $v$ with
  $w(T_v)>0$ has \emph{light depth}
  $\ell(v)\leq \floor*{\log_2\frac{w(T_r)}{w(T_v)}}$.
\end{lemma}
\begin{proof}
  Let $\ell=\ell(v)$, and let $l_0, l_1, \ldots, l_{\ell-1}$ be the nodes
  on the path from $v$ to the root that are light children, and let $l_\ell=r$.
  Now, for $i\in\set{0,\ldots,\ell-1}$,
  \begin{align*}
    w(T_{l_{i+1}}) \geq w(T_{p(l_i)}) \geq 2 w(T_{l_i})
  \end{align*}
  Thus,
  \begin{align*}
    w(T_r) &= w(T_{l_\ell})
    \geq 2^\ell w(T_{l_0})
    \geq 2^\ell w(T_v)
  \end{align*}
  And therefore $\ell \leq \log_2\frac{w(T_r)}{w(T_v)}$.
\end{proof}
Similarly, we can define the \emph{dashed depth} of a node $v$ as the number of dashed edges on the $r\cdots v$ path.
\begin{lemma}[Dashed depth]\todo{Do we use/need this?}
  In a $u$-exposed heavy path decomposition, the dashed depth of a
  node is at most $1$ more than the light depth.
\end{lemma}
\begin{proof}
  In a $u$-exposed heavy path decomposition, if the light depth of a
  node is $0$ then its dashed depth is at most $1$, otherwise its
  dashed depth is at most its light depth.
\end{proof}

\subsection{Standard Operations}

The main operations usually used for manipulating the path decomposition are:
\begin{description}
\item[expose$(v)$:] Where $v$ is a node in a proper heavy path
  decomposition of $T$.  Make it $v$-exposed instead.

\item[conceil$(r)$:] Where $r$ is the root of a $v$-exposed heavy path
  decomposition.  Make it proper instead.

\item[reverse$(r)$:] Where $r$ is the root of a $v$-exposed heavy path
  decomposition.  Make $v$ the root instead, and change it into an
  $r$-exposed heavy path decomposition (reverse every edge on $r\cdots
  v$).

\item[link-exposed$(v,u;\ldots)$] Where $v$ is the root of a proper
  heavy path decomposition and $u$ is a node in an $u$-exposed heavy
  path decomposition with root $r\neq v$. Combine $T_{v}$ and $T_{r}$
  into a single new tree by making $u$ the dashed parent of $v$.  The
  remaining parameters are used to initialize the new edge.  The
  resulting tree is $u$-exposed.

\item[cut-exposed$(v)$:] Where $v$ is a node in a $u$-exposed heavy
  path decomposition, where $u=p(v)$.  Split $T$ into two new trees
  $T_r$ and $T_v$ by deleting $(v,u)$.  Here $T_r$ is
  $u$-exposed, and $T_v$ is proper.

\end{description}
And to this we add
\begin{description}
\item[reweight-root$(r,i)$:] Where $r$ is the root of a $u$-exposed or
  proper tree $T_r$. Set the weight of $r$ to $i$.  If the original
  tree was proper, $r$ may gain or lose a solid (heavy) child to
  become proper again.  Otherwise the resulting tree is $u$-exposed
  (but not proper).

\end{description}
The expose and conceil operations are internally implemented in terms
of the following $2$ operations:
\begin{description}
\item[splice$(x)$:] Where $x$ is a light child of $y=p(x)$. If $(x,y)$
  is dashed, make the solid edge $(c(y),y)$ dashed and make $(x,y)$
  solid (set $c(y):=x$).

  An expose$(v)$ works by first making the solid child of $v$ (if any)
  dashed, and then calling splice on each light child on the path from
  $v$ to the root, in that order.

\item[slice$(x)$:] Where $x$ is a light child of $y=p(x)$. If $(x,y)$
  is solid, make $(x,y)$ dashed and make the dashed edge $(h(y),y)$
  solid (set $c(y):=h(y)$).

  A conceil$(r)$ works by first calling slice on each light child on
  the selected path $r\cdots v$ from $r$ to $v$ in that order. Then if
  $v$ has a heavy child, it makes the edge $(h(v),v)$ solid.
\end{description}
Note that during an expose or conceil, the path decomposition may
temporarily not be $u$-exposed for any $u$, according to our
definition.  We could ``fix'' this by reversing the order of
operations, but the number, order, and exact locations of these
operations are going to be important for the following (in particular for Lemma~\ref{lem:dyntree-internal-telescope} and~\ref{lem:dyntree-internal-optimes}):
\begin{lemma}\label{lem:dyntree-internal-changecount}
  The number of edge changes made by expose$(v)$ on a proper tree, or
  by conceil$(r)$ on a $v$-exposed tree is at most $2\ell(v)+1$.  Each
  of link-exposed$(v,r)$, cut-exposed$(v)$, and reweight-root$(v,i)$
  change at most one edge, and reverse$(r)$ changes no edges.
\end{lemma}
\begin{proof}
  Follows directly from the description of splice and slice above.
\end{proof}
\begin{lemma}\label{lem:dyntree-internal-telescope}
  If the cost of changing an edge $(v,p(v))$ where $w(T_v)>0$ from
  solid to dashed or vice versa is $\OO\paren*{ 1 +
    \log\frac{w(T_{p(v)})}{w(T_v)} }$, then the cost of an expose$(v)$
  or conceil$(r)$ where $w(T_v)>0$ is $\OO\paren*{ 1 +
    \log\frac{W}{w(T_v)} }$, where $W$ is the sum of the weights in
  the tree.
\end{lemma}
\begin{proof}
  We consider only expose$(v)$, as conceil$(r)$ is completely
  symmentric.  If $v$ is a heavy child, then
  $\OO\paren*{1+\log\frac{w(T_{p(v)})}{w(T_v)}}=\OO(1)$, so the total
  cost incurred by expose$(v)$ for changing the at most $\ell(v)+1$
  (heavy) edges from solid to dashed is $\OO(\ell(v)+1)$. All the
  (light) edges that change from dashed to solid lie on the $r\ldots
  v$ path.  In particular, let $\ell=\ell(v)$ and let
  $(v_1,p(v_1)),\ldots,(v_\ell,p(v_\ell))$ be the edges that change from
  dashed to solid, and let $v_{\ell+1}=r$, then
  \begin{align*}
    0 < w(T_v) \leq w(T_{v_1}) < w(T_{p(v_1)}) \leq \cdots \leq
    w(T_{v_\ell}) < w(T_{p(v_\ell)}) \leq w(T_{v_{\ell+1}}) = W
  \end{align*}
  so
  \begin{align*}
    0 \leq
    \sum_{i=1}^{\ell} \log\frac{w(T_{p(v_i)})}{w(T_{v_i})}
    \leq
    \sum_{i=1}^{\ell} \log\frac{w(T_{v_{i+1}})}{w(T_{v_i})}
    =
    \log\frac{w(T_{v_{\ell+1}})}{w(T_{v_1})}
    \leq
    \log\frac{W}{w(T_v)}
  \end{align*}
  And thus the cost of changing the dashed edges to solid is
  $\OO\paren*{\ell(v)+\log\frac{W}{w(T_v)}}$.  Adding the costs for
  the heavy and light edges and applying
  Lemma~\ref{lem:dyntree-lightdepth} then gives the desired result.
\end{proof}

\begin{definition}\label{def:k-positive}
  We say that a nonnegative weight function $w$ is
  \emph{$k$-positive} iff
  \begin{enumerate}
  \item Every node $v$ of degree $1$ has $w(v)>0$.
  \item Every path $u\cdots v$ consisting of at least $k$ nodes of
    degree $2$ in $T$ contains a node $c\in u\cdots v$ with $w(c)>0$.
  \end{enumerate}
\end{definition}
\begin{lemma}\label{lem:dyntree-internal-optimes}
  \sloppy We can implement biased dynamic trees with $k$-positive
  integer weights such that: expose$(v)$, conceil$(r)$,
  link-exposed$(v,u;\ldots)$, and cut-exposed$(v)$ each take worst
  case $\OO\paren*{ 1 + \log\frac{W}{\max\set{w(v),\frac{1}{k}}} }$
  time; and reverse$(r)$, and reweight-root$(r,i)$ takes worst case
  constant time.
\end{lemma}

The proof of Lemma~\ref{lem:dyntree-internal-optimes} is deferred to  Section~\ref{sec:implementation}.

\begin{lemma}[{precise formulation of Lemma~\ref{lem:intro-biaseddyntree}}]\label{lem:dyntree-external-optimes}
  \sloppy Our algorithm maintains heavy path decompositions for a
  dynamic collection of rooted trees with $k$-positive integer node
  weights, in worst case $\OO\paren*{ 1 +
    \log\frac{W}{\max\set{w(v),\frac{1}{k}}} }$ time per operation for
  link$(v,u)$ and cut$(v)$, worst case $\OO\paren*{ 1
    +\log\frac{W}{\max\set{w(v),\frac{1}{k}}} + \log\frac{W}{\max\set{w(r),\frac{1}{k}}}
  }$ time per operation for evert$(v)$, and worst case $\OO\paren*{ 1
    + \log\frac{W}{\max\set{w(v),\frac{1}{k}}} +
    \log\frac{W'}{\max\set{i,\frac{1}{k}}} }$ time per operation for
  reweight$(v,i)$, where $W$ and $W'$ are the sums of the weights in
  the trees involved before and after the operation, and $r$ is the
  root of the tree containing $v$ before the operation.
\end{lemma}
\begin{proof}
  Follows directly from the obvious implementation of these operations
  in terms of those in Lemma~\ref{lem:dyntree-internal-optimes}.
\end{proof}
What is more important to us, however, is the following Lemma which allows us to count both the number of separation flips in the SPQR trees (each such change costs $1$), 
and the sum of those changes in the BC-trees, to get a total change cost of $O(\log n)$.
\begin{lemma}\label{lem:dyntree-external-cost}
  If the cost of changing an edge $(v,p(v))$ from solid to dashed or
  vice versa is $\OO\paren*{ 1 + \log\frac{w(T_{p(v)})}{w(T_v)} }$,
  then our algorithm maintains heavy path decompositions for a dynamic
  collection of rooted trees with $k$-positive integer node weights,
  in worst case $\OO\paren*{ 1 + \log\frac{W}{w(T_v)} }$ cost per
  operation for link$(v,u)$ and cut$(v)$, worst case $\OO\paren*{ 1 +
    \log\frac{W}{w(T_v)} + \log\frac{W}{w(T'_r)} }$ cost per operation
  for evert$(v)$, and worst case $\OO\paren*{ 1 + \log\frac{W}{w(T_v)}
    + \log\frac{W'}{w(T'_v)} }$ cost per operation for
  reweight$(v,i)$, where $W$ and $W'$ are the sums of the weights in
  the trees $T$,$T'$ involved before and after the operation, and $r$
  is the root of the tree containing $v$ before the operation.
\end{lemma}
\begin{proof}
  Follows directly from the obvious implementation of these operations
  in terms of those in Lemma~\ref{lem:dyntree-internal-optimes},
  together with the cost of these operations given by
  Lemma~\ref{lem:dyntree-internal-telescope}.
\end{proof}
Note here there is no dependency on $k$, and that the denominators in
each case have changed from the weights of single nodes to the weights
of whole subtrees.

\subsection{Heavy paths in a weighted tree decomposition}\label{sec:heavyweighted}

We want to use our algorithm for biased dynamic trees on BC and SPQR
trees where the underlying graph have positive vertex weights.  These
are both special cases of \emph{tree decompositions}\cite{Halin1976,ROBERTSON198449}. A
tree decomposition for a connected graph $G$, is a pair $(T,\bag)$
where $T$ is a tree and each node $u\in V[T]$ is associated with a
\emph{bag} $\bag(u)\subseteq V[G]$, such that:
\begin{enumerate}[label={T\arabic*.}, ref={T\arabic*}]
\item\label{it:treedecomp-allvertices} $\bigcup_{v\in V[T]}\bag(v)=V[G]$.
\item\label{it:treedecomp-connected} For every path $u\cdots v$ in
  $T$, if $c\in V[u\cdots v]$, then
  $\bag(u)\cap\bag(v)\subseteq\bag(c)$.
\item\label{it:treedecomp-alledges} For every $(x,y)\in E[G]$ there exists $u\in V[T]$ such that $\set{x,y}\subseteq\bag(u)$.
\end{enumerate}
For an edge $(u,v)$ we use the notation $\bag((u,v)) = \bag(u)\cap \bag(v)$.

In principle, the tree in a tree decomposition can be arbitrarily
large, because an edge can be subdivided any number of times.
The traditional definitions of BC trees and SPQR trees 
can be described as the unique minimal tree
decompositions with certain properties, which overcomes
this. Our ``relaxed" pre-split trees will not necessarily be minimal, but will
have the property that there is a constant $a_{\max}$ (called the
\emph{adhesion}), such that for every edge $(u,v)$:
\begin{enumerate}[label={R\arabic*.}, ref={R\arabic*}]
\item\label{it:treedecomp-restrict-adh}
  $1\leq\abs{\bag(u)\cap\bag(v)}\leq a_{\max}$. \hfill(bounded adhesion)
\item\label{it:treedecomp-restrict-leaf} If $u$ is a leaf, then
  $\bag(u)\setminus\bag(v)\neq\emptyset$. \hfill(bounded leaves)
\item\label{it:treedecomp-restrict-path} If both $u$ and $v$ have degree $2$
  then $\bag(u)\neq\bag(v)$. \hfill(bounded paths)
\end{enumerate}
For each vertex $x\in V[G]$ let $b(x)$ be the node $u$ closest to the
current root such that $x\in\bag(u)$.  Then for each tree node $v$ we can define
\begin{align*}
  b^{-1}(v) &:= \set{x\in \bag(v)\suchthat b(x)=v}
  \\
  w(v) &:= \sum_{x\in b^{-1}(v)}w(x)
\end{align*}

\begin{lemma}\label{lem:treedecomp-leafweight}
  For every leaf $u$, $w(u)>0$.
\end{lemma}
\begin{proof}
  If $u$ is a leaf in $T$ with neighbor $v$, then by
  property~\ref{it:treedecomp-restrict-leaf} we have
  $\emptyset\neq\bag(u)\setminus\bag(v)\subseteq b^{-1}(u)$ and thus
  $w(u)>0$.
\end{proof}

\begin{lemma}\label{lem:treedecomp-pathweight}
  Let $a_{\max}$ be the adhesion of the tree decomposition, then any
  path consisting of $2a_{\max}+2$ nodes of degree $2$ in $T$ contains
  a node $v$ with $w(v)>0$.
\end{lemma}
\begin{proof}
  Let $k=2a_{\max}+2$, and consider any path $u_1\cdots u_k$ of $k$
  nodes of degree $2$ in $T$.  Assume for contradiction that
  $w(u_1)=\cdots=w(u_k)=0$.  Let $u_0$ and $u_{k+1}$ be the neighbors of
  $u_0$ and $u_k$ that are not on $u_1\cdots u_k$.
  Since $w(u_i)=0$ for all $i\in\set{1,\ldots,k}$, we have $\bag(u_i)
  \subseteq (\bag(u_0)\cap\bag(u_1)) \cup
  (\bag(u_k)\cap\bag(u_{k+1}))$.
  For any $i\in\set{1,\ldots,k-1}$ consider the edge $(u_i,u_{i+1})$.
  Since (by property~\ref{it:treedecomp-connected}) the set of nodes
  whose bags contain a given vertex is connected,
  $\abs{\bag(u_0)\cap\bag(u_i)} \geq \abs{\bag(u_0)\cap\bag(u_{i+1})}$
  and $\abs{\bag(u_i)\cap\bag(u_{k+1})} \leq
  \abs{\bag(u_{i+1})\cap\bag(u_{k+1})}$.  And by
  property~\ref{it:treedecomp-restrict-path}, we have
  $\bag(u_i)\neq\bag(u_{i+1})$, so at least one of those inequalities
  must be strict.  But there are $k-1=2a_{\max}+1$ such edges and each
  inequality can be strict for at most $a_{\max}$ of them, so we have
  our contradiction, and have shown that for at least one $u\in
  u_1\cdots u_k$ we must have $w(u)>0$.
\end{proof}

Thus, this system of weights is $(2a_{\max}\!\!+\!\!2)$-positive, and
we can use them to define our biased dynamic trees.  Furthermore, we
have the following useful property
\begin{lemma}\label{lem:treedecomp-lightdepth}
  For any vertex $x\in V[G]$, and any node $v\in V[T]$ that is an
  ancestor to $b(x)$, the light depth of $v$ is at most
  $\floor*{\log_2\frac{w(T_r)}{w(x)}}$.
\end{lemma}
\begin{proof}
  By Lemma~\ref{lem:dyntree-lightdepth} the light depth is at most
  $\floor*{\log_2\frac{w(T_r)}{w(T_v)}}$, and since $b(x)\in
  V[T_v]$ we have $w(x)\leq w(T_v)$ and thus
  $\log_2\frac{w(T_r)}{w(T_v)} \leq
  \log_2\frac{w(T_r)}{w(x)}$.
\end{proof}

Combining Lemma~\ref{lem:treedecomp-lightdepth} with
Lemma~\ref{lem:dyntree-internal-changecount}, we get the following Lemma that counts the number of heavy/light changes in the SPQR tree based on the weights from the BC tree:
\begin{lemma}\label{lem:treedecomp-internal-changecount}
  For any vertex $x\in V[G]$, and any node $v\in V[T]$ that is an
  ancestor to $b(x)$, the number of edge changes made to the heavy
  path decomposition by expose$(v)$ on a proper tree, or by
  conceil$(r)$ on a $v$-exposed tree is at most
  $1+2\floor*{\log_2\frac{w(T_r)}{w(x)}}$.  Each of
  link-exposed$(v,r)$, cut-exposed$(v)$, and reweight-root$(v,i)$
  change at most one edge, and reverse$(r)$ changes no edges.
\end{lemma}
\begin{proof}
  By Lemma~\ref{lem:dyntree-internal-changecount} the number of
  changes is at most $2\ell(v)+1$, and by
  Lemma~\ref{lem:treedecomp-lightdepth} this is the desired result.
\end{proof}

In fact, by combining with Lemma~\ref{lem:dyntree-internal-telescope} we get the more general
\begin{observation}\label{obs:treedecomp-internal-telescope}
  If the cost of changing an edge $(v,p(v))$ from solid to dashed or
  vice versa is $\OO\paren*{ 1 + \log\frac{w(T_{p(v)})}{w(T_v)} }$,
  then for any vertex $x\in V[G]$ and any node $v\in V[T]$ that is an
  ancestor to $b(x)$, the cost of an expose$(v)$ or conceil$(r)$ is
  $\OO\paren*{ 1 + \log\frac{W}{w(x)} }$, where $W$ is the sum of the
  weights in the tree.
\end{observation}\todo{We probably need to refer back to this later...}
The observation above is useful, because it lets us sum the changes from the SPQR tree as we work in the BC tree. 

Furthermore, we have the following useful lemma, that will allow us to find critical paths in SPQR-trees:
\begin{lemma}\label{lem:subcritic}
Given a node $v$ and a vertex $y\in \bag(v)$, with $v$ exposed, we can in constant time find the node $v'$ closest to the root such that $y\in\bag(v')$, and, symmetrically, 
given a vertex $x\in\bag(r)$, find the node $r'$ furthest from the root such that $x\in\bag(r')$.
\end{lemma}
\begin{proof}
Because of bounded adhesion, the solid path $r\cdots v$ only has a constant number of vertices that are shared between $\bag(v)$ and other nodes on the path. Thus, it can for each such vertex store its root-nearest occurrence. The second case is symmetric. 
\end{proof}

Our definition of $b(x)$ and $w(v)$ presents a problem in connection
with the reverse$(r)$ operation. In particular, we can't store $b(x)$
or $w(v)$ explicitly, since too many of these values may change.
Instead, for each vertex $x$, we store an \emph{arbitrary} node
$\hat{b}(x)$ such that $x\in\bag(\hat{b}(x))$.  And we explicitly maintain a weight $\hat{w}(v) = \sum_{x\in\bag(v), \hat{b}(x)=v}w(x)$ for each node $v$. The relationship between $w(v)$ and $\hat{w}(v)$ is that
\begin{align*}
  w(v) =
  \begin{cases}
    \hat{w}(v) & \text{if $v$ is the root}
    \\
    \hat{w}(v) - \sum_{\substack{
        x\in\bag((v,p(v))),\\
        \hat{b}(x)\in T_v
    }} w(x) & \text{otherwise}
  \end{cases}
\end{align*}
Thus, we can compute $w(v)$ from $\hat{w}(v)$ in the time that it
takes to determine if $v$ is an ancestor to $\hat{b}(x)$ for each of
the at most $a_{\max}$ values of $x\in\bag((v,p(v)))$.  

\begin{lemma}\label{lem:treedecomp-external-optimes}
  \sloppy Our algorithm maintains heavy path decompositions for a
  dynamic collection of weighted and rooted tree-decompositions of
  bounded adhesion $a_{\max}$, in worst case $\OO\paren*{ \paren*{1 +
      \log\frac{W}{\max\set{w(v),\frac{1}{k}}} }\log^2n }$ time per
  operation for link$(v,u)$ and cut$(v)$, worst case $\OO\paren*{
    \paren*{1 +\log\frac{W}{\max\set{w(v),\frac{1}{k}}} +
      \log\frac{W}{\max\set{w(r),\frac{1}{k}}} }\log^2n }$ time per
  operation for evert$(v)$, and worst case $\OO\paren*{ \paren*{1 +
      \log\frac{W}{\max\set{w(v),\frac{1}{k}}} +
      \log\frac{W'}{\max\set{i,\frac{1}{k}}} }\log^2n }$ time per
  operation for reweight$(v,i)$, where $k=2a_{\max}+2$, $W$ and $W'$
  are the sums of the weights in the trees involved before and after
  the operation, and $r$ is the root of the tree containing $v$ before
  the operation.
\end{lemma}
\begin{proof}[Proof of Lemma~\ref{lem:treedecomp-external-optimes}]
  The time for determining if $v$ is an ancestor to $\hat{b}(x)$ is
  $\OO(\log n)$, so we can compute $w(v)$ in $\OO(\log n)$ time
  whenever we need it.  This 
  costs a
  factor of
  $\OO(\log^2n)$ in the running time for maintaining such a tree
  decomposition.
  Since $\hat{b}(x)$ and $\hat{w}(v)$ do not depend on the
  choice of root or the solid paths, we can trivially maintain them
  during the expose, conceil, and reverse operations.  The operations
  link-exposed$(v,u)$ and cut-exposed$(v)$ do require that we update
  $\hat{b}(x)$ for each of the at most $a_{\max}$ vertices
  $x\in\bag((u,v))$, and this again requires changing the at most
  $\OO(a_{\max}\log n)$ places in the data structure where $w(x)$
  contributes to the stored weight $\hat{w}(\cdot)$, but by definition
  this does not change any edges from heavy to light or vice versa, and
  thus does not change the heavy path decomposition.
\end{proof}
 
\section{Dynamic SPQR trees}\label{sec:dynSPQR}

We now proceed to give a data structure for maintaining an implicit representation of an SPQR-tree of a weighted biconnected graph subject to edge insertions and undo edge insertion.

When an edge is inserted, all changes to the SPQR-tree lie on a path. In the most complicated case, this path is nontrivial, and everything along it becomes triconnected, and we have to update the SPQR-tree by substituting the path with a node representing its contraction. Thus, to handle such an edge-insertion, we need to find that SPQR-path and contract it. 

Our approach goes via maintaining a path decomposition of the SPQR-tree, and regarding all solid paths as \emph{precontracted} in a way that admits undo. Intuitively, the precontractions in question approximate insertions of imaginary edges. 
However, simply contracting a path in the SPQR-tree may differ in more than a constant number of places from what could be the SPQR-tree of any graph. 
Thus, we introduce the notion of \emph{relaxed SPQR trees}, where internal nodes on heavy paths are \emph{pre-split} in the way they would be if an edge would cover them.
Once all solid paths are pre-contracted, the actual insertion of an edge covering the solid path, or its undo, can be implemented by toggling a bit for the heavy path.
We show that we can maintain a path decomposition with pre-contracted solid paths with only $O(\log n)$ changes per insertion or undo, where here changes include pre-splitting and its reverse operation. 

Now, if the insertion of an edge in the biconnected graph causes a path in the SPQR-tree to be contracted, we need to find the endpoints of that path. Recall that either endpoint of the inserted edge may reside in several SPQR-nodes. Nonetheless, if there is no node containing both of them, there is a unique shortest path connecting them, their \emph{critical path}, and we give an algorithm for finding this path. 
In the other cases, where there exist SPQR-nodes that contain both endpoints of the inserted edge, we give an algorithm for finding the at most $3$ such nodes.

We present our construction in the following order: 
In Section~\ref{sec:defSPQR}, we recall the standard definition of SPQR-tree, and present our notion of a \emph{relaxed} SPQR-tree that allows for pre-split nodes. 
In Section~\ref{sec:dynspqr-changes}, we describe how the SPQR-tree may change reflecting the insertion of an edge. In Section~\ref{sec:heavyweighted}, we show how to maintain path decomposition in a weighted SPQR-tree; to state our findings in their fullest generality, we generalise SPQR-trees to tree-decompositions of bounded adhesion, and show how to maintain path decompositions of those. 
In Section~\ref{sec:dynspqr-presplit}, we give the details of pre-splitting of SPQR-nodes. 
In Section~\ref{sec:critical}, we show how to find that path in the SPQR-tree that needs to be modified in case an edge is inserted. 
Finally, in Section~\ref{sec:dynspqr-precontract}, we account for the pre-contracting of solid paths in the decomposition.

\subsection{Strict and relaxed SPQR trees}\label{sec:defSPQR}
\begin{definition}[Hopcroft and Tarjan~{\cite[p. 6]{DBLP:journals/siamcomp/HopcroftT73}}]\label{def:separationpair}
        Let $\set{a,b}$ be a pair of vertices in a biconnected multigraph $G$. Suppose the edges of $G$ are divided into equivalence classes $E_1,E_2,\ldots,E_k$, such that two edges which lie on a common path not containing any vertex of $\set{a,b}$ except as an end-point are in the same class. The classes $E_i$ are called the \emph{separation classes} of $G$ with respect to $\set{a,b}$. If there are at least two separation classes, then $\set{a,b}$ is a \emph{separation pair} of $G$ unless (i) there are exactly two separation classes, and one class consists of a single edge\footnote{So in a triconnected graph, the endpoints of an edge do not constitute a separation pair.}, or (ii) there are exactly three classes, each consisting of a single edge\footnote{So the graph consisting of two vertices connected by $3$ parallel edges is triconnected.}
\end{definition}

\begin{definition}[\cite{DBLP:conf/esa/HolmIKLR18}]\label{def:SPQR}
	The \emph{(strict) SPQR-tree} for a biconnected multigraph $G=(V,E)$ with at least $3$ edges is a tree with nodes labeled S, P, or R, where each node $x$ has an associated \emph{skeleton graph} $\Gamma(x)$ with the following properties:
	\begin{itemize}
		\item For every node $x$ in the SPQR-tree, $V(\Gamma(x))\subseteq V$.
		\item For every edge $e\in E$ there is a unique node $x=b(e)$ in the SPQR-tree such that $e\in E(\Gamma(x))$.
		\item For every edge $(x,y)$ in the SPQR-tree, $V(\Gamma(x))\cap V(\Gamma(y))$ is a separation pair $\set{a,b}$ in $G$, and~there is a \emph{virtual edge} $ab$ in each of $\Gamma(x)$ and $\Gamma(y)$ that corresponds to $(x,y)$.
		\item For every node $x$ in the SPQR-tree, every edge in $\Gamma(x)$ is either in $E$ or a virtual edge.
		\item If $x$ is an S-node, $\Gamma(x)$ is a simple cycle with at least
		$3$ edges.
		\item If $x$ is a P-node, $\Gamma(x)$ consists of a pair of vertices with at least $3$ parallel edges.
		\item If $x$ is an R-node, $\Gamma(x)$ is a simple 		triconnected graph.
		\item No two S-nodes are neighbors, and no two P-nodes are neighbors.
	\end{itemize}
\end{definition}
The SPQR-tree for a biconnected graph is unique (see e.g.~\cite{DBLP:journals/algorithmica/BattistaT96}). 
The (skeleton graphs associated with) the SPQR-nodes are sometimes referred to as $G$'s triconnected components. 

In this paper, we use the term \emph{relaxed SPQR tree} to denote a tree that satisfies all but the last of the condition. Unlike the strict SPQR-tree, the relaxed SPQR-tree is never unique. 

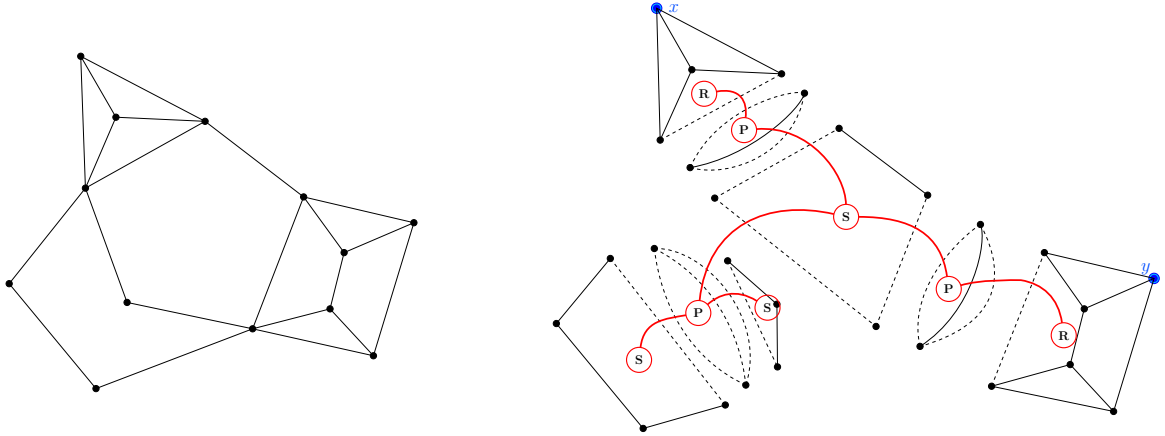
\begin{figure}
	\begin{adjustbox}{max width=\textwidth}
		\begin{tikzpicture}[y=0.80pt, x=0.80pt, yscale=-1.000000, xscale=1.000000, inner sep=0pt, outer sep=0pt,
  edge/.style={
    draw=black,
    line join=miter,
    line cap=butt,
    even odd rule,
    line width=0.800pt,
  },
  vertex/.style={
    fill=black,
    line join=miter,
    line cap=butt,
    line width=0.800pt,
    draw,
    circle,
    minimum size=3mm,
  },
]
  \clip (0,0) rectangle (297mm, 210mm);
  \begin{scope}[
      shift={(0,-308.26772)},
    ]
    \coordinate (p9) at (474,506);
    \coordinate (p10) at (642,635);
    \coordinate (p11) at (555,860);
    \coordinate (p12) at (341,815);
    \coordinate (p13) at (270,620);
    \coordinate (p14) at (262,395);
    \coordinate (p15) at (322,499);
    \coordinate (p16) at (140,783);
    \coordinate (p17) at (288,962);
    \coordinate (p22) at (830,679);
    \coordinate (p23) at (761,906);
    \coordinate (p24) at (687,826);
    \coordinate (p26) at (711,730);

    \path[edge] (p9) -- (p10) -- (p11) -- (p12) -- (p13) -- cycle;
    \path[edge] (p14) -- (p9) -- (p15) -- (p13) -- cycle;
    \path[edge] (p15) -- (p14);
    \path[edge] (p13) -- (p16) -- (p17) -- (p11);
    \path[edge] (p10) -- (p22) -- (p23) -- (p11) -- (p24) -- (p23);
    \path[edge] (p24) -- (p26) -- (p22);
    \path[edge] (p10) -- (p26);

    \node[vertex] at (p9) {};
    \node[vertex] at (p10) {};
    \node[vertex] at (p11) {};
    \node[vertex] at (p12) {};
    \node[vertex] at (p13) {};
    \node[vertex] at (p14) {};
    \node[vertex] at (p15) {};
    \node[vertex] at (p16) {};
    \node[vertex] at (p17) {};
    \node[vertex] at (p22) {};
    \node[vertex] at (p23) {};
    \node[vertex] at (p24) {};
    \node[vertex] at (p26) {};
  \end{scope}

\end{tikzpicture}

 		\begin{tikzpicture}[y=0.80pt, x=0.80pt, yscale=-1.000000, xscale=1.000000, inner sep=0pt, outer sep=0pt,
  graph edge/.style={
    draw=black,
    line join=miter,
    line cap=butt,
    even odd rule,
    line width=0.800pt,
  },
  virtual edge/.style={
    draw=black,
    dash pattern=on 4.80pt off 4.80pt,
    line join=miter,
    line cap=butt,
    miter limit=4.00,
    even odd rule,
    line width=0.800pt,
  },
  tree edge/.style={
    draw=red,
    line join=miter,
    line cap=butt,
    miter limit=4.00,
    even odd rule,
    line width=2.400pt,
  },
  font={\LARGE\bf},
]
  \clip (0,-6) rectangle (297mm, 210mm);
  \begin{scope}[
      shift={(0,-308.26772)},
      tree node/.style={
        draw=red,
        line cap=round,
        miter limit=4.00,
        line width=1.545pt,
        circle,
        minimum size=12mm,
      },
      vertex/.style={
        draw=black,
        fill=black,
        minimum size=3mm,
        circle,
      }
    ]
    \coordinate (p1) at (372,824);
    \coordinate (p9) at (183,313);
    \coordinate (p10) at (396,425);
    \coordinate (p11) at (243,418);
    \coordinate (p12) at (189,538);
    \coordinate (p13) at (104,740);
    \coordinate (p14) at (12,851);
    \coordinate (p15) at (160,1030);
    \coordinate (p16) at (300,990);
    \coordinate (p21) at (844,730);
    \coordinate (p22) at (1031,774);
    \coordinate (p23) at (962,1001);
    \coordinate (p24) at (754,958);
    \coordinate (p25) at (888,921);
    \coordinate (p27) at (912,826);
    \coordinate (p33) at (720,690);
    \coordinate (p34) at (632,890);
    \coordinate (p35) at (724,831);
    \coordinate (p36) at (747,710);
    \coordinate (p37) at (735,682);
    \coordinate (p38) at (607,780);
    \coordinate (p39) at (557,856);
    \coordinate (p40) at (645,632);
    \coordinate (p41) at (494,518);
    \coordinate (p42) at (282,637);
    \coordinate (p43) at (240,585);
    \coordinate (p44) at (343,567);
    \coordinate (p45) at (429,485);
    \coordinate (p46) at (435,458);
    \coordinate (p47) at (287,479);
    \coordinate (p48) at (419,450);
    \coordinate (p57) at (304,744);
    \coordinate (p58) at (388,818);
    \coordinate (p59) at (389,925);
    \coordinate (p63) at (179,723);
    \coordinate (p64) at (275,772);
    \coordinate (p65) at (325,862);
    \coordinate (p66) at (335,956);
    \coordinate (p67) at (242,913);
    \coordinate (p68) at (191,823);
    \coordinate (p72) at (786,777);
    \coordinate (p73) at (758,855);
    \coordinate (p74) at (309,604);
    \coordinate (p75) at (419,560);
    \coordinate (p77) at (618,669);
    \coordinate (p78) at (668,699);
    \coordinate (p81) at (444,512);
    \coordinate (p82) at (505,596);
    \coordinate (p85) at (326,448);
    \coordinate (p86) at (338,472);
    \coordinate (p89) at (734,776);
    \coordinate (p90) at (804,777);
    \coordinate (p91) at (841,783);
    \coordinate (p92) at (868,807);
    \coordinate (p95) at (270,681);
    \coordinate (p96) at (374,639);
    \coordinate (p99) at (198,843);
    \coordinate (p100) at (161,845);
    \coordinate (p103) at (288,802);
    \coordinate (p104) at (318,790);
    \coordinate (p106) at (153,913);
    \coordinate (p107) at (254,833);
    \coordinate (p108) at (506,669);
    \coordinate (p109) at (681,792);
    \coordinate (p110) at (877,871);
    \coordinate (p111) at (332,521);
    \coordinate (p112) at (264,459);
    \coordinate (p113) at (265,720);
    \coordinate (p114) at (372,842);

    \node[vertex,blue,minimum size=5mm,draw,fill=DeepSkyBlue,label=right:{\Huge\textcolor{blue!60!DeepSkyBlue}{~$x$}}] at (p9) {};
    \node[vertex,blue] at (p9) {};
    \node[vertex] at (p10) {};
    \node[vertex] at (p11) {};
    \node[vertex] at (p12) {};
    \node[vertex] at (p13) {};
    \node[vertex] at (p14) {};
    \node[vertex] at (p15) {};
    \node[vertex] at (p16) {};
    \node[vertex] at (p21) {};
    \node[vertex,blue,minimum size=5mm,draw,fill=DeepSkyBlue,label=above left:{\Huge\textcolor{blue!60!DeepSkyBlue}{~$y$}}] at (p22) {};
    \node[vertex,blue] at (p22) {};
    \node[vertex] at (p23) {};
    \node[vertex] at (p24) {};
    \node[vertex] at (p25) {};
    \node[vertex] at (p27) {};
    \node[vertex] at (p34) {};
    \node[vertex] at (p37) {};
    \node[vertex] at (p39) {};
    \node[vertex] at (p40) {};
    \node[vertex] at (p41) {};
    \node[vertex] at (p42) {};
    \node[vertex] at (p43) {};
    \node[vertex] at (p46) {};
    \node[vertex] at (p57) {};
    \node[vertex] at (p58) {};
    \node[vertex] at (p59) {};
    \node[vertex] at (p63) {};
    \node[vertex] at (p66) {};

    \path[graph edge] (p9) -- (p10) -- (p11) -- (p12) -- cycle;
    \path[graph edge] (p11) -- (p9);
    \path[graph edge] (p13) -- (p14) -- (p15) -- (p16);
    \path[graph edge] (p21) -- (p22) -- (p23) -- (p24) -- (p25) -- (p23);
    \path[graph edge] (p25) -- (p27) -- (p22);
    \path[graph edge] (p21) -- (p27);
    \path[virtual edge] (p12) -- (p10);
    \path[virtual edge] (p24) -- (p21);
    \path[graph edge] (p34) .. controls (p35) and (p36) .. (p37);
    \path[virtual edge] (p34) .. controls (p38) and (p33) .. (p37);
    \path[virtual edge] (p39) -- (p40);
    \path[graph edge] (p41) -- (p40);
    \path[virtual edge] (p42) -- (p41);
    \path[graph edge] (p43) .. controls (p44) and (p45) .. (p46);
    \path[virtual edge] (p43) .. controls (p47) and (p48) .. (p46);
    \path[virtual edge] (p42) -- (p39);
    \path[graph edge] (p57) -- (p58) -- (p59);
    \path[virtual edge] (p57) -- (p59);
    \path[virtual edge] (p63) .. controls (p64) and (p65) .. (p66);
    \path[virtual edge] (p66) .. controls (p67) and (p68) .. (p63);
    \path[virtual edge] (p13) -- (p16);
    \path[virtual edge] (p37) .. controls (p72) and (p73) .. (p34);
    \path[virtual edge] (p43) .. controls (p74) and (p75) .. (p46);
    \path[virtual edge] (p63) .. controls (p113) and (p114) .. (p66);

    \node[tree node] (v106) at (p106) {S};
    \node[tree node] (v107) at (p107) {P};
    \node[tree node] (v1) at (p1) {S};
    \node[tree node] (v108) at (p108) {S};
    \node[tree node] (v109) at (p109) {P};
    \node[tree node] (v110) at (p110) {R};
    \node[tree node] (v111) at (p111) {P};
    \node[tree node] (v112) at (p112) {R};

    \path[tree edge] (v108) .. controls (p77) and (p78) .. (v109);
    \path[tree edge] (v111) .. controls (p81) and (p82) .. (v108);
    \path[tree edge] (v112) .. controls (p85) and (p86) .. (v111);
    \path[tree edge] (v109) .. controls (p89) and (p72) .. (p90)
                            .. controls (p91) and (p92) .. (v110);
    \path[tree edge] (v107) .. controls (p95) and (p96) .. (v108);
    \path[tree edge] (v107) .. controls (p99) and (p100) .. (v106);
    \path[tree edge] (v107) .. controls (p103) and (p104) .. (v1);

\end{scope}

\end{tikzpicture}

 	\end{adjustbox}
	\caption{A biconnected graph and its SPQR-tree. 		Note that adding the edge \textcolor{DeepSkyBlue!50!blue}{$xy$} would collapse a path of SPQR-nodes into one. Deletion can thus result in the opposite transformation.~\cite{DBLP:conf/esa/HolmIKLR18}}
	\label{fig:SPQR}
\end{figure}

\subsection{Changes during insert}\label{sec:dynspqr-changes}
Given a (strict) SPQR tree, we would like to be able to maintain it as
the underlying graph changes.  When adding edge $(x,y)$, the change to
the SPQR tree happens along the \emph{critical path}, denoted $m(x,y)$
defined as:
\begin{enumerate}
\item\label{it:m-def-path} If no SPQR node contains both $x$ and $y$,
  then $m(x,y)$ is the unique shortest path $u\cdots v$ in the SPQR
  tree such that $u$ contains $x$ and $v$ contains $y$.\footnote{Note that such a path exists and is unique: Take any SPQR-node containing $x$ and any SPQR-node containing $y$, and consider the tree path $p$ connecting them. It can be subdivided into three segments: the first segment contains $x$ in all SPQR-nodes and separation pairs, the last segment contains $y$ in all SPQR-nodes and separation pairs, and the middle segment contains neither $x$ nor $y$. Extending the middle segment by its neighbours on $p$ yields $m(x,y)$. Since the middle segment contains a compact representation of all separation pairs that separate $x$ from $y$, $m(x,y)$ is unique.} \item\label{it:m-def-single} If exactly one SPQR node $v$ contains
  both $x$ and $y$, then $m(x,y)=v$.  \item\label{it:m-def-pair} If exactly two SPQR nodes $u,v$ contain
  both $x$ and $y$, then they are adjacent and $m(x,y)=u\cdots
  v$. \item\label{it:m-def-P} If more than two SPQR nodes contain both $x$
  and $y$, then $m(x,y)$ is the unique $P$ node containing both $x$
  and $y$.
\end{enumerate}
Note that the path is considered to be undirected, so
$m(x,y)=m(y,x)$.

This notion of a critical path is implicit in \cite{DBLP:conf/focs/BattistaT89}, where $m(x,y)$ is either the path connecting their representatives $\mu (x)$ and $\mu (y)$ of $x$ and $y$, or, in the case where one of them (say, $z$) is a descendant of the representative of the other (say, $z'$), the path connecting $\mu(z)$ to $\mu$, where $\mu$ is the lowest SPQR-node containing $x$ and $y$ in itself or its descendants.

The actual change to the SPQR-tree caused by inserting $(x,y)$ can be
described in terms of $m(x,y)$, as follows:
\begin{enumerate}
\item If $m(x,y)$ has at least two edges, or has a single edge that
  does not correspond to the separation pair $\set{x,y}$
  (Case~\ref{it:m-def-path} above):
  \begin{itemize}
  \item Each $S$ node     $s$ that is at the end of the path is split into
    at most $3$ pieces wrt. the vertex in $\Gamma(s)\cap\set{x,y}$ and
    the closest edge on the path.  \item Each $S$ node that is internal to $m(x,y)$ is split into at
    most $3$ pieces wrt. its neighboring edges on the path.  \item Each $P$ node that is internal to $m(x,y)$ is split into at
    most $2$ pieces wrt. its neighboring edges on the path.  \item And then the remainder of $m(x,y)$ is contracted into a single $R$
    node.  \end{itemize}
\item If $m(x,y)$ is a single $S$ node $s$ (part of
  Case~\ref{it:m-def-single} above): The node $s$ is split into two
  $S$ nodes $s_1,s_2$, with the edge $(s_1,s_2)$ corresponding to the
  separation pair $\set{x,y}$ and we proceed as if the path had a
  single edge.
\item If $m(x,y)$ is a single edge $(v_1,v_2)$ that corresponds to the
  separation pair $\set{x,y}$ (Case~\ref{it:m-def-pair}): The edge is
  subdivided by adding a new $P$ node $p$ with
  $V(\Gamma(p))=\set{x,y}$, and we proceed as if the path consisted of
  the single node $p$.
\item If $m(x,y)$ is a single $R$ node $v$ and $\Gamma(v)$ already
  contains an edge $e=(x,y)$ (part of Case~\ref{it:m-def-single}): A
  new leaf $P$ node $v'$ is added as neighbor of $v$, turning $e$ into
  a virtual edge. Then the new $(x,y)$ edge is added to $\Gamma(v')$,
  and we proceed as if the path consisted only of $v'$.
\item Finally, if $m(x,y)$ is a single $P$ or $R$ node $v$ (rest of
  Case~\ref{it:m-def-single} and Case~\ref{it:m-def-P}): The edge is
  simply added to $\Gamma(v)$.
\end{enumerate}

\subsection{Critical paths in SPQR-trees.}\label{sec:critical}

We will be working a lot with paths in $T$ defined by a pair of
vertices in $x,y\in V[G]$, and it will be useful to have the following
query operation.
\begin{description}
\item[common-path$(x,y)$:] Return the unique subpath $u\cdots v$ of
  $b(x)\cdots b(y)$ that is either
  \begin{itemize}
  \item the longest subpath of $b(x)\cdots b(y)$ such that
    $\set{x,y}\subseteq\bag(u)\cap\bag(v)$; or
  \item the unique shortest path in $T$ such that $x\in\bag(u)$ and
    $y\in\bag(v)$.
  \end{itemize}
\end{description}

Now, we can use Lemma~\ref{lem:subcritic} as a subroutine for finding and exposing $m(x,y)$:
\begin{lemma}\todo{We need to refer to this somewhere}
Given vertices $x,y$ with $x\in\bag(r)$ and $y\in \bag(u)$, where $u$ is exposed, $m(x,y)$ can be found in constant time.
\end{lemma}
\begin{proof}
Let $r'$ and $v'$ denote the bags found in Lemma~\ref{lem:subcritic}. They form the endpoints of the common path of $x$ and $y$. 
What remains is a simple case analysis. When $v'$ is further from the root than $r'$, or $r'$ is further from the root than $v'$ and no $P$-node lies between them, then $v'$ and $r'$ are the endpoints of the critical path. Otherwise, $v'$ is at most $2$ edges closer to the root han $r'$, and between them, there exists a $P$-node containing $x$ and $y$; this $P$-node is $m(x,y)$.
\end{proof}

\begin{lemma}\label{lem:spqr-mpath-changecount}\todo{We need to refer to this somewhere}
  Let $G$ be a vertex weighted biconnected graph with vertices
  $x,y,y'\in V[G]$. Let $T$ be the SPQR tree for $G$.  Let
  $m_T(x,y)=r\cdots u$ where $x\in\bag(r)$ and $y\in\bag(u)$, and let
  $m_T(x,y')=r'\cdots u'$ where $x\in\bag(r')$ and $y'\in\bag(u')$ and
  $r'$ is on the path $r\cdots u'$.  Then converting a $u$-exposed
  heavy path decomposition of $T$ with root $r$ into a $u'$-exposed
  heavy path decomposition of $T$ with root $r'$ can be done by:
  conceil$(r)$, expose$(r')$, reverse$(r)$, conceil$(r')$,
  expose$(u')$. During this sequence, $\OO\paren*{ 1 +
    \log\frac{w(T)}{w(y)} + \log\frac{w(T)}{w(y')} }$ edges change
  from solid to dashed or vice versa.
\end{lemma}
\begin{proof}
  By Lemma~\ref{lem:dyntree-internal-changecount} conceil$(r)$ changes
  at most $2\ell(u)+1$ edges. By definition of $m_T(x,y)$, if
  $y\in\bag(r)$ then $u$ and $r$ are neighbors, so $\ell(u)\leq 1$;
  otherwise $u=b(y)$ and by Lemma~\ref{lem:treedecomp-lightdepth} we
  have $\ell(u)\leq\floor*{\log_2\frac{w(T)}{w(y)}}$. Thus, the number
  of edge changes made by conceil$(r)$ is at most
  \begin{align*}
    2\max\set*{ 1, \floor*{\log_2\frac{w(T)}{w(y)}} }+1\leq
    2\floor*{\log_2\frac{w(T)}{w(y)}}+3
  \end{align*}
  Similarly, by Lemma~\ref{lem:dyntree-internal-changecount}
  expose$(r')$ changes at most $2\ell(r')+1$ edges.  If
  $y'\not\in\bag(r')$ then $r'$ is an ancestor to $b(y')$, so by
  Lemma~\ref{lem:treedecomp-lightdepth} we have
  $\ell(r')\leq\floor*{\log_2\frac{w(T)}{w(y')}}$.  Otherwise, by
  definition of $m_T(x,y)$, $\set{x,y'}\subseteq\bag(r')\cap\bag(y')$,
  and since no $3$ consecutive edges in a (strict) SPQR tree can
  correspond to the same separation pair,
  $b(y')\in\set{r',p(r')}$, so $\ell(r')\leq\ell(b(y'))+1$ and so by
  Lemma~\ref{lem:treedecomp-lightdepth}
  $\ell(r')\leq\floor*{\log_2\frac{w(T)}{w(y')}}+1$. Thus, the number
  of edge changes made by expose$(r')$ is at most
  \begin{align*}
    2\max\set*{ \floor*{\log_2\frac{w(T)}{w(y')}},
      \floor*{\log_2\frac{w(T)}{w(y')}}+1 }+1 =
    2\floor*{\log_2\frac{w(T)}{w(y')}}+3
  \end{align*}
  The reverse$(r)$ does not change any edges, and by completely
  symmetric arguments, conceil$(r')$ changes at most
  $2\floor*{\log_2\frac{w(T)}{w(y)}}+3$ edges, and expose$(u')$
  changes at most $2\floor*{\log_2\frac{w(T)}{w(y')}}+3$ edges.
  The total number of edge changes is thus at most
  \begin{align*}
    2\paren*{2\floor*{\log_2\frac{w(T)}{w(y)}}+3}
    +{}&
    2\paren*{2\floor*{\log_2\frac{w(T)}{w(y')}}+3}
    \\
    &=
    4\floor*{\log_2\frac{w(T)}{w(y)}} +
    4\floor*{\log_2\frac{w(T)}{w(y')}} + 12
    \\
    &\in \OO\paren*{1 + \log_2\frac{w(T)}{w(y)} + \log_2\frac{w(T)}{w(y')}}
    \qedhere
  \end{align*}
\end{proof}

\subsection{Pre-splitting}\label{sec:dynspqr-presplit}

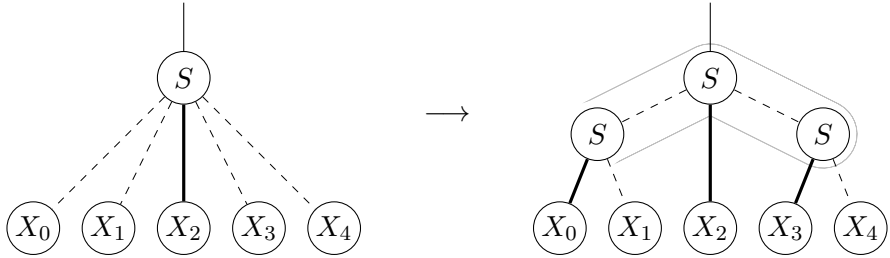
\begin{figure}[H]
  \center
  \begin{tikzpicture}[
    spqr-node/.style={
      draw,
      circle,
      fill=white,
      minimum size=7mm,
      inner sep=0pt,
    },
    spqr-edge/.style={
      draw,
    },
    dashed spqr-edge/.style={
      spqr-edge,
      very thin,
      dashed
    },
    solid spqr-edge/.style={
      spqr-edge,
      very thick,
    }
  ]
  \begin{scope}
    \coordinate (p) at (2,3);
    \node[spqr-node] (s0) at (2,2) {$S$};
    \node[spqr-node] (l0) at (0,0) {$X_0$};
    \node[spqr-node] (l1) at (1,0) {$X_1$};
    \node[spqr-node] (l2) at (2,0) {$X_2$};
    \node[spqr-node] (l3) at (3,0) {$X_3$};
    \node[spqr-node] (l4) at (4,0) {$X_4$};
    \draw[spqr-edge] (p) -- (s0);
    \draw[dashed spqr-edge] (s0) -- (l0);
    \draw[dashed spqr-edge] (s0) -- (l1);
    \draw[solid spqr-edge] (s0) -- (l2);
    \draw[dashed spqr-edge] (s0) -- (l3);
    \draw[dashed spqr-edge] (s0) -- (l4);
  \end{scope}
  \begin{scope}[shift={(5.5,0)}]
    \node at (0,1.5) {$\longrightarrow$};
  \end{scope}
  \begin{scope}[shift={(7,0)}]
    \coordinate (p) at (2,3);
    \node[spqr-node] (s0) at (2,2) {$S$};
    \node[spqr-node] (s1) at (0.5,1.25) {$S$};
    \node[spqr-node] (s2) at (3.5,1.25) {$S$};
    \node[spqr-node] (l0) at (0,0) {$X_0$};
    \node[spqr-node] (l1) at (1,0) {$X_1$};
    \node[spqr-node] (l2) at (2,0) {$X_2$};
    \node[spqr-node] (l3) at (3,0) {$X_3$};
    \node[spqr-node] (l4) at (4,0) {$X_4$};
    \draw[spqr-edge] (p) -- (s0);
    \draw[dashed spqr-edge] (s0) -- (s1);
    \draw[dashed spqr-edge] (s0) -- (s2);
    \draw[solid spqr-edge] (s1) -- (l0);
    \draw[dashed spqr-edge] (s1) -- (l1);
    \draw[solid spqr-edge] (s0) -- (l2);
    \draw[solid spqr-edge] (s2) -- (l3);
    \draw[dashed spqr-edge] (s2) -- (l4);
    \begin{pgfonlayer}{background}
      \filldraw[line width=26,line join=round,gray!50](s0.center)--(s1.center)--(s0.center)--(s2.center)--cycle;
      \filldraw[line width=25,line join=round,white](s0.center)--(s1.center)--(s0.center)--(s2.center)--cycle;
    \end{pgfonlayer}
  \end{scope}
\end{tikzpicture}
   \caption{\label{fig:pre-splitting-S}Pre-splitting an $S$ node
    creates a primary and at most $2$ secondary child nodes. The
    primary inherits the solid child, the secondaries partition the
    remaining children and choose their heaviest child as solid.}
\end{figure}

We want to maintain an SPQR tree efficiently as the underlying graph
changes.  However, the strict SPQR tree can change by up to $\OO(n)$
splits and the contraction of a path with up to $\OO(n)$ nodes during
a single edge insertion in the underlying graph, so we can not hope to
maintain an explicit representation of that.  In this section we
define a \emph{pre-split} SPQR tree, which is just a particular
relaxed SPQR tree with a ($u$-exposed or proper) heavy path
decomposition, that we have better control over. In particular, we
ensure that an edge insertion only takes $\OO(\log n)$ splits,
followed by the contraction of a solid path. In
Section~\ref{sec:dynspqr-precontract} we show how we can
handle this contraction efficiently.  The main difficulty in handling
the splits is that the expose/conceil/etc. operations will actually be
changing the tree we are working on, by splitting and merging certain
nodes.  For the definition and analysis of our pre-split SPQR tree, we
will therefore start with a ($u$-exposed or proper) heavy path
decomposition of the strict SPQR tree.
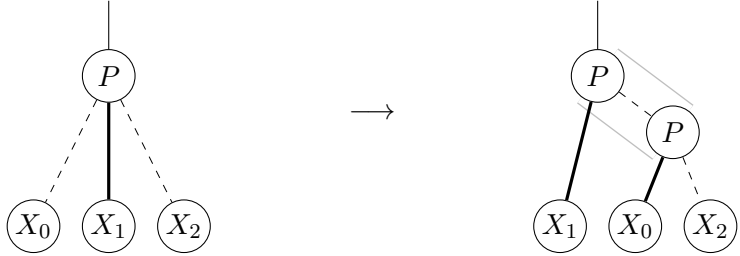
\begin{figure}[H]
	\center
	\begin{tikzpicture}[
    spqr-node/.style={
      draw,
      circle,
      fill=white,
      minimum size=7mm,
      inner sep=0pt,
    },
    spqr-edge/.style={
      draw,
    },
    dashed spqr-edge/.style={
      spqr-edge,
      very thin,
      dashed
    },
    solid spqr-edge/.style={
      spqr-edge,
      very thick,
    }
  ]
  \begin{scope}
    \coordinate (p) at (1,3);
    \node[spqr-node] (p0) at (1,2) {$P$};
    \node[spqr-node] (l0) at (0,0) {$X_0$};
    \node[spqr-node] (l1) at (1,0) {$X_1$};
    \node[spqr-node] (l2) at (2,0) {$X_2$};
    \draw[spqr-edge] (p) -- (p0);
    \draw[dashed spqr-edge] (p0) -- (l0);
    \draw[solid spqr-edge] (p0) -- (l1);
    \draw[dashed spqr-edge] (p0) -- (l2);
  \end{scope}
  \begin{scope}[shift={(4.5,0)}]
    \node at (0,1.5) {$\longrightarrow$};
  \end{scope}
  \begin{scope}[shift={(7,0)}]
    \coordinate (p) at (.5,3);
    \node[spqr-node] (p0) at (.5,2) {$P$};
    \node[spqr-node] (p1) at (1.5,1.25) {$P$};
    \node[spqr-node] (l1) at (0,0) {$X_1$};
    \node[spqr-node] (l0) at (1,0) {$X_0$};
    \node[spqr-node] (l2) at (2,0) {$X_2$};
    \draw[spqr-edge] (p) -- (p0);
    \draw[solid spqr-edge] (p0) -- (l1);
    \draw[dashed spqr-edge] (p0) -- (p1);
    \draw[solid spqr-edge] (p1) -- (l0);
    \draw[dashed spqr-edge] (p1) -- (l2);
    \begin{pgfonlayer}{background}
      \filldraw[line width=26,line join=round,gray!50](p0.center)--(p1.center)--cycle;
      \filldraw[line width=25,line join=round,white](p0.center)--(p1.center)--cycle;
    \end{pgfonlayer}
  \end{scope}
\end{tikzpicture}
 	\caption{\label{fig:pre-splitting-P}Pre-splitting a $P$ node creates
		a primary and at most $1$ secondary child node. The primary
		inherits the solid child, the secondary inherits the remaining
		children and chooses its heaviest child as solid.}
\end{figure}
\begin{definition}
  A node $v$ in a rooted relaxed SPQR tree is \emph{primary} if either
  $v$ is the root; or $v$ is an $R$ node; or $v$ is an $S$ or $P$ node
  whose parent $p(v)$ has different type. All other nodes are called
  \emph{secondary}.
\end{definition}
\noindent In any relaxed SPQR tree $T$ there is a one-to-one
correspondence between the primary nodes in $T$ and the nodes in the
corresponding strict SPQR tree.

\begin{definition}\label{def:ps-selected}
  For each node $u$, we choose a set $s(u)\subseteq\bag(u)$ of
  \emph{selected} vertices, and use them to define $1$ or $2$
  \emph{child groups} of $u$ as disjoint subsets of the edges of
  $\Gamma(u)$, as follows:
  \begin{itemize}

  \item If $u$ is a primary $S$ node with both a parent edge $e_p$ and
    a solid child edge $e_c$, then $s(u)=\emptyset$, and the two child
    groups are simply the (possibly empty) ordered sets of edges on
    the clockwise and counterclockwise paths in $\Gamma(u)$ between the 
    virtual edges corresponding to 
    $e_p$ and $e_c$ (excluding $e_p$ and $e_c$).

  \item For the at most one primary non-root $S$ node $u$ that has a
    parent edge $e_p$ but no solid child and is connected to the root
    by a solid path, we may choose a vertex $x\in\bag(u)$ and set
    $s(u)=\set{x}$. Then the two (possibly empty) child groups are the
    ordered sets of edges edges that correspond to the clockwise and
    counterclockwise paths in $\Gamma(u)$ from $e_p$ to $x$ (excluding
    $e_p$).

  \item If $u$ is the root and is an $S$ node with a solid child edge
    $e_c$ we may choose a vertex $x\in\bag(u)$, and set
    $s(u)=\set{x}$. Then the two (possibly empty) child groups are the
    ordered sets of edges that correspond to the clockwise and
    counterclockwise paths in $\Gamma(u)$ from $x$ to $e_c$ (excluding
    $e_c$).

  \item If $u$ is the root and is an $S$ node with no solid children,
    we may choose two selected vertices $x,y\in\bag(u)$, and set
    $s(u)=\set{x,y}$.  Then the two (possibly empty) child groups are
    the ordered sets of edges that correspond to the clockwise and
    counterclockwise paths in $\Gamma(u)$ between $x$ and $y$.

  \item All other nodes (including $R$ and $P$ nodes and all secondary
    nodes) have $s(u)=\emptyset$, and only one (possibly empty) child
    group, consisting of all edges in $\Gamma(u)$ except the ones
    corresponding to the parent edge or the solid child.

  \end{itemize}
  Note that for every dashed child, its corresponding edge in
  $\Gamma(u)$ belongs to a unique child group, and that the parent
  edge and solid child edge (if any) do not belong to a child group.
\end{definition}
We are going to use the sets of selected nodes to control the exact
locations where $S$ nodes are split when preparing to insert a new
edge. To do this we define a new operation:
\begin{description}
\item[select$(v, X)$:] Where $v\in\set{r,u}$ is an $S$ node in a
  $u$-exposed SPQR tree with root $r$, and $X$ (possibly empty) is a
  valid choice of selected vertices according to
  definition~\ref{def:ps-selected}.  Set $s(v):=X$.
\end{description}

\begin{definition}
  Define a node $v$ to be \emph{splittable} if either:
  \begin{itemize}
  \item $v$ is a $P$ node with both a parent and a solid child and a
    child group with at least $2$ edges; or
  \item $v$ is an $S$ node with two groups of dashed children and at
    least one child group with at least $2$ edges.
  \end{itemize}
  All other nodes are \emph{unsplittable}. In particular, the root,
  the leaves, and the $R$ nodes are unsplittable.
\end{definition}
\begin{definition}
  Define a node to be \emph{mergeable} if it has a secondary child.
  All other nodes are \emph{unmergeable}.
\end{definition}
\noindent In a strict SPQR tree many nodes may be splittable, but no
node is mergeable.
We will define our relaxed SPQR tree by \emph{pre-splitting} some
subset of the splittable nodes of a strict SPQR tree, using one of the
following operations:
\begin{description}
\item[split$(v)$:] Where $v$ is a splittable node.  For each child
  group at least $2$ edges, replace the group of dashed children
  corresponding to that group with a single new dashed secondary child
  node $c$ inheriting the edges that it replaces, and with
  $s(c)=\emptyset$. For each new node $c$ created this way, if it has
  a heavy child $h$ make $(h,c)$ solid.
  
  Note that somewhat counter-intuitively, this definition may
  split an $S$ node $v$ with $\abs{s(v)}=2$ into $3$, where the root
  node is a $2$-cycle. 
\end{description}
As the underlying graph changes, we also need the inverse operation:
\begin{description}
\item[merge$(v)$:] Where $v$ is a mergeable node. For each
  secondary child $c$: if $(c,v)$ is dashed and $c$ has a solid child
  $c_2$, make $(c_2,c)$ dashed; then contract edge $(c,v)$.
\end{description}

\begin{lemma}\label{lem:ps-relaxed-invariants}
  If we start with a heavy path decomposition of a strict SPQR tree
  with some set of selected vertices and only split or merge primary
  nodes, the following invariants hold:
  \begin{enumerate}[label={I-\arabic*.}, ref={I-\arabic*}, series=presplitting-invariants]
  \item\label{inv:ps-not-both-split-merge} No node is both
    splittable and mergeable.
  \item\label{inv:ps-secondary} Each secondary node has a skeleton
    graph with at least $3$ edges, and is a dashed child of a primary
    node which is either a $P$ node with both a parent and a solid
    child, or an $S$ node with two groups of children.
  \end{enumerate}
  Furthermore, any heavy path decomposition of a relaxed SPQR tree
  where these inveriants hold can be constructed from the
  corresponding strict SPQR tree by splitting a unique subset of the
  splittable primary nodes in any order.
\end{lemma}
\begin{proof}
  The invariants hold trivially for a strict SPQR tree, since every
  node is primary.

  Assume now that the invariants hold for the current relaxed tree,
  and that $p$ is a splittable primary node.  Then after split$(p)$,
  $p$ is no longer splittable but is instead mergeable,
  so~\ref{inv:ps-not-both-split-merge} still holds at $p$.  Since $p$
  was splittable and not mergeable, either $p$ is a $P$ node that is
  not the root and has a solid child, or $p$ is an $S$ node with two
  groups of children, and thus each secondary node created by the split
  satisfies~\ref{inv:ps-secondary}.

  Since $p$ was not mergeable, each node $h$ that got a new parent $c$
  in the split was primary, so making the parent edge $(h,c)$ solid if
  it was heavy does not violate~\ref{inv:ps-secondary}. For every
  other node, nothing has changed so the invariants trivially still
  hold.

  Similarly, suppose that the invariants hold for the current relaxed
  tree and that $p$ is a mergeable node.
  Since~\ref{inv:ps-secondary} holds, $p$ is primary, and after
  merge$(p)$ we have $p$ splittable but not mergeable
  so~\ref{inv:ps-not-both-split-merge} still holds at $p$. The
  merge may make some solid edges dashed and removes some
  secondary nodes, but doesn't change any of the invariants for any of
  the other nodes.

  Finally note that as long as the invariants hold, the primary nodes
  in the relaxed tree that are splittable or mergeable are exactly
  those that correspond to a splittable node in the strict SPQR tree.
  To get a particular relaxed tree, start with the strict tree and
  split exactly those primary nodes that should be mergeable.
\end{proof}
In particular, invariant~\ref{inv:ps-not-both-split-merge}
and~\ref{inv:ps-secondary} means that each node in the original strict
SPQR tree is represented by at most $3$ (a primary and at most $2$
secondary) nodes in the relaxed SPQR tree. This is somewhat analogous
to the way a 2-3-4-tree is related to a red-black tree (See Guibas and Sedgewick~\cite{DBLP:conf/focs/GuibasS78}).

So far, we have said very little about which nodes to actually split and
merge.  Each time an edge in the strict SPQR tree changes from
solid to dashed or vice versa during an update, our relaxed tree may
have to change. We want to ensure that each such edge change causes at
most a constant number of splits and merges. In order to get
there, we need to distinguish between two types of dashed edges.
\begin{definition}\label{def:stable}
  Define the \emph{solid degree} of a node $v$ to be $\abs{s(v)}$ plus
  the number of solid edges incident to $v$.  Define an edge
  $(c,p(c))$ to be \emph{stable} if it is dashed and either
  \begin{itemize}
  \item $p(c)$ is an $R$ node; or
  \item $p(c)$ is an $S$ or $P$ node with solid degree $2$ and the
    child group of $p(c)$ containing $(c,p(c))$ contains no other
    edges.
  \end{itemize}
  All other dashed edges are \emph{unstable}. In particular, every
  dashed child edge of an $S$ or $P$ node that is at the end of a
  solid path not containing the root is unstable, as are all child
  edges of a splittable node that would get moved during a split.
\end{definition}
\begin{observation}\label{obs:contract-stable}
  If the edge $(c,p(c))$ is stable, then contracting the solid path
  containing $p(c)$ into an $R$ node during an edge insertion does not
  change this.
\end{observation}
\begin{lemma}\label{lem:split-merge-nonmoving-child}
  If splitting or merging the node $p$ does not give its dashed child $c$
  a new parent, then $(c,p)$ is unstable after the operation if and only
  if it was unstable before the operation.
\end{lemma}
\begin{proof}
  Since $c$ does not get a new parent, $(c,p)$ was stable if and only
  if $p$ was an $S$ node of solid degree $2$ and the child group $p$
  containing $(c,p)$ had no other elements. The split does not change
  the solid degree of $p$, and only changes the other child group of
  $p$.
\end{proof}
\begin{lemma}\label{lem:split-moving-child}
  If splitting the node $p$ gives its dashed child $c$ a new parent
  $v$, then $(c,p)$ was unstable before the split, $(c,v)$ is either
  solid or unstable after the split, and $\Gamma(v)$ has at least $3$
  edges.
\end{lemma}
\begin{proof}
  Since $c$ gets a new parent during the split, the child group of $p$
  containing $(c,p)$ has at least $2$ edges, and either $p$ is a $P$
  node or $p$ is an $S$ node with at least $2$ child groups.
  Thus by definition of stable, $(c,p)$ is unstable before the split.
  By definition of split, $\Gamma(v)$ has at least $3$ edges after the
  split, and $(v,p)$ is dashed and $v$ has solid degree at most $1$,
  so by definition of stable $(c,v)$ is not stable.
\end{proof}
\begin{lemma}\label{lem:merge-moving-child}
  If merging node $p$ contracts the dashed edge $(v,p)$ where $v$ has
  a child $c$ and $\Gamma(v)$ has at least $3$ edges, then $(c,v)$ was
  solid or unstable before the merge, and $(c,p)$ is unstable after
  the merge.  
\end{lemma}
\begin{proof}
  Since $(v,p)$ is dashed, any child edge of $v$ is either solid or
  unstable before the merge, and is dashed after the merge.  Since
  $\Gamma(v)$ had at least $3$ edges, $(c,p)$ is unstable after the
  merge.
\end{proof}
Thus, if the edge $(c,p(c))$ is solid or unstable then after a split
or merge that changes $p(c)$, the new parent edge of $c$ is still
either solid or unstable (but not necessarily the same one as before).
Given these properties we can now state
\begin{lemma}\label{lem:ps-unique-relaxed}
  Given a $u$-exposed (or proper) heavy path decomposition of a strict
  SPQR tree $T$ with root $r$, and sets $s(u)$ and $s(r)$, there is
  a unique corresponding relaxed SPQR tree $T'$ that satisfies the
  inveriants of Lemma~\ref{lem:ps-relaxed-invariants} as well as the
  following:
  \begin{enumerate}[resume*=presplitting-invariants]
  \item\label{inv:ps-split-rule} The root, and each primary node with
    a solid or unstable parent edge is not splittable.
  \item\label{inv:ps-merge-rule} A primary node with a stable
    parent edge is not mergeable.
  \end{enumerate}
  We call $T'$ the \emph{pre-split} SPQR tree corresponding to
  $(T,s(u),s(r))$.
\end{lemma}
\begin{proof}
  The strict SPQR tree trivially satisfies all invariants
  except~\ref{inv:ps-split-rule}. Let $X$ be the set of nodes in the
  strict SPQR tree that violate~\ref{inv:ps-split-rule}.  For each
  node $v\in X$, calling split$(v)$ satisfies~\ref{inv:ps-split-rule}
  at $v$, without breaking it for any other nodes.
\end{proof}

\begin{lemma}\label{lem:ps-restorecount}
  If a relaxed SPQR tree violates the pre-splitting invariants in at
  most $k$ nodes, then they can be restored by at most $20k$ split or
  merge operations.
\end{lemma}
\begin{proof}
  For each secondary node $v$ that is a solid child, we have a
  violation of invariant~\ref{inv:ps-secondary}, and we can set
  $p=p(v)$ and do merge$(p)$, followed (if the resulting node is
  splittable) by a split$(p)$. Each such merge+split removes the
  offending node and violation, but may create up to $2$ new
  violations of~\ref{inv:ps-secondary} (at the new dashed children of
  $p$), and at most $4$ new violations of~\ref{inv:ps-split-rule} (at
  stable children of $v$ and $p(v)$ whose parent edges become
  unstable).  After doing this for the at most $k$ secondary nodes
  that are solid children (at a cost of at most $2k$ operations),
  every secondary node is a dashed child, and we have at most $6k$
  other violations.
  For each secondary node $v$ that is now a child of a secondary node,
  there is a violation of invariant~\ref{inv:ps-secondary}, and we can
  do a merge$(p(v))$. Each such merge removes the offending
  node and violation, but may create a new violation
  of~\ref{inv:ps-not-both-split-merge},~\ref{inv:ps-secondary},
  or~\ref{inv:ps-split-rule} at the parent. However, after at most $6k$
  such merges, every secondary node has a primary parent.
  If invariant~\ref{inv:ps-secondary} is now violated at a secondary
  node $v$, then a merge$(p(v))$ will remove both the node and the
  violation, but may create a new violation of
  invariant~\ref{inv:ps-not-both-split-merge}
  or~\ref{inv:ps-split-rule} at the parent instead.  So after at most
  $6k$ additional merges, invariant~\ref{inv:ps-secondary} is
  satisfied at every secondary node, but we may still have $6k$
  violations of the other invariants.
  If invariant~\ref{inv:ps-not-both-split-merge}
  and/or~\ref{inv:ps-merge-rule} is now violated at a node $v$,
  then a merge$(v)$ fixes it and reduces the total number of
  violations by at least one.
  Finally, if invariant~\ref{inv:ps-split-rule} is now violated at a
  node $v$, then a split$(v)$ fixes it and reduces the total number of
  violations by one.
\end{proof}
Note that this is a very loose upper bound, and we can often do much
better by analyzing the specific violations.

\begin{lemma}\label{lem:ps-select-restorecount}
  In a $u$-exposed pre-split SPQR tree with root $r$, restoring the
  invariants after a select$(u,X)$ or select$(r,X)$ takes at most one
  split or merge operation.
\end{lemma}
\begin{proof}
  For $v\in\set{u,r}$, the select$(v,X)$ operation can only affect
  invariants~\ref{inv:ps-not-both-split-merge}
  and~\ref{inv:ps-split-rule} at $v$ and~\ref{inv:ps-secondary} at
  children of $v$. A single merge or split will fix this without
  introducing any new violations. To see this, first note that if
  $s(v)\neq\emptyset$ then the only valid value for $X$ is
  $X=\emptyset$. In this case, invariant~\ref{inv:ps-secondary} may be
  violated at a child $c$ of $v$ because $v$ only has one child group
  after the operation, and this is fixed by a merge$(v)$. Afterwards,
  $v$ is neither mergeable or splittable, so
  both~\ref{inv:ps-not-both-split-merge} and~\ref{inv:ps-split-rule}
  are also satisfied at $v$.  If $s(v)=\emptyset$ and
  $X\neq\emptyset$, then $v$ may become splittable and hence
  violate~\ref{inv:ps-split-rule}, but this is fixed by a single
  split$(v)$.
\end{proof}

\begin{lemma}\label{lem:ps-edge-change-restorecount}
  In the pre-split SPQR tree, if $c$ and $p=p(c)$ are both primary,
  then changing edge $(c,p)$ from solid to dashed or vice versa may
  violate the invariants at $p$ or $c$ or at most $3$ of their
  children, but nowhere else.  The invariants can be restored by at
  most $3$ split or merge operations.
\end{lemma}
\begin{proof}
  Suppose first that $(c,p)$ is changed from solid to dashed.
  This does not cause~\ref{inv:ps-not-both-split-merge} to be
  violated anywhere, but may cause~\ref{inv:ps-secondary} to be
  violated at secondary children of $p$ (because they lose a solid
  sibling), cause~\ref{inv:ps-split-rule} to be violated at at most
  one primary child of each of $p$ or $c$ (because their parent edge
  changes from stable to unstable), and
  cause~\ref{inv:ps-merge-rule} to be violated at $c$ (because its
  parent edge changes from solid to stable).
  Note that~\ref{inv:ps-secondary} can only be violated at a secondary
  child of $p$ if $p$ \emph{is not} an $R$ node,
  and~\ref{inv:ps-merge-rule} can only be violated at $c$ if $p$
  \emph{is} an $R$ node, so these two types of violation do not occur
  together.
  To restore the invariants takes at at most one merge operation at
  either $p$ (to fix the violation of~\ref{inv:ps-secondary} at the
  secondary children of $p$) or $c$ (to fix the violation
  of~\ref{inv:ps-merge-rule} at $c$), and at most one split
  operation at each primary child $v$ of $p$ or $c$ whose parent edge
  changes from stable to unstable (to fix the violation
  of~\ref{inv:ps-split-rule}).

  Suppose now that $(c,p)$ is changed from dashed to solid.
  This does not cause~\ref{inv:ps-not-both-split-merge}
  or~\ref{inv:ps-secondary} to be violated anywhere, but may
  cause~\ref{inv:ps-split-rule} to be violated at $p$ (because $p$
  becomes splittable) or $c$ (because its parent edge changes from
  stable to solid), and cause~\ref{inv:ps-merge-rule} to be
  violated at at most one primary child of each of $p$ and $c$
  (because their parent edge changes from unstable to stable).
  Note that~\ref{inv:ps-split-rule} can only be violated at $p$ if $p$
  \emph{is not} and $R$ node, and can only be violated at $c$ if $p$
  \emph{is} an $R$ node, the invariant is at most violated at one of
  them.
  To restore the invariants takes at most one split operation at
  either $p$ or $c$ (to fix the violation of~\ref{inv:ps-split-rule}),
  and at most one merge
  operation at each primary child $v$ of $p$ or $c$ whose parent edge
  changes from unstable to stable (to fix the violation
  of~\ref{inv:ps-merge-rule}).
\end{proof}

To handle expose$(u)$ of a primary node $u$ in a proper heavy path
decomposition of the pre-split SPQR tree with root $r$, or
conceil$(r)$ in a $u$-exposed heavy path decomposition of the
pre-split SPQR tree with root $r$ and primary node $u$, we need to
slightly change the implementation of splice$(v)$ and slice$(v)$ as
follows:
\begin{description}
\item[splice$(v)$:] Where $v$ is a primary node that is a dashed light
  child.
  If $p(v)$ is a secondary node, let $p=p(p(v))$, else let
  $p=p(v)$. Now $p$ is a primary node. If $p$ has a solid child then
  make it dashed and do the necessary splits and merges to restore
  the invariants. Now $p$ is the parent of $v$ and $(v,p)$ is
  dashed. Make $(v,p)$ solid, and do the necesary splits and merges
  to restore the invariants.

\item[slice$(v)$:] Where $v$ is a primary node that is a solid light
  child.
  Let $p=p(v)$, then $p$ is primary and $(v,p)$ is solid. Make $(v,p)$
  dashed and do the necessary splits and merges to restore the
  invariants. Now $p$ has no solid children. If $p$ has a heavy child
  $h$, make $(h,p)$ solid and do the necessary splits and merges to
  restore the invariants.
\end{description}

\begin{lemma}\label{lem:ps-expose-conceil}
  An expose$(u)$ where $u$ is a primary node in a proper heavy path
  decomposition of the pre-split SPQR tree with no selected vertices
  and root $r$, or a conceil$(r)$ in a $u$-exposed heavy path
  decomposition of the pre-split SPQR tree with no selected vertices
  and root $r$ each costs at most $6\ell(u)+3$ splits or merges, in
  addition to the at most $2\ell(u)+1$ edges changing from solid to
  dashed or vice versa, where $\ell(u)$ is the light depth of $u$ in
  the corresponding strict SPQR tree.
\end{lemma}
\begin{proof}
  With the modified definition of splice and slice given, expose$(u)$
  and conceil$(r)$ works exactly as before, except that after changing
  the solid/dashed state of a child of $u$ we need to do the necessary
  splits and merges to restore the invariants. The number of edges we
  explicitly change (as opposed to what happens while restoring the
  invariants) is unchanged $2\ell(u)+1$.  Each edge that we explicitly
  change from solid to dashed or vice versa has a primary node at each
  end, so by Lemma~\ref{lem:ps-edge-change-restorecount}, restoring
  the invariants costs at most $3$ splits or merges per explicit edge
  change. Thus the total number of merges and splits is at most
  $3(2\ell(u)+1)=6\ell(u)+3$.
\end{proof}

\begin{lemma}\label{lem:ps-link-exposed}
  Given a proper pre-split SPQR tree $T_v$ with no selected vertices,
  and an $u$-exposed pre-split SPQR tree $T_r$, where $(x,y)$ is a
  non-virtual edge in both $\Gamma(v)$ and $\Gamma(u)$, restoring the
  invariants after a link-exposed$(v,u)$ takes at most one split or
  merge operation.
\end{lemma}
\begin{proof}
  After link-exposed$(v,u)$,
  invariant~\ref{inv:ps-secondary},~\ref{inv:ps-split-rule},
  or~\ref{inv:ps-merge-rule} may be violated at $v$, but no other
  violations occur anywhere in the tree.
  If~\ref{inv:ps-secondary} is violated at $v$, by definition of
  secondary $v$ must have the same type as $u$, so the violation is
  fixed with a merge$(u)$.
  If~\ref{inv:ps-split-rule} or~\ref{inv:ps-merge-rule} is violated at
  $v$ then $v$ was a node of different type from $u$ before the
  operation, so the violation is fixed with either a split$(v)$ or a
  merge$(v)$.
  Since the type of $v$ is not both the same as and different from $u$, at
  most one split or merge operation is needed in total to restore the
  invariants.
\end{proof}

\begin{lemma}\label{lem:ps-cut-exposed}
  Given a $u$-exposed pre-split SPQR tree $T_r$ and a child $v$ of
  $u$, restoring the invariants after a cut-exposed$(v)$ takes at most
  one merge operation if $v$ is a primary node, and no violations of
  the invariants occur otherwise.
\end{lemma}
\begin{proof}
  No violations of the invariants occur in the part of the tree
  containing $u$, because no parents or child groups are changed
  there.
  If $v$ was a primary node before, then after cut-exposed$(v)$
  invariant~\ref{inv:ps-secondary} may be violated at children of $v$
  (because $s(v)=\emptyset$ and $v$ loses its parent). If this
  happens, a merge$(v)$ fixes it.
  If $v$ was secondary, every child of $v$ is primary so no violations
  of~\ref{inv:ps-secondary} can occur there.
  By definition, $s(v)=\emptyset$ so $v$ is not splittable
  so~\ref{inv:ps-split-rule} is satisfied at $v$, and $v$ has no
  parent, so~\ref{inv:ps-merge-rule} is satisfied at $v$. All other
  invariants remain unchanged.
\end{proof}

\begin{lemma}\label{lem:ps-reverse}
  In the pre-split SPQR tree with root $r$, reverse$(r)$ does not
  violate any of the invariants.
\end{lemma}
\begin{proof}
  Let $u$ be the other end of the solid path containing $r$.
  Since no node changes whether it is splittable or mergeable,
  invariant~\ref{inv:ps-not-both-split-merge} is trivially
  preserved.
  By invariant~\ref{inv:ps-secondary}, each node $v\in r\cdots u$ is
  primary. By definition of splittable and
  invariant~\ref{inv:ps-split-rule}, every node $v\in r\cdots u$ is
  unsplittable before the operation. This will still be true after the
  operation, so invariant~\ref{inv:ps-split-rule} is preserved for
  these nodes.
  None of the nodes on $r\cdots u$ have a stable parent before or
  after the operation, so invariant~\ref{inv:ps-merge-rule} is
  trivially preserved for these nodes.
  Each child $c$ of a node $p\in r\cdots u$ is either a secondary node
  where~\ref{inv:ps-secondary} is preserved, or a primary node, whose
  type of parent edge (stable or unstable) is preserved and
  thus~\ref{inv:ps-split-rule} and~\ref{inv:ps-merge-rule} is
  preserved for $c$.
  The invariants are trivially preserved for the remaining nodes.
\end{proof}

Combining these Lemmas, we obtain the following main theorem stating that we can support the aforementioned operations while only making changes proportional to the number of edges that change status between solid and dashed:

\begin{theorem}\label{thm:SPQRthm}
  The number of splits and merges needed to maintain a pre-split SPQR
  tree under expose, conceil, reverse, link-exposed, and cut-exposed
  is proportional to the number of edges changing from solid to dashed
  or vice versa during the same operation in the strict SPQR tree.
\end{theorem}
 \subsection{Pre-contracting}\label{sec:dynspqr-precontract}
We have now shown that a pre-split SPQR tree doesn't change too much
during each of the operations we need.  However, we still need to
describe how to handle the actual path contraction that happens during
an edge insertion, and the corresponding path expansion that happens
during an uninsert, in worst case $\OO(\log n)$ time. Formally we
define the operations
\begin{description}

\item[contract-path$(u)$:] Given a $u$-exposed pre-split SPQR tree
  $T_r$, where each $v\in\set{r,u}$ is either an $R$ node, or an $S$
  node with $\abs{bag(v)}=3$.  Contract the path $r\cdots u$ into a
  single \emph{pseudo-$R$} node $v$. The result is a $v$-exposed
  pre-split SPQR tree.

\item[expand-path$(r)$:] Given a $r$-exposed pre-split SPQR tree
  $T_r$, where $r$ is a pseudo-$R$ node contracted from a path
  $u\cdots v$, expand the node $r$ into that path.  The result is
  either a $u$-exposed pre-split SPQR tree with root $v$, or a
  $v$-exposed pre-split SPQR tree with root $u$.
\end{description}
The idea is to represent each $R$ node and each solid path as
variations of the same tree data structure, such that one can be
turned into the other just by changing information at the root of that
tree.

Note that every $R$ node we have, will be the result of some earlier
contraction, and that because of pre-splitting every $S$ or $P$ node
that gets contracted has a skeleton graph with $3$ or $4$ edges.

We start by defining a \emph{path representation} $\Pi(v)$ of an SPQR
node $v$, as a graph constructed from $v$ by creating a vertex
$\pi(v_i)$ for each original $S$ or $P$ node $v_i$ it was contracted
from, and connecting them with edges into a path. We assign each edge
a positive integer weight.  Note that for an $S$ or $P$ node $v$, the
path representation is just the single-vertex graph
$\Pi(v)=(\set{\pi(v)},\emptyset)$.

Given path representations $\Pi(c)$ for each node $c$ on a solid path
$u\cdots v$ in the SPQR tree, we can create a combined path
representation $\Pi(u\cdots v)$ of the whole path by connecting the
ends of the paths for each node with new edges of weight $0$.  Thus,
in $\Pi(u\cdots v)$, each node in $u\cdots v$ corresponds to a maximal
segment of the path where every edge has positive weight. In
particular, each $R$ node corresponds to a segment with at least one
edge.
Now, the weight of a vertex $x$ is added to its
representative $a(x)$ defined in Section~\ref{sec:heavyweighted}. However, since we do not explicitly maintain the SPQR-nodes, the weight of $x$ is added to a fragment on the path representing this node.

The idea is now that when doing a contract-path$(r)$, resulting in a
new root $r'$, we can get from $\Pi(r\cdots u)$ to $\Pi(r')$ by simply
adding $1$ to all edges on $\Pi(r\cdots u)$.  Similarly, if we later
do an expand-path$(r')$, we can reconstruct $\Pi(r\cdots u)$ from
$\Pi(r')$ by subtracting $1$ from every edge.

By storing each path using a balanced binary search tree, and doing
the edge updates lazily, we can turn this idea into an effective data
structure with $\OO(\log n)$ worst case time for each of the
operations we need.  
\begin{lemma}{\cite{DBLP:journals/jcss/SleatorT83}}
  We can maintain such a tree in $\OO(\log n)$ time per join, split,
  or add-to-path.
\end{lemma}

In the resulting tree over $\Pi(u\cdots v)$, for every node $c$ in
$u\cdots v$ there is a unique deepest node $\pi(c)$ such that the
subtree rooted in $\pi(c)$ contains all vertices of $\Pi(c)$.  If $v$
is an $S$ or $P$ node in the SPQR tree, this will be a leaf (and so
$\pi(c)$ is the same as already defined above), but if $c$ is an $R$
node this will be not be the case.

We have defined all our operations on pre-split SPQR trees in terms of
the actual nodes, but because we want the path representation tree to
be balanced at all times, the node in this tree that correspond to an
$R$ node may change.  Thus, every pointer to an $R$ node will have to
be implicit, in the sense that what is actually stored is a pointer to
one of the (now obsolete) $S$ or $P$ nodes that it was contracted from, and every time we need the node we have to find it in the tree.
\begin{lemma}\label{lem:pc-find-rep}
  Let $u\cdots v$ be a solid path in the pre-split SPQR tree, with $u$
  an ancestor to $v$.  Let $\Pi=\Pi(u\cdots v)$ be a path
  representation tree for $u\cdots v$, and let $u'$ be en end
  fragment of $\Pi$ that is part of $\Pi(u)$.

  Given a fragment $\pi(c')\in \Pi$ of an SPQR node $c\in u\cdots v$,
  we can in $\OO(\log n)$ time find the first zero-weight edge on
  $\pi(c')\cdots \pi(u')$ (if any); and the internal node representing
  $\pi(u)$.  And in $\OO(\log^2 n)$ time we can find the first
  zero-weight \emph{light} edge on $\pi(c')\cdots \pi(u')$ (if any).\end{lemma}
\begin{proof}
  Each internal node $\pi$ in the path representation tree corresponds
  both to the edge $e_\pi$ separating the two subtrees, and the whole
  subpath $\Pi$ spanned by the subtree rooted at $\pi$.  Let $\pi$
  store a ``lazy'' value $\Delta(\pi)$ that must be added to the
  weights of all edges in the subtree, as well as $w(e_\pi)$ and
  $m_\pi:=\min_{e\in\Pi}w(e)$. These values can easily be maintained as
  the tree is rebalanced.

  To find $\pi(v)$ we first find the path to the root of the path
  representation tree containing $v'$, then we ``push down'' all the
  $\Delta$ values along the path, such that each node $\pi$ along this
  path knows the actual values of $w(e_\pi)$ and $m_\pi$.

  Given these values it is easy to find $\pi(v)$ and to find the first zero-weight edge towards the root.
  
  To find the first zero-weight light edge towards the root, we simulate the algorithm from~\cite{DBLP:journals/jcss/SleatorT83}. This costs an extra $\log$ factor, because instead of storing all the weights, we calculate weights spending $O(\log n)$ time per weight. 
\end{proof}

As the path decomposition (and the tree changes), we will use the
following internal operations to update the path representations of
each solid path:
\begin{description}
\item[join-path$(\pi(v))$:] Where $v$ is a dashed child with no solid
  siblings. This operation is called when the edge $(v,p(v))$ changes from dashed to
  solid.  Join the path representation trees $\Pi_v$ and $\Pi_p$
  containing $\pi(v)$ and $\pi(p(v))$ into a single tree by adding a
  new internal node representing a weight $0$ path-edge between an end
  of $\Pi_v$ and an end of $\Pi_p$, and rebalancing as needed.

  Note that while each tree contains a node that corresponds to
  $(v,p(v))$, this node is unlikely to be at the end of the path.
  Rather than deleting it or rearranging the tree to move the node to
  an end of the path (which would make undo tricky), we simply choose
  one of the ends for each path that is connected via nonzero-weight
  edges to this node.  Since a nonzero-weight edge is interpreted as a
  contraction, this ensures that all the zero-weight edges still form
  the correct path.

\item[split-path$(v)$:] Where $v$ is a solid child. This operation is called when the
  edge $(v,p(v))$ changes from solid to dashed.  Let $a$ be the node
  on the solid path containing $v$ that is closest to the root. Split
  the path representation tree $\Pi_v$ containing $\pi(v)$ into two by
  finding the first zero-weight edge on the path from
  $\pi(v)$ to $\pi(a)$ in $\Pi_v$ and deleting it.
\end{description}

These join- and split-operations can be supported in a time bounded by the time stated in Lemma~\ref{lem:pc-find-rep}
\begin{observation}
Join-path and split-path can be supported in $O(\log ^2 n )$ time, and contract can be supported in $O(\log n)$ time.
\end{observation}

We are now ready to formally define the insert and undo insert operations:

\begin{description}
\item[insert$(r; x,y)$:] In the SPQR-tree rooted in $r$ with $m(x,y)$ exposed, update the SPQR-tree reflecting that the edge $(x,y)$ is inserted.
\item[uninsert$(r; x,y)$:] In the SPQR-tree rooted in $r$ with $r$ exposed and $(x,y)\in \Gamma(r)$, update the SPQR-tree reflecting that the edge $(x,y)$ is removed.
\end{description}

\begin{lemma}\label{lem:ps-insert-uninsert}
  Given vertices $x,y\in V[G]$ and a $u$-exposed pre-split SPQR tree
  $T_r$ where $r\cdots u=m(x,y)$, restoring the invariants during an
  insert$(T_r;x,y)$ takes at most $\OO(1)$ splits and merges. While
  undoing the insert in the resulting tree $T$ with an uninsert$(T;
  (x,y))$, restoring the invariants takes the same number of splits
  and merges.
\end{lemma}
\begin{proof}
  If $x\in\bag(r)\setminus\bag(u)$ and $y\in\bag(u)\setminus\bag(r)$,
  then the insert consist of a select$(r,\set{x})$, a
  select$(u,\set{y})$, and a contract-path$(r)$. By
  Lemma~\ref{lem:ps-select-restorecount}, each select costs at most
  one split, and then the contract-path does not violate any of the
  invariants.

  If $r=u$ is an $S$ node, then the insert starts by a
  select$(r,\set{x,y})$. Restoring the invariants then costs at most
  one split, which turns $r$ into a kind of pseudo-$S$ node with only
  two children.  This node is then converted to a pseudo-$P$ node
  (which doesn't violate any of the invariants), and the edge $(x,y)$
  is added to $\Gamma(r)$ turning it into a real $P$ node.

  If $r\cdots u$ consists of a single edge $(r,u)$ with
  $\bag((r,u))=\set{x,y}$, this edge is first   subdivided by adding a new pseudo-P node $c$
  between $r$ and $u$ (which doesn't violate any of the invariants),
  and then, the edge $(x,y)$ is added to $\Gamma(c)$, turning it into a
  real $P$ node.

  If $r=u$ is a real $R$ node and $\Gamma(r)$ already contains an edge
  $(x,y)$, then a new pseudo-$P$ root $r'$ is added containing
  $(x,y)$, and the new duplicate $(x,y)$ edge is added to $\Gamma(r')$
  turning it into a real $P$ node.

  Finally, if $r=u$ is an $R$ node that does not contain an $(x,y)$
  edge, or a $P$ node, $(x,y)$ is simply added to $\Gamma(r)$.

  Thus in all cases, the number of merges and splits is at most $2$.
\end{proof}

\subsection{Incremental SPQR-trees of biconnected graphs}\label{sec:incr-spqr}

\todo[inline]{Improve, include references to theorems/lemmas/observations}
In conclusion, this section has presented a data structure for updating an implicit representation of the SPQR-tree of a biconnected graph subject to edge insertion and undo edge insertion. 

Each update gives rise to $O(\log n)$ changes in the implicit representation of the SPQR-tree, and the update is supported in $O(\log^2 n)$ time. 

This structure may be used to answer triconnectivity of an incremental biconnected graph:
\begin{proof}[Proof of Theorem~\ref{thm:incr3con} for biconnected graphs]
  Let all vertices have weight $1$, and build the implicit representation of the SPQR-tree. Upon query, given vertices $a,b$ of the graph, find their critical path $m(a,b)$. If $m(a,b)$ is a single $R$ or $P$ node, the answer is yes, otherwise, the answer is no. 	

  Since the time for finding $m(x,y)$ is $O(\log n)$ and insertions are handled in $O(\log^2 n)$ time, the data structure has an update time of $O(\log^2 n)$ and a query time of $O(\log n)$.
\end{proof}
  
\section{Dynamic BC trees}\label{sec:dynbc}
In the previous section, we gave an incremental algorithm for maintaining the SPQR-tree of a biconnected graph. Similarly, we here present an algorithm for maintaining the BC-tree of a connected component. Many of the challenges for SPQR-trees appear in a simpler form for BC-trees, and we address them in as similar a way as we can: To accommodate edge-insertions, we maintain a path decomposition, where we precontract solid paths, and ensure certain node types have been pre-split accordingly. 

The main technical difference is that while the solid path contractions in SPQR-trees would approximately correspond to inserted edges, here, in the BC-tree, we actually insert struts (see Section~\ref{sec:dynbc-precontract}) that are dummy edges covering each heavy path, such that the union $G_s$ of our graph $G$ and the struts has a BC-tree of height $O(\log n)$ corresponding only to the light edges.

Now, for each block of $G_s$, we maintain its SPQR-tree.
These SPQR-trees will have the elegant property that they are rooted in the S-node containing the strut, and that each edge from the root separates a block in $G$ via its two cutvertex neighbours on the solid path. 
Maintaining these struts leads to a different flavour of algorithmic challenge: Namely, now any change in the path decomposition in the BC-tree is reflected as a change of struts which will lead to links, cuts, and local changes in SPQR-trees. 

We show (in Section~\ref{sec:SPQRBC}) how to accommodate the changes in the SPQR-tree driven by the changes in the BC-tree simply by putting together some of the atomic operations such as split, link, cut, and contract, that already are supported by the structure.

\subsection{Strict and relaxed BC trees}
\begin{definition}
  Let $x$ be a vertex in a connected loopless multigraph
  $G$. Then $x$ is an \emph{articulation point} if $G-\set{x}$ is not
  connected.
\end{definition}
\begin{definition}
  A (strict) BC-tree for a connected loopless multigraph $G=(V,E)$ with at
  least $1$ edge is a tree with nodes labelled $B$ and $C$, where each
  node $v$ has an associated \emph{skeleton graph} $\Gamma(v)$ with
  the following properties:
  \begin{itemize}
  \item For every node $v$ in the BC-tree, $V(\Gamma(v))\subseteq V$.
  \item For every edge $e\in E$ there is a unique node $v=b(e)$ in the
    BC-tree such that $e\in E(\Gamma(v))$.
  \item For every edge $(u,v)$ in the BC-tree, $V(\Gamma(u))\cap
    V(\Gamma(v))\neq\emptyset$ and either $u$ or $v$ is a $C$ node.
  \item If $v$ is a $B$ node, $\Gamma(v)$ is either a single edge or a
    biconnected graph.
  \item If $v$ is a $C$ node, $\Gamma(v)$ consists of a single vertex,
    which is an articulation point in $G$.
  \item No two $B$ nodes are neighbors.
  \item No two $C$ nodes are neighbors.
  \end{itemize}
\end{definition}
The BC-tree for a connected graph is unique. The
(skeleton graphs associated with) the $B$ nodes are sometimes referred
to as $G$'s biconnected components.

In this paper, we use the term \emph{relaxed} BC tree to denote a tree that satisfies all but the last condition.  Unlike the strict BC tree, the relaxed BC tree is not unique.

\subsection{Changes during insert}\label{sec:dynbc-changes}
  When inserting edge $(x,y)$, the change to the BC tree happens along
  the \emph{critical path} $m(x,y)$ defined as:
  \begin{itemize}
  \item If no node contains both $x$ and $y$, then $m(x,y)$ is the
    unique shortest path $u\cdots v$ in the BC tree such that $u$
    contains (or is) $x$ and $v$ contains (or is) $y$.
  \item Otherwise exactly one node $u$ contains both $x$ and $y$, and
    $m(x,y)$ is the path consisting of that node.
  \end{itemize}
  Note that the path is considererd to be undirected, so $m(x,y)=m(y,x)$.
  Note also that in each case, both ends of $m(x,y)$ are $B$-nodes.
  This notion of a critical path is implicit in~\cite{DBLP:journals/siamcomp/PoutreW98}.

  The actual change caused by inserting $(x,y)$ consists of first
  \emph{splitting} each $C$ node $u$ of degree $d(u)\geq 3$ that is
  internal to $m(x,y)$ into a pair of adjacent nodes $u_p,u_c$ where
  \begin{itemize}
  \item $u_p$ inherits the neighbors of $u$ that are on the path, so $d(u_p)=3$.
  \item $u_c$ inherits the remaining neighbors of $u$, so $d(u_c)=d(u)-1$.
  \end{itemize}
  After the splits, the modified path is contracted into a single $B$ node.

Observe that given $x$ and $y$, we may determine their critical path $m(x,y)$ in $O(\log n)$ time.

\subsection{Heavy paths in an unweighted BC tree}
We want to use a version of our biased dynamic trees to represent our
BC trees.  For this, we could use the scheme defined in
Section~\ref{sec:heavyweighted}, but since our underlying graph is
unweighted, we can use a simpler way of assigning weights.

We will be working with relaxed BC trees where each $C$ node has at
most one $C$ node among its neighbors.  For each vertex $x\in V(G)$ in
such a BC tree let $b(x)$ be either the (unique, if it exists) $C$
node $u$ furthest from the root such that $x\in\bag(u)$, and if no
such node exists $b(x)$ is the unique $B$ node $u$ such that
$x\in\bag(u)$.  Then, just as in Section~\ref{sec:heavyweighted} we
can define
\begin{align*}
  b^{-1}(v) &:= \set{x\in\bag(v)\suchthat b(x)=v}
  \\
  w(v) &:= \sum_{x\in b^{-1}(v)} w(x) = \abs{b^{-1}(v)}
\end{align*}

The resulting weights are $3$-positive because in every path
consisting of $3$ nodes of degree $2$, at least one is a $C$ node of
positive weight.  Thus we can use these weights to define our heavy
path decomposition of a relaxed BC tree.  Furthermore, if we ensure
that every edge in the BC tree where both endpoints are $C$ nodes is
dashed, then the weights are not changed by a reverse$(r)$ operation,
so we can maintain these weights explicitly.

\subsection{SPQR tree changes caused by BC tree changes}\label{sec:SPQRBC}
The changes in the path decomposition in the BC-tree gives rise to changes in the SPQR-trees. Namely, each solid path in the BC-tree is precontracted by the insertion of a strut, and is thus represented as SPQR-tree rooted in an S-node with all SPQR-trees of the blocks along the solid path as children.
When an edge in the BC-tree changes status from solid to dashed, the artificial $S$-node containing its strut gets split into three nodes; we call this the \emph{sever} operation. Its opposite, for when a dashed edge becomes solid, we call \emph{meld} (see Figure~\ref{fig:meldsever}).
Meld can be implemented in two steps by performing two links followed by a contraction in case S-nodes were linked. Reversely, sever can be implemented using split and cut.
\begin{figure}[h]\begin{center}
\end{center}
  \centering
  \begin{tikzpicture}[
  treenode/.style={
    circle,
    draw,
    minimum width=.8cm,
  },
  ]
  \begin{scope}
    \node[treenode] (s0) at (.75,.5) {$S$};
    \node[treenode] (l0) at (0,-.5) {};
    \node[treenode] (r0) at (1.5,-.5) {};

    \draw (s0) --++ (-90:.75);
    \foreach \a in {105,125,145}{
      \draw (l0) --++ (-\a:.75);
    };
    \foreach \a in {35,55,75}{
      \draw (r0) --++ (-\a:.75);
    };
  \end{scope}

  \begin{scope}[shift={(4,0)}]
    \node[treenode] (s1) at (.75,.5) {$S$};
    \node[treenode] (l1) at (0,-.5) {};
    \node[treenode] (r1) at (1.5,-.5) {};

    \draw (s1) --++ (-90:.75);
    \draw (s1) -- (l1);
    \draw (s1) -- (r1);
    \foreach \a in {105,125,145}{
      \draw (l1) --++ (-\a:.75);
    };
    \foreach \a in {35,55,75}{
      \draw (r1) --++ (-\a:.75);
    };
  \end{scope}

  \begin{scope}[shift={(8,0)}]
    \node[treenode] (s2) at (.75,0) {$S$};
    \foreach \a in {35,55,75,90,105,125,145}{
      \draw (s2) --++ (-\a:.75);
    };
  \end{scope}

  \draw[-stealth] (1,1) to[bend left=15] node[midway,above] {meld} (8.5,1);

  \draw[-stealth] (2.25,0.25) to[bend left=15] node[midway,above] {link} (3.25,0.25);
  \draw[-stealth] (3.25,0.125) to[bend left=15] node[midway,below] {cut} (2.25,0.125);

  \draw[-stealth] (6.25,0.25) to[bend left=15] node[midway,above] {merge} (7.25,0.25);
  \draw[-stealth] (7.25,0.125) to[bend left=15] node[midway,below] {split} (6.25,0.125);

  \draw[-stealth] (8.5,-1.5) to[bend left=15] node[midway,below] {sever} (1,-1.5);

\end{tikzpicture}
 \caption{The \emph{meld} and \emph{sever} operations.\label{fig:meldsever}}
\end{figure}
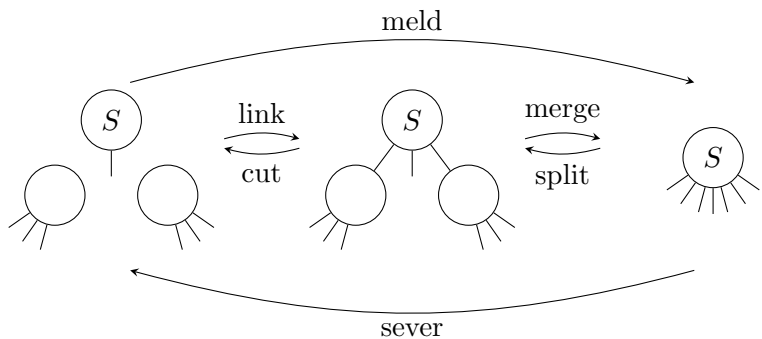

Recall that the BC-tree is bipartite, thus, an edge that changes status between solid and dashed is incident to a block. It is the neighbours of this block along the solid path that toggle their status of being endpoints of struts.

Stated in the language of Section~\ref{sec:dynspqr-changes}, if the heavy BC-path that has the strut $(c,d)$ is severed in the block whose neighbours on the heavy paths are $a$ and $b$, its SPQR-tree is updated as follows:
\begin{itemize}
\item The $S$ node at its root is split into at most $3$ pieces with respect to the edges $(c,d)$ and $(a,b)$. \item The edges connecting these at most $3$ resulting $S$ nodes are cut.
\item The $S$ node $x$ with the $\le 4$-cycle containing both $(a,b)$ and $(c,d)$ is now the root of a tree and has only one child $y$. The edge $(x,y)$ is cut, and $y$ becomes the root of its tree.
\end{itemize}

Meld is simply the opposite operation. Let $a$ and $b$ denote the cutvertices that are neighbours of the block $m$ incident to the edge that changes status from dashed to solid. Let $(c,d)$ be the endpoints of the resulting strut. Let $T_1$ and $T_2$ denote the SPQR-trees of the heavy paths to be melded, and denote by $T_m$ the SPQR-tree over $m$.
\begin{itemize}
	\item Construct an $S$ node $x$ consisting of a $\le 4$-cycle $a,c,d,b$, and link $x$ to the root of $T_m$ via the edge $(a,b)$.
	\item Link $x$ to $T_1$ and $T_2$ via the edges $(a,c)$ and $(b,d)$, respectively.
	\item Contract the edges $(a,c)$ and $(b,d)$. \end{itemize}

In the descriptions above, we did not fully address the case where $a=c$ or $b=d$, however, either of these would correspond to $x$ being a $3$-cycle, and thus, only differ in that one link and one contract operation, or one split and one cut operation, is omitted.

Formally, we define two new operations
\begin{description}

\item[meld$(T_1,\ldots,T_k; x_0,\ldots,x_k)$:] Where each $T_i$ is a
  proper heavy path decomposition of a SPQR tree with root $r_i$ such
  that $(x_{i-1},x_i)$ is a non-virtual edge in $\Gamma(r_i)$.

  Combine the trees into a single new tree as follows: First construct
  an $S$ node $r$ where $\Gamma(r)$ is the cycle $x_0\cdots x_k$. This
  node by itself is an $r$-exposed SPQR tree. Then use
  link-exposed$(r_i,r)$ to link each $T_i$ to $r$. Finally call
  conceil$(r)$ to make it a proper heavy path decomposed SPQR tree.

\item[sever$(T; x_0,\ldots,x_k)$:] Where $T$ is a proper heavy path
  decomposition of a SPQR tree whose root $r$ is an $S$ node containing
  $x_0,\ldots,x_k$ in that order along the cycle.

  Split the tree into $k$ smaller trees as follows: First call
  expose$(r)$. Then for each $i\in\set{1,\ldots,k}$ call
  select$(r,\set{x_{i-1},x_i})$, letting $r_i$ be the new child that
  this creates, and then call cut-exposed$(r_i)$. Finally,
  destroy the (now isolated) node $r$.

\end{description}

\begin{lemma}\label{lem:ps-meld}
  Given vertices $x_0,\ldots,x_k\in V[G]$ and proper pre-split SPQR
  trees $T_1,\ldots,T_k$, maintaining the pre-split invariants during
  meld$(T_1,\ldots,T_k; x_0,\ldots,x_k)$ takes at most $3+k$ splits
  and merges.
\end{lemma}
\begin{proof}
  By Lemma~\ref{lem:ps-link-exposed}, for each $i\in\set{1,\ldots,k}$
  the violations caused by link-exposed$(r_i,r)$ are repaired using at
  most one merge or split.
  In total, at most $k$ splits or merges suffice to restore the
  invariants before the final conceil.
  Since $r$ is the root, $\ell(r)=0$ and so by
  Lemma~\ref{lem:ps-expose-conceil} fixing the invariants during the
  conceil$(r)$ takes at most $3$ merge and split operations.
\end{proof}

\begin{lemma}\label{lem:ps-sever}
  Given a proper pre-split SPQR tree $T_r$ and vertices
  $x_0,\ldots,x_k\in \bag(r)$, maintaining the pre-split invariants
  during a sever$(T_r; x_0,\ldots,x_k)$ takes at most $3+k$ splits and
  merges.
\end{lemma}
\begin{proof}
  Since $r$ is the root, $\ell(r)=0$ and so by
  Lemma~\ref{lem:ps-expose-conceil} fixing the invariants during the
  expose$(r)$ takes at most $3$ merge and split operations.
  By Lemma~\ref{lem:ps-select-restorecount}, restoring the invariants
  after each select$(r_i, \set{x_{i-1},x_i})$ costs at most one
  split. Note that if this is necessary, the corresponding
  $r_i$ in the subsequent cut-exposed$(r_i)$ is a secondary
  node.
  By Lemma~\ref{lem:ps-cut-exposed}, for each $r_i$ the violations
  caused by cut-exposed$(r_i)$ are repaired by at most one merge or
  split if $r_i$ is a primary node, and None otherwise. Thus for each
  pair of a select$(r_i, \set{x_{i-1},x_i})$ and a cut-exposed$(r_i)$
  the cost of restoring the invariants is at most $1$ merge or split.
\end{proof}

\subsection{Pre-splitting}
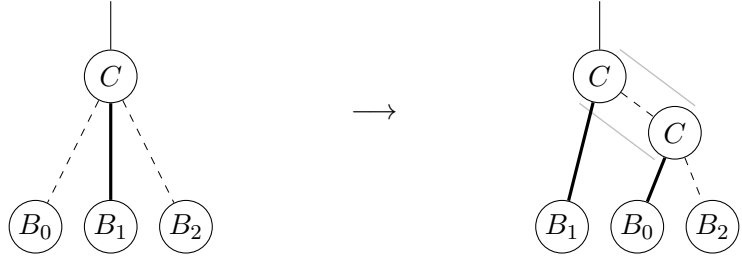
\begin{figure}[H]
	\center
	\begin{tikzpicture}[
    spqr-node/.style={
      draw,
      circle,
      fill=white,
      minimum size=7mm,
      inner sep=0pt,
    },
    spqr-edge/.style={
      draw,
    },
    dashed spqr-edge/.style={
      spqr-edge,
      very thin,
      dashed
    },
    solid spqr-edge/.style={
      spqr-edge,
      very thick,
    }
  ]
  \begin{scope}
    \coordinate (p) at (1,3);
    \node[spqr-node] (p0) at (1,2) {$C$};
    \node[spqr-node] (l0) at (0,0) {$B_0$};
    \node[spqr-node] (l1) at (1,0) {$B_1$};
    \node[spqr-node] (l2) at (2,0) {$B_2$};
    \draw[spqr-edge] (p) -- (p0);
    \draw[dashed spqr-edge] (p0) -- (l0);
    \draw[solid spqr-edge] (p0) -- (l1);
    \draw[dashed spqr-edge] (p0) -- (l2);
  \end{scope}
  \begin{scope}[shift={(4.5,0)}]
    \node at (0,1.5) {$\longrightarrow$};
  \end{scope}
  \begin{scope}[shift={(7,0)}]
    \coordinate (p) at (.5,3);
    \node[spqr-node] (p0) at (.5,2) {$C$};
    \node[spqr-node] (p1) at (1.5,1.25) {$C$};
    \node[spqr-node] (l1) at (0,0) {$B_1$};
    \node[spqr-node] (l0) at (1,0) {$B_0$};
    \node[spqr-node] (l2) at (2,0) {$B_2$};
    \draw[spqr-edge] (p) -- (p0);
    \draw[solid spqr-edge] (p0) -- (l1);
    \draw[dashed spqr-edge] (p0) -- (p1);
    \draw[solid spqr-edge] (p1) -- (l0);
    \draw[dashed spqr-edge] (p1) -- (l2);
    \begin{pgfonlayer}{background}
      \filldraw[line width=26,line join=round,gray!50](p0.center)--(p1.center)--cycle;
      \filldraw[line width=25,line join=round,white](p0.center)--(p1.center)--cycle;
    \end{pgfonlayer}
  \end{scope}
\end{tikzpicture}
 	\caption{\label{fig:pre-splitting-C}Pre-splitting a $C$ node creates
		a primary and at most $1$ secondary child node. The primary
		inherits the solid child, the secondary inherits the remaining
		children and chooses its heaviest child as solid.}
\end{figure}

Since inserting an edge naively would require $\OO(n)$ vertex splits,
we instead use a variation of the pre-splitting ideas from
Section~\ref{sec:dynspqr-presplit} to maintain a \emph{pre-split}
BC-tree.  In fact if we treat $B$ nodes as $R$ nodes, and $C$ nodes as
$P$ nodes, we can use almost exactly the same definitions and get the
same result.
\begin{definition}
  A node $v$ in a rooted relaxed BC tree is \emph{primary} if either
  $v$ is the root; or $v$ is a $B$ node; or $v$ is a $C$ node whose
  parent is a $B$ node. All other nodes are called \emph{secondary}.
\end{definition}
In any relaxed BC tree $T$ there is a one-to-one correspondence
between the primary nodes in $T$ and the nodes in the corresponding
strict BC tree.
\begin{definition}
  Define a node $v$ in a BC tree to be \emph{splittable}, if $v$ is a
  $C$ node with both a parent and a solid child, and at least $2$
  primary children. All other nodes are \emph{unsplittable}. In
  particular, the root, the leaves, and the $B$ nodes are
  unsplittable.
\end{definition}
\begin{definition}
  Define a node to be \emph{mergeable} if it has a secondary
  child. All other nodes are \emph{unmergeable}.
\end{definition}
In a strict BC tree many nodes may be splittable, but no node is
mergeable. We will use the following operations.
\begin{description}
\item[split$(v)$:] Where $v$ is a splittable node. Replace the group
  of all the dashed children of $v$ with a single new dashed secondary
  child node $c$, inheriting the edges that it replaces.  If $c$ has a
  heavy child $h$, make $(h,c)$ solid.
\item[merge$(v)$:] Where $v$ is a mergeable node. For each secondary
  child $c$: if $(c,v)$ is dashed and $c$ has a solid child $c_2$,
  make $(c_2,c)$ dashed; then contract edge $(c,v)$.
\end{description}
\begin{lemma}\label{lem:bc-ps-relaxed-invariants}
  If we start with a heavy path decomposition of a strict BC tree, and
  only split or merge primary nodes, the following innvariants hold:
  \begin{enumerate}[label={I'-\arabic*.}, ref={I'-\arabic*}, series=bc-presplitting-invariants]
  \item\label{inv:bc-ps-not-both-split-merge} No node is both
    splittable and mergeable.
  \item\label{inv:bc-ps-secondary} Each secondary node has at least
    $2$ children, and is a dashed child of a primary node which has
    both a parent and a solid child.
  \end{enumerate}
  Furthermore, any heavy path decomposition of a relaxed BC tree
  where these inveriants hold can be constructed from the
  corresponding strict BC tree by splitting a unique subset of the
  splittable primary nodes in any order.
\end{lemma}
\begin{proof}
  Completely analoguous to Lemma~\ref{lem:ps-relaxed-invariants}.
\end{proof}
\begin{definition}
  Define the \emph{solid degree} of a node $v$ to be the number solid
  edges incident to $v$. Define an edge $(c,p(c))$ to be \emph{stable}
  if it is dashed and either
  \begin{itemize}
  \item $p(c)$ is a $B$ node; or
  \item $p(c)$ is a $C$ node with solid degree $2$ and $c$ has no
    dashed siblings.
  \end{itemize}
  All other dashed edges are \emph{unstable}. In particular, every
  dashed child edge of a $C$ node that is at the end of a solid path
  is unstable, as are all child edges of a splittable node that would
  get moved during a split.
\end{definition}
\begin{observation}
  If the edge $(c,p(c))$ is stable, then contracting the solid path
  containing $p(c)$ into a $B$ node during an edge insertion does not
  change this.
\end{observation}
\begin{lemma}
  If splitting or merging the node $p$ does not give its dashed child
  $c$ a new parent, the $(c,p)$ is unstable after the operation if and
  only if it was unstable before the operation.
\end{lemma}
\begin{proof}
  Completely analoguous to Lemma~\ref{lem:split-merge-nonmoving-child}.
\end{proof}

\begin{lemma}\label{lem:bc-split-moving-child}
  If splitting the node $p$ gives its dashed child $c$ a new parent
  $v$, then $(c,p)$ was unstable before the split, $(c,v)$ is either
  solid or unstable after the split, and $\Gamma(v)$ has at least $3$
  edges.
\end{lemma}
\begin{proof}
  Completely analoguous to Lemma~\ref{lem:split-moving-child}.
\end{proof}
\begin{lemma}\label{lem:bc-merge-moving-child}
  If merging node $p$ contracts the dashed edge $(v,p)$ where $v$ has
  a child $c$ and at least $2$ children, then $(c,v)$ was
  solid or unstable before the merge, and $(c,p)$ is unstable after
  the merge.
\end{lemma}
\begin{proof}
  Completely analoguous to Lemma~\ref{lem:merge-moving-child}.
\end{proof}
Thus, if the edge $(c,p(c))$ is solid or unstable then after a split
or merge that changes $p(c)$, the new parent edge of $c$ is still
either solid or unstable (but not necessarily the same one as before).
Given these properties we can now state
\begin{lemma}\label{lem:bc-ps-unique-relaxed}
  Given a $u$-exposed (or proper) heavy path decomposition of a strict
  BC tree $T$ with root $r$, there is a unique corresponding relaxed
  BC tree $T'$ that satisfies the inveriants of
  Lemma~\ref{lem:bc-ps-relaxed-invariants} as well as the following:
  \begin{enumerate}[resume*=bc-presplitting-invariants]
  \item\label{inv:bc-ps-split-rule} The root, and each primary node with
    a solid or unstable parent edge is not splittable.
  \item\label{inv:bc-ps-merge-rule} A primary node with a stable
    parent edge is not mergeable.
  \end{enumerate}
  We call $T'$ the \emph{pre-split} BC tree corresponding to
  $T$.
\end{lemma}
\begin{proof}
  Completely analoguous to Lemma~\ref{lem:bc-ps-unique-relaxed}.
\end{proof}

\begin{lemma}\label{lem:bc-ps-edge-change-restorecount}
  In the pre-split BC tree, if $c$ and $p=p(c)$ are both primary,
  then changing edge $(c,p)$ from solid to dashed or vice versa may
  violate the invariants at $p$ or $c$ or at most $3$ of their
  children, but nowhere else.  The invariants can be restored by at
  most $3$ split or merge operations.
\end{lemma}
\begin{proof}
  Completely analoguous to Lemma~\ref{lem:bc-ps-edge-change-restorecount}.
\end{proof}

Just like for pre-split SPQR trees, to handle expose$(u)$ of a primary
node $u$ in a proper heavy path decomposition of the pre-split BC tree
with root $r$, or conceil$(r)$ in a $u$-exposed heavy path
decomposition of the pre-split BC tree with root $r$ and primary node
$u$, we need to slightly change the implementation of splice$(v)$ and
slice$(v)$ as follows:
\begin{description}
\item[splice$(v)$:] Where $v$ is a primary node that is a dashed light
  child.
  If $p(v)$ is a secondary node, let $p=p(p(v))$, else let
  $p=p(v)$. Now $p$ is a primary node. If $p$ has a solid child then
  make it dashed and do the necessary splits and merges to restore
  the invariants. Now $p$ is the parent of $v$ and $(v,p)$ is
  dashed. Make $(v,p)$ solid, and do the necesary splits and merges
  to restore the invariants.

\item[slice$(v)$:] Where $v$ is a primary node that is a solid light
  child.
  Let $p=p(v)$, then $p$ is primary and $(v,p)$ is solid. Make $(v,p)$
  dashed and do the necessary splits and merges to restore the
  invariants. Now $p$ has no solid children. If $p$ has a heavy child
  $h$, make $(h,p)$ solid and do the necessary splits and merges to
  restore the invariants.
\end{description}

\begin{lemma}\label{lem:bc-ps-expose-conceil}
  An expose$(u)$ where $u$ is a primary node in a proper heavy path
  decomposition of the pre-split BC tree with root $r$, or a
  conceil$(r)$ in a $u$-exposed heavy path decomposition of the
  pre-split SPQR tree with root $r$ each costs at most $6\ell(u)+3$
  splits or merges, in addition to the at most $2\ell(u)+1$ edges
  changing from solid to dashed or vice versa, where $\ell(u)$ is the
  light depth of $u$ in the corresponding strict SPQR tree.
\end{lemma}
\begin{proof}
  Completely analoguous to Lemma~\ref{lem:ps-expose-conceil}.
\end{proof}

\begin{lemma}\label{lem:bc-ps-link-exposed}
  Given a proper pre-split BC tree $T_v$, and an $u$-exposed pre-split
  BC tree $T_r$, where $x$ is a vertex edge in both
  $\Gamma(v)$ and $\Gamma(u)$, restoring the invariants after a
  link-exposed$(v,u)$ takes at most one split or merge operation.
\end{lemma}
\begin{proof}
  Completely analoguous to Lemma~\ref{lem:ps-link-exposed}.
\end{proof}

\begin{lemma}\label{lem:bc-ps-cut-exposed}
  Given a $u$-exposed pre-split BC tree $T_r$ and a child $v$ of
  $u$, restoring the invariants after a cut-exposed$(v)$ takes at most
  one merge operation if $v$ is a primary node, and no violations of
  the invariants occur otherwise.
\end{lemma}
\begin{proof}
  Completely analoguous to Lemma~\ref{lem:ps-cut-exposed}.
\end{proof}

\begin{lemma}\label{lem:bc-ps-reverse}
  In the pre-split BC tree with root $r$, reverse$(r)$ does not
  violate any of the invariants.
\end{lemma}
\begin{proof}
  Completely analoguous to Lemma~\ref{lem:ps-reverse}.
\end{proof}

Combining the previous Lemmas whose proofs were all completely analogous to the corresponding proofs for SPQR-trees, we now obtain a Theorem completely analogous to Theorem~\ref{thm:SPQRthm} for SPQR trees:

\begin{theorem}
  The number of splits and merges needed to maintain a pre-split BC
  tree under expose, conceil, reverse, link-exposed, and cut-exposed
  is proportional to the number of edges changing from solid to dashed
  or vice versa during the same operation in the strict BC-tree.
\end{theorem}

\subsection{Pre-contract}\label{sec:dynbc-precontract}

Given a path decomposition with an appropriate pre-splitting of the affected vertices, we precontract heavy paths simply by adding struts that cover them. Most aspects of maintaining this path decomposition dynamically under insert and undo insert is completely analogous to the SPQR-case. However, the difference lies in how changes in the path decomposition drive changes in the SPQR-forest as described in Section~\ref{sec:SPQRBC}.

Let us first describe exactly which struts are inserted because of the BC-tree. Consider a heavy path. For each endpoint of the path, we choose a representing vertex. If either endpoint is an articulation point, it will be representing itself. If the root-nearest node is a B-node and has a dashed parent $p$, it will be represented by $p$. If the root-furthest node is a B-node and has children, it will be represented by its heaviest child. Finally, in all remaining cases, the B-node chooses any arbitrary vertex different from its neighbour along the heavy path.

To implement that a solid edge turns dashed, there are two cases. In case its root-nearest node is a block, one first performs a sever operation freeing the involved block, and then performs a meld ensuring that the heaviest child of the block is endpoint of the strut. In case its root-nearest node is an articulation point, this articulation point becomes the endpoint of both new struts.

Note that with the rule in place saying that we always choose the heaviest child cutvertex in place, we are always able to pay the cost of expose.

\subsection{Incremental BC/SPQR-trees and triconnectivity in $\OO(\log^3 n)$ insertion time}\label{sec:incr-bc}

\todo[inline]{Improve, include references to theorems/lemmas/observations}

We now have a structure for maintaining the BC- and SPQR trees of an incremental graph, and answering triconnectivity queries:
\begin{proof}[Proof of Theorem~\ref{thm:incr3con}]
  Maintain the pre-contracted path decomposed BC-tree, and the SPQR-trees over its (pre-contracted) blocks. Let $a,b$ be the vertices of the graph that a query specifies. First, we find $m(a,b)$ in the BC-tree. If $m(a,b)$ is not a single (precontracted) block, the answer is no. Otherwise, we maintain the SPQR-tree of the precontracted block. In the SPQR-tree, we can find the critical path $m(a,b)$, and if this is an R-node or P-node that is not the child of the root, the answer is yes. If it is not an R- or P-node, the answer is no. Finally, if it is an R- or P-node that is a child of the root, perform a sever operation, possibly leading to an un-contract path, and answer the question within the block.
\end{proof}

\section{Embeddings of biconnected planar graphs}\label{sec:dynEmb}

In this section we show how we can use the dynamic SPQR trees to
control the embedding of a biconnected graph under edge insertions
with backtracking.  First, we define an abstract scheme for describing
an embedding, given a relaxed SPQR tree.  The crucial part of this
description is the definition of certain \emph{flip-bits}, where toggling a flip-bit corresponds to performing a separation-flip in the embedded graph (see Section~\ref{sec:SPQRembedding}). We go on to define \emph{good embeddings} as those where almost all flip-bits on the heavy paths are set in a way that is favourable of an edge insertion covering it (see Section~\ref{sec:heavySPQRembedding}). As an example of a good embedding, we give the canonical embedding (Section~\ref{sec:canonical}), by making arbitrary but  consistent choices whenever the embedding is not uniquely determined by our definition of good embeddings.

As a consequence, we obtain an algorithm for incremental planarity testing of a biconnected graph with worst-case polylog update time (see Section~\ref{sec:incrPlanrBicon}).

\subsection{Specifying an embedding in an SPQR tree}\label{sec:SPQRembedding}
Given a biconnected planar graph $G=(V,E)$, and any relaxed SPQR tree
for it, we can define a unique embedding of $G$ by choosing:
\begin{itemize}
\item A \emph{root} for the SPQR tree.
\item A \emph{flip bit} for each SPQR edge.
\item A \emph{local embedding} of $\Gamma(v)$ for each SPQR node
  $v$. Note: For $S$,
  $P$, and $R$ nodes, respectively, the number of choices are $1$, $(d(v)-1)!$, and $2$.
\end{itemize}

Given these choices, we can now construct a unique embedding of $G$.
For each node $v$ in the SPQR tree, define $G_v$ as the subgraph of
$G$ induced by the vertices in $T_v$, possibly together with a virtual
edge representing the rest of the graph.  Thus $G_r=G$.  In bottom-up
order, construct $\emb(G_v)$ from the embeddings
$\emb(G_{c_1}),\ldots,\emb(G_{c_k})$ of its children as follows:
\begin{itemize}
\item If $v$ is a leaf, $\emb(G_v)$ is just the local embedding of
  $\Gamma(v)$.
\item Otherwise, take each child $c_i$, construct $\emb(G_{c_i})$
  recursively, flip it if the flip-bit for $(v,c_i)$ is set, and
  choose an outer face that contains the virtual edge corresponding to
  $(v,c_i)$, and delete that virtual edge to get a plane graph
  $C_i$. Note that we get the same graph no matter which of the two
  faces adjacent to the virtual edge we picked as outer face, since
  they get merged anyway when the edge is deleted. Then take the local
  embedding of $\Gamma(v)$ and replace each virtual edge corresponding
  to a child edge $(v,c_i)$ with the corresponding $C_i$.
\end{itemize}

Note that not all combinations of these give distinct embeddings,
e.g.:
\begin{itemize}

\item Changing the root either does nothing or flips the whole
  embedding.

\item For an $S$-node $v$, changing all the incident flip-bits
  flips the whole embedding if $v$ is the root, and gives exactly the
  same embedding otherwise

\item For a $P$-node or $R$-node, changing all the incident flip-bits
  and flipping the local embedding of $\Gamma(v)$ flips the whole
  embedding if $v$ is the root, and gives exactly the same embedding
  otherwise.

\end{itemize}

This way of specifying an embedding has the nice property that
toggling a flip bit or moving a single edge to a new position in the
local embedding of a $P$ node each correspond to a \emph{separation
  flip} in the corresponding embedding $G$.  Thus, if we can bound the
number of these changes during an edge insertion, we bound the number
of separation flips.

\subsection{Good embeddings}\label{sec:heavySPQRembedding}
For a biconnected planar graph $G=(V,E)$, we will define our class of
\emph{good} embeddings $\EmbGood(G)$ based on the family of pre-split
SPQR trees for $G$, by restricting the values of the flip-bits for (a
subset of) the solid edges.  The canonical embedding is then defined
by uniquely specifying the local embedding of each $P$ node and all
the remaining flip-bits.

Let $x,y\in V$ be distinct vertices in the biconnected planar graph
$G=(V,E)$.  Let $T$ be the unique pre-split SPQR tree for $G$ that is
$u$-exposed and has root $r$ such that $m(x,y)=r\cdots u$, and where
for $v\in\set{r,u}$, if $v$ is an $S$ node, then
$s(v)=\bag(v)\cap\set{x,y}$.

We will define a set $\EmbGood(G; x,y)$ of good embeddings of $G$ by
defining flip-bits for (most of) the solid edges in $T$.  Then define
\begin{align*}
  \EmbGood(G) &:= \bigcup_{x,y\in V}\EmbGood(G; x,y)
\end{align*}

Our goal is to define $\EmbGood(G; x,y)$ such that if $G\cup (x,y)$ is planar then any $H\in\EmbGood(G; x,y)$ admits inserting $(x,y)$. Or in other words:
\begin{lemma}\label{lem:biconn-embgood-insert}
  For any biconnected planar graph $G$ with $n$ vertices, and any edge
  $(x,y)\in G$: $\flipdist(\EmbGood(G-(x,y); x,y), \EmbGood(G;
  x,y))\in\OO(1)$.
\end{lemma}
By additionally making the flip-bits for each solid edge depend only
on its closest few neighbors on the solid path, we then immediately
get:
\begin{lemma}\label{lem:biconn-embgood-mpath}
  For any biconnected planar graph $G$ with $n$ vertices, and any two
  pairs of distinct vertices $x_1,y_1$ and $x_2,y_2$:
  $\flipdist(\EmbGood(G; x_1,y_1), \EmbGood(G; x_2,y_2))\in \OO(\log
  n)$.
\end{lemma}
And combining Lemma~\ref{lem:biconn-embgood-insert}
and~\ref{lem:biconn-embgood-mpath} then gives:
\begin{theorem}\label{thm:biconn-embgood-flipdist}
  For any biconnected planar graph $G$ with $n$ vertices, and any edge
  $e\in G$: $\flipdist(\EmbGood(G-e),\EmbGood(G))\in\OO(\log n)$
\end{theorem}
\begin{proof}
  Let $e=(x,y)$. Then by the triangle inequality
  \begin{align*}
    \flipdist(\EmbGood(G-e),\EmbGood(G))
    &\leq \flipdist(\EmbGood(G-e),\EmbGood(G-e;x,y))
    \\
    &\quad + \flipdist(\EmbGood(G-e;x,y),\EmbGood(G;x,y))
    \\
    &\quad + \flipdist(\EmbGood(G;x,y),\EmbGood(G))
  \end{align*}
  By Lemma~\ref{lem:biconn-embgood-mpath} the first and last of these
  terms is $\OO(\log n)$ and by Lemma~\ref{lem:biconn-embgood-insert}
  the middle term is $\OO(1)$.
\end{proof}

Let $T$ be a pre-split SPQR tree that is $u$-exposed and has root
$r$ where $r\cdots u=m(x,y)$.
\begin{definition}\label{def:gammaM}
  For each solid path $M$ in $T$, define the \emph{relevant part} of
  $M$ as the maximal subpath $M_r$ of $M$ that does not end in a
  $P$-node.  If $M$ consists only of a $P$-node, $M_r$ is the empty
  path.  For any (nonempty) path $H$, let $\Gamma(H)$ denote the graph
  ``glued'' from the skeleton graphs on $H$.
\end{definition}
We will define flip-bits for $M_r$ by defining a unique embedding for
$\Gamma(M_r)$ and letting the flip-bits be exactly the ones required to
describe that embedding.

\begin{definition}
  For each $R$ node that is internal to a solid path, we say that it
  is \emph{happy} if the two virtual edges corresponding to the path
  are in the same face (note that this will be the same in all
  embeddings), and \emph{cross} otherwise.  All other nodes are
  happy if they are internal to a solid path, and cross if they are
  isolated or at the end of a solid path.
\end{definition}

\begin{definition}
  Define a \emph{run} on the relevant part of a solid path as any
  maximal subpath with at least one edge and no internal cross
  nodes. Note that a run can not end in a $P$-node.  For every run
  $A=v_1\ldots v_2$  , define the \emph{strut}\footnote{From the
    dictionary definition of a \emph{strut}: a rod or bar forming part
    of a framework and designed to resist compression. ``a supporting
    strut".}  $r(A)$ as the edge $r(A):=(r(A,v_1),r(A,v_2))$, where $r(A,v_i)$ is
the vertex with smallest id that is in the same face in
$\emb(\Gamma(v_i))$ as the virtual edge $e_{v_i}$ incident to $v_i$ on
$A$, but is not an endpoint of $e_{v_i}$.  As a special case, if
$A=m(x,y)$ then treat $x$ and $y$ as having $-\infty$ in this
comparison.  Thus if $G\cup(x,y)$ is planar and $m(x,y)$ is not a
single $P$-node, $r(m(x,y))=(x,y)$.
\end{definition}

\begin{definition}\label{def:embspqr-solid}
  Now consider the graph $\Gamma'(M_r)$ obtained from $\Gamma(M_r)$ by
  adding $r(A_1),\ldots,r(A_k)$ where $A_1,\ldots,A_k$ are the runs
  comprising $M_r$.  This graph is a subdivision of a planar
  triconnected graph and thus has exactly two planar embeddings,
  that are mirror images of each other.  Choose one of these
  embeddings 
  to be $\emb(\Gamma'(M_r))$, and
  define $\emb(\Gamma(M_r))$ to be the corresponding embedding of $M_r$
  obtained by removing the struts from $\emb(\Gamma'(M_r))$.
  In particular, if $G\cup(x,y)$ is planar and $m(x,y)$ is not a
  single $P$-node, the embedding of $\Gamma(m(x,y))$ must have $x$ and
  $y$ in the same face.
\end{definition}

Call a $P$ node on the relevant part of a solid path \emph{up-free} (resp. down-free) if is incident to an $S$-node containing the root-nearest (resp. root-furthest) endpoint of the strut.

 \begin{lemma}\label{lem:unique-flip-bits}
  Let $M_r$ be the relevant part of a solid path. Given
  $\emb(\Gamma(M_r))$, and given local embeddings of all $P$- and $R$-nodes on $M_r$, 
  there is a unique assignment of flip-bits to all
  internal edges on $M_r$, such that for every internal $S$ node or up-free $P$ node $v\in M_r$, the flip-bit of $(v,p(v))$ is $0$, and for every internal down-free $P$-node $v\in M_r$ with heavy child $h(v)$, the flip-bit of $(h(v),v)$ is $0$. 
\end{lemma} 
\begin{proof}
  For an edge between $R$-nodes on $M_r$, this is trivial,
  since every assignment of flip-bits to the edge gives a
  different embedding, and only one of these is $\emb(\Gamma(M_r))$.  If
  $v$ is an internal $S$ node on $M_r$, then (as observed above when we
  described how the flip-bits define an embedding) changing the
  flip-bits of both edges incident to $v$ together does not change the
  embedding. In particular, we can always choose the flip-bit of
  $(v,p(v))$ to be zero. 
  
  If $v$ is a $P$-node, then since the SPQR-tree is pre-split, $v$ has degree exactly $3$. Thus, there are only two different circular orderings of its incident edges, both of which have its heavy edges as neighbours. If $v$ is not free, every setting of flip-bits would lead to a different embedding, and only one of these is $\emb(\Gamma(M_r))$. If $v$ is a free $P$-node, then the flip-bit on the heavy edge connecting it to an endpoint of the strut is irrelevant to the embedding, and can be chosen to be set to $0$.  
\end{proof}

\begin{definition}
  For any biconnected planar graph $G=(V,E)$ and any $x,y\in V$, let
  $\EmbGood(G;x,y)\subseteq\Emb(G)$ be the class of embeddings of $G$
  corresponding to choosing arbitrary embeddings of $P$- and
  $R$-nodes, setting the flip-bits of all internal edges of the
  relevant part of each solid path in the pre-split SPQR tree that is
  $u$-exposed and has root $r$ where $r\cdots u=m(x,y)$ as in
  Lemma~\ref{lem:unique-flip-bits},
  and letting all other flip-bits be arbitrary.
\end{definition}

\begin{lemma}\label{lem:embspqr-flipbit-changecount}
  Changing an edge from solid to dashed or vice versa, or doing a
  merge or split in the pre-split SPQR tree, requires changes to at
  most $\OO(1)$ flip-bits.
\end{lemma}
\begin{proof}
  Only the flip-bits of edges that become internal edges on a relevant
  part of a solid path need to change. 
  
  Changing an edge from dashed to solid, makes one or two edges (where two is in case the child is a $P$-node) internal on a relevant part of a solid path. Vice versa, changing an edge from solid to dashed affects at most two edges, however, their flip-bits are free and need not change.
  
  Upon pre-splitting a $P$- or $S$-node into a parent $p(x_i)$ and at most two children $x_1$ and possibly $x_2$, the edge that arises from the splitting is itself dashed, and each child $x_i$ gets a heavy child. In case we pre-split a $P$-node, nothing happens, but in case for $S$-nodes, each $x_i$ gets a heavy child, that is, one edge changes status from dashed to solid. As already argued, this can lead to at most two edges becoming relevant. The case for merging is symmetric. 
\end{proof}
\begin{lemma}\label{lem:reverse-changecount}
	Reversing the solid root path requires no flips.
\end{lemma}
\begin{proof}
Although many flipbits on internal edges of the solid root path may have to change in order for it to live up to the convention in Lemma~\ref{lem:unique-flip-bits}, these changes in flipbits do not lead to a change in the embedding, but only reflect a different way of describing the same embedding.
\end{proof}

\begin{proof}[Proof of Lemma~\ref{lem:biconn-embgood-insert}]
  By Definition~\ref{def:embspqr-solid} $\emb(\Gamma(m(x,y)))$ and
  hence any embedding $H\in\EmbGood(G-(x,y);x,y)$ admits adding
  $(x,y)$. By Lemma~\ref{lem:ps-insert-uninsert} the insert$(T;(x,y))$
  changes only a constant number of edges, and hence by
  Lemma~\ref{lem:embspqr-flipbit-changecount} at most a constant
  number of flip-bits need to be changed.
\end{proof}
\begin{proof}[Proof of Lemma~\ref{lem:biconn-embgood-mpath}]
  By Lemma~\ref{lem:spqr-mpath-changecount}, the number of edge
  changes when transforming the pre-split SPQR tree is $\OO\paren*{ 1
    + \log\frac{w(T)}{w(x_1)} + \log\frac{w(T)}{w(y_1)} +
    \log\frac{w(T)}{w(x_2)} + \log\frac{w(T)}{w(y_2)} }$, and by
  Lemma~\ref{lem:embspqr-flipbit-changecount} this is also the number
  of flip-bits that need to be changed.
\end{proof}

\subsection{Flip parity}\label{sec:flipp}
The embeddings we construct are recursively ``glued'' together from embeddings of smaller graphs, by replacing a virtual edge $e$ with an embedded subgraph $C_e$. Each time we do this, we have a choice of flipping the subgraph $C_e$ first. We want to make this choice consistently over time, even as parts of the graph $G$ containing $e$ changes. In particular, whenever possible, we want that if the local subgraph of $G$ surrounding $e$ gets flipped, the choice of whether to flip $C_e$ before gluing is reversed.  We can formalize this by introducing the concept of a \emph{flip-parity} bit for each edge $e$ in $G$.

Given a graph $G$, we uniquely partition the edges into \emph{local
  paths}, defined as the subsets of edges, such that any pair of edges
with a common end vertex of degree $2$ in $G$ are in the same set. Two distinct local paths $H_1$ and $H_2$ are \emph{neighbors}, if
there exists edges $e_1\in H_1$ and $e_2\in H_2$ such that $e_1$ and
$e_2$ are adjacent in the cyclic order around some common vertex.

    A \emph{flip-parity function} is a function
    which for every plane embedded multigraph $G$ and every edge $e\in
    G$ assigns a value $\flipparity(G,e)\in\set{0,1}$ such that:
    \begin{itemize}
    \item
      \begin{sloppypar}
        If $G$ is connected and $G'\neq G$ is the mirror embedding of
        $G$, then for every edge $e\in G'$,
        $\flipparity(G',e)=1-\flipparity(G,e)$.
        
       That is, flipping a
	   connected graph with at least one vertex of degree $\geq3$
       toggles the flip parity of every edge.
      \end{sloppypar}

    \item
      \begin{sloppypar}
        If $G'$ is a connected component of $G$, then for every edge
        $e\in G'$, $\flipparity(G',e)=\flipparity(G,e)$.
        
        That is, the flip
          parity is determined independently for each connected component.
      \end{sloppypar}

    \item
      \begin{sloppypar}
        If $e_1$ and $e_2$ are on the same local path in $G$, then
        $\flipparity(G,e_1)=\flipparity(G,e_2)$.
        
        That is, every edge on a
          local path has the same flip parity.
      \end{sloppypar}

    \item If $G=G'\setminus e'$ where $G'$ is a plane embedded
      multigraph, then for every edge $e\in G$, exept edges on the
      local path of $e'$ in $G'$ and the at most $4$ local paths that
      are its neighbors in $G'$,
      $\flipparity(G,e)=\flipparity(G',e)$.
      
      That is, deleting an edge only
        changes the flip-parity on its local path and its neighbors.

    \end{itemize}

    \begin{observation}
If the plane multigraphs $G_1$ and $G_2$ are joined into a
new graph $G'$ by identifying the virtual edges $e_1\in G_1$ and $e_2\in G_2$, then for all edges $e\in G'$ with the exception of
      edges on at most $8$ local paths,
      \begin{align*}
        \flipparity(G',e) &= 
        \begin{cases}
          \flipparity(G_1,e) & \text{if $e\in G_1$}
          \\
          \flipparity(G_2,e) & \text{if $e\in G_2$}
        \end{cases}
      \end{align*}
    \end{observation}
    \begin{observation}\label{obs:fixedflipparity}
      If a separation flip in $G$ flips a subgraph $A$ 
      to create a new graph $G'$, then for all edges $e\in G'$ with
      the exception of edges on at most $8$ local paths,
      \begin{align*}
        \flipparity(G',e) &= 
        \begin{cases}
          1-\flipparity(G,e) & \text{if $e\in A$}
          \\
          \flipparity(G,e) & \text{otherwise}
        \end{cases}
      \end{align*}
    \end{observation}

    The point is that given such a function, we can use
    $\flipparity(\emb(\Gamma(M)),e_c)$ as the flip-bit for a dashed child of
    a node $c$ on $M$ corresponding to the virtual edge $e_c$ in
    $\Gamma(M)$.

    To show that $\flipparity$ functions do exist, we here give an example of one possible definition of a $\flipparity$ function:
    \begin{align*}
      \flipparity(G,e) &:=
      \begin{cases}
        0 & \text{if $C$ is a single local path or $\operatorname{id}(e^-)<\operatorname{id}(e^+)$}
        \\
        1 & \text{otherwise}
      \end{cases}
      \\
      \text{where}&
      \\
      C &= \text{the connected component of $G$ containing $e$}
      \\
      C' &= \parbox[t]{.65\textwidth}{        $C$ with each local path replaced by a single directed edge
        whose id is the minimal edge id on the path, and whose
        direction is such that its tail has degree $\geq3$. If both
        endpoints have degree $3$, the direction is such that the
        vertex-ids on the minimal-id edge are in increasing order.
      }
      \\
      e' &= \text{the edge in $C'$ corresponding to $e$'s local path}
      \\
      e^- &= \text{the clockwise predecessor of $e'$ around its tail in $C'$}
      \\
      e^+ &= \text{the clockwise successor of $e'$ around its tail in $C'$}
      \\
    \end{align*}

    \begin{lemma}
      This definition of $\flipparity(G,e)$ satisfies the requirements.
    \end{lemma}
    \begin{proof}
      If $G$ is connected and $G'\neq G$ is the mirror embedding of
      $G$, then $C=G$ has at least one vertex of degree $\geq 3$, so
      $C$ is not a single local path. By definition, $e^-$ and $e^+$
      will then be neighbors around a node of degree $\geq 3$ in $C'$,
      and since all ids are distinct, they will have distinct ids, and
      so $\flipparity(G,e)=1-\flipparity(G',e)$.
      By definition, the flip-parity is determined independently for
      each component.
      Also by definition, every edge on the same local path has the
      same flip-partity.

      If an edge is deleted, the local path $e'$ that contained it
      breaks. For each edge on the fragments, the flip-parity is
      unchanged. For each end vertex $v$ of the local path $e'$, the
      edges on the local paths that are the clockwise and
      counterclockwise neighbors of $e'$ may change because the
      fragment of $e'$ attached to $v$ disappears or changes its
      id. There are at most $4$ such local paths, and no other local
      paths are affected.
    \end{proof}

\subsection{The canonical embedding for biconnected graphs}\label{sec:canonical}
We have already defined a set of good embeddings of a graph (see Section~\ref{sec:heavySPQRembedding}). However, it has many degrees of freedom, in particular for all $P$-nodes, and all dashed edges. In order to define a canonical embedding, we now proceed to assign these in some consistent manner. 

  We will use a slightly modified version of the technique from
  Section~\ref{sec:SPQRembedding} to assign a unique embedding
  $\emb(G; x,y)$ to every biconnected planar graph $G=(V,E)$ and every
  pair of vertices $x,y\in V$.

  For each node $v$ in the pre-split SPQR tree, select the local
  embedding of $\Gamma(v)$ as follows:

  \begin{itemize}

  \item For each $S$ node the embedding is the only thing it can be.

  \item For each $P$ node, let $e_1,\ldots,e_k$ be its edges,
    ordered by increasing $\operatorname{id}(e_i)$, where the id of
    a virtual edge corresponding to the parent is counted as
    $-\infty$, and the id of a virtual edge corresponding to a child
    $c$ is $\min_{e\in T_c}\operatorname{id}(e)$.

    Choose the local embedding such that the clockwise cyclic order
    around $\min\Gamma(v)$ is $c_1,\ldots,c_k$.

  \item For each $R$ node, choose one of the two possible choices in
    any arbitrary but deterministic way, depending only on
    $\Gamma(v)$.

  \end{itemize}

  Define the \emph{compressed} SPQR tree as the tree obtained from the
  pre-split SPQR tree by contracting the relevant part $M_r$ of each solid
  path $M$ (with at least $2$ nodes) into a single node.

    For each node $v$ in the compressed tree, select the local
    embedding of $\Gamma(v)$ as follows:
    \begin{itemize}

    \item If the node corresponds to a single node in the pre-split
      tree, the local embedding is unchanged.

    \item If the node corresponds to the relevant part $M_r$ of a
      solid path $M$, the local embedding $\emb(\Gamma(v))$ is defined
      as $\emb(\Gamma(M_r))$ from Definition~\ref{def:embspqr-solid}.

    \end{itemize}

    For any edge $(p,c)$ in the compressed tree, corresponding
    to virtual edges $e_c\in\Gamma(p)$ and $e_p\in\Gamma(c)$ we can
    now define    \begin{align*}
      \operatorname{flip-bit}(p,c) &:=
      \flipparity(\emb(\Gamma(p)),e_c)
      \oplus
      \flipparity(\emb(\Gamma(c)),e_p)
    \end{align*}
    The point of this definition is that this guarantees that
    \begin{align*}
      \flipparity(\emb(\Gamma(p)),e_c)
      &=
      \operatorname{flip-bit}(p,c)
      \oplus
      \flipparity(\emb(\Gamma(c)),e_p)
    \end{align*}
    thereby forcing (most of) the children to follow the embedding of
    the parent even as that graph changes.

    Finally define $\emb(G; x,y)$ using the simple recursive
    definition on the compressed tree.

  \begin{lemma}
    Pre-splitting a $P$ node causes $O(1)$ flips.\label{lem:presplit-p}
  \end{lemma}\begin{proof}
Before the $P$-node is split into a node $x$ and its parent $p(x)$ such that $p(x)$ has degree $3$, a slide flip may be needed to ensure that the edges that become incident to $p(x)$ are neighbours.
After the split, $x$ may obtain at most one solid child, which does not cause a change in flip-bits. The ordering of the children of $x$ is unchanged, as their ids are unchanged.
Thus, this operation corresponds in total to at most one slide flip.
Such a slide flip may change the flip-parity of up to $4$ edges in 
$\Gamma(M)$ where $M$ is the solid path containing the node that was split. 
\end{proof}

  \begin{lemma}\label{lem:presplit-s}
    Pre-splitting an $S$ node at the end of a solid path causes $O(1)$ flips.
  \end{lemma}

\begin{proof}
Assume we split the $S$-node into nodes $s_1$, possibly $s_2$, and $p(s_i)$. Then, each $s_i$ obtains at most one solid child, and we may have to toggle the flipbit on that one solid edge. Since it is neighbour to an $S$-node, this solid child can either be a $P$-node, or an $R$-node. If the solid child is an $R$-node, we may have to change the flip-bit on the edge between them. If the solid child is a $P$-node, then we could be in the case where a $P$-node becomes internal on the relevant part of a path, and we must change the flip-bit on its other solid edge. Since the at most $2$ nodes denoted by $s_i$ cause at most $2$ flips each, and since no flips are incurred by $p(s_i)$, the operation leads to a total of $4$ flips. 
Each of these flips lead to only $O(1)$ changes in flip-parity.
\end{proof}

Note that it was important that we had invariantly pre-split any P-node with a solid child (and non-P parent), for the cascade to stop here.

\begin{lemma}
Splicing two solid paths by making a dashed edge
solid causes $O(1)$ flips.
\end{lemma}
\begin{proof}
Note first that joining two heavy paths forms a new contracted node $v' = \operatorname{contract}(v,p(v))$, whose smallest element-label may be smaller than that of $p(v)$. If the parent of $v'$ is a $P$-node, we may thus see a change in the ordering of its edges.

The splice may itself change a flip-bit, causing one flip itself, and at most $8$ flips due to Observation~\ref{obs:fixedflipparity}.

When the affected edge is between $R$-nodes, no pre-splitting is necessary.

When the affected edge involves an $S$-node, then pre-splitting the $S$-node (see Lemma~\ref{lem:presplit-s}) may lead to a change in $O(1)$ flip-bits.

When the affected edge involves a $P$-node, then possibly pre-splitting a $P$-node (see Lemma~\ref{lem:presplit-p}) may lead to $O(1)$ flip, and solidifying a dashed edge may increase the number of relevant edges by $2$, both of which may change flip-bit, incurring $O(1)$ flips.
\end{proof}

\begin{lemma}
  Reversing the solid root path causes $O(1)$ flips.
\end{lemma}
\begin{proof}
The reverse operation itself can lead to one reflect flip, reflecting the entire embedding. Furthermore, if there are $P$-nodes in the end of the root path, each of these may lead to a slide flip, due to the convention of treating the parent edge as having $\textnormal{id} = -\infty$.
Each slide-flip may lead to only $O(1)$ further flips.
\end{proof}

\begin{lemma}
    When adding $(x,y)$ to $\emb(\Gamma(m(x,y)))$ and contracting
    $m(x,y)$, at most $4$ dashed edges change their flip-bits and
    cause flips.
\end{lemma}
\begin{proof}
	The only case where edges may change flip-bits is when one or both of $x$ and $y$ lie on an $S$-node, in which case the edges between the $S$-node and its pre-split-children may change their bit after a contraction. Since this happens in at most $2$ places (both endpoints), an since there are at most $2$ pre-split children, this yields at most $4$ changes in total.
\end{proof}

Thus, by the above lemmas, and by symmetry, we may conclude the following: 

  \begin{corollary}
    Each basic operation on the pre-split SPQR tree results in $O(1)$
    separation flips.
  \end{corollary}

\begin{proof}[Proof of Theorem~\ref{thm:canonical-emb} for biconnected graphs]
The canoncial embedding $\emb(G)$ defined here lives up to the criteria.
\end{proof}

\section{Embeddings of general planar graphs}\label{sec:dynEmb-general}

In this section, we generalise the definitions of flip bits, good embedding, and canonical embedding, from biconnected graphs to connected components of a graph. 
We show how to endow the BC-tree with information that specifies the embedding (Section~\ref{sec:BCembedding}), we define good embeddings (Section~\ref{sec:EmbGoodBiconn}), and we fix the canonical embedding by making arbitrary but consistent choices when the good embeddings are ambiguous (Section~\ref{sec:canonBC}).

\subsection{Specifying an embedding in a BC tree}\label{sec:BCembedding}
Given a connected planar graph $G=(V,E)$, and any relaxed BC tree for
it, we can specify an embedding of $G$ by choosing:
\begin{itemize}
\item A \emph{root} for the BC tree.
\item A \emph{flip-bit} for each BC edge. Changing this flip-bit corresponds to flips of the \emph{reflect} type.
\item A \emph{local embedding} $\emb(\Gamma(v))\in\Emb(\Gamma(v))$ for
  each $B$ node $v$ in the BC tree.
\item A \emph{local order} of the children of each $C$ node $c$.
\item A \emph{glue corner}\footnote{For any embedded graph, the \emph{corner} between a pair of consecutive edges corresponds to an edge of its medial graph; see e.g.~\cite{DBLP:journals/mst/HolmR17} for a definition.} in the graph $\emb(\Gamma(b))$ for each BC tree
  edge $(b, c)$ incident to a $B$ node $b$.  (We need a glue corner for
  each side of every edge, but the remaining corners are given by the
  local order at each $C$ node). Changing this glue corner corresponds to flips of the \emph{slide} type. 
\end{itemize}
Given these choices, we can now construct a unique embedding of $\Gamma(T_v)$ recursively for each BC node $v$ as follows:
\begin{itemize}
\item If $v$ is a leaf, it must be a $B$ node, and we just use the chosen
  local embedding $\emb(\Gamma(v))$.
\item If $v$ is a non-leaf $B$ node take each child $c_i$, construct
  $\emb(\Gamma(T_{c_i}))$ recursively, flip it if the flip-bit for
  $(v,c_i)$ is set, and then attach it at the glue corner chosen for
  $(v,c_i)$.
\item If $v$ is a $C$ node take each child $c_i$, construct
  $\emb(\Gamma(T_{c_i}))$ recursively, flip it if the flip-bit for
  $(v,c_i)$ is set, and then join the glue corners chosen for $(v,c_i)$ in
  the order chosen for $v$.  The corner between the first edge of $c_0$ and the last
  edge of $c_{\max}$ will be the glue corner by which this graph is attached to
  its parent.
\end{itemize}
Note that not all possible embeddings of $G$ can be described this
way\footnote{E.g. most embeddings of the so-called
  \emph{friendship graph} are excluded, but we retain exactly those
  embeddings where every vertex is in the outer face.}, but we will
restrict our attention to those that can.
Note also, that not all combinations of these choices give distinct
embeddings, e.g.:
\begin{itemize}
\item Changing the root either does nothing or flips the whole
  embedding.
\item For each non-root $B$ node $v$, if we change all the incident
  flip bits and flip the local embedding of $\Gamma(v)$ then we get
  exactly the same embedding of $G$.
\item For each non-root $C$ node $v$, if we change all the incident
  flip bits and reverse the local order at $v$ then we get
  exactly the same embedding of $G$.
\end{itemize}

\subsection{Good embeddings}\label{sec:EmbGoodBiconn}
Let $x,y\in V$ be distinct vertices in the connected planar graph
$G=(V,E)$.  Let $T$ be the unique pre-split BC-tree for $G$ that is
$u$-exposed and has root $r$ such that $m(x,y)=r\cdots u$.

We will define a set $\EmbGood(G; x,y)$ of good embeddings of $G$ by
defining flip-bits for (most of) the solid edges in $T$.  Then define
\begin{align*}
  \EmbGood(G) &:= \bigcup_{x,y\in V}\EmbGood(G; x,y)
\end{align*}

Given a solid path $M$ with a strut $(a,b)$ in a BC tree, we can define
a subset of good embeddings $\EmbSolid(\Gamma(M); a,b)$ for $\Gamma(M)$ as
follows:
\begin{itemize}
\item If $G_M=\Gamma(M)\cup(a,b)$ is planar, then since $G_M$ is
  biconnected, $\EmbGood(G_H; a,b)$ is defined in
  Section~\ref{sec:flip-dist} and we can just set
  \begin{align*}
    \EmbSolid(\Gamma(H); a,b) := \set*{
      A-(a,b)
      \suchthat
      A\in\EmbGood(G_H; a,b)
    }
  \end{align*}
\item Otherwise, in the pre-split SPQR tree for $G_M$, the path
  $m(a,b)$ in the pre-split SPQR-tree consists of a single $S$ node $r$, and in each child $c_i$
  the virtual edge $(x_i,y_i)$ corresponding to the parent edge has
  the property that if the corresponding subgraph $G_i$ is not planar
  then $G_i-(x_i,y_i)$ is biconnected and planar.  Thus for each
  $G_i$, if it is planar we can use any embedding in $\EmbGood(G_i)$
  as the local embedding, otherwise we can start with any embedding in
  $\EmbGood(G_i-(x_i,y_i))$ and add the virtual edge $(x_i,y_i)$ at
  arbitrary corners around $x_i$ and $y_i$ to get a local embedding.
  Finally, we can combine these local embeddings into an embedding of
  $\Gamma(M)$ as usual.  The set of all embeddings of $\Gamma(M)$ that
  can result from this process is then $\EmbSolid(\Gamma(M); a,b)$.
\end{itemize}

\begin{definition}
  Now for any connected planar graph $G=(V,E)$, and vertices $x,y\in V$,
  we can define $\EmbGood(G; x,y)$ based on the $u$-exposed pre-split BC
  tree with root $r$, where $r\cdots u=m(x,y)$, as the set of all
  embeddings where for each solid path $M$ with strut $(x',y')$, the
  local embedding of $\Gamma(M)$ is in $\EmbSolid(\Gamma(M); x',y')$.
\end{definition}

\begin{lemma}\label{lem:gen-embgood-insert}
  For any planar graph $G$ with $n$ vertices, and any edge $(x,y)\in
  G$: $\flipdist(\EmbGood(G-(x,y); x,y), \EmbGood(G; x,y))=1$.
\end{lemma}
\begin{proof}
  By definition $\EmbGood(G-(x,y);x,y)$ and $\EmbGood(G;x,y)$ are
  defined based on the same pre-split BC tree, and $(x,y)$ is already
  inserted as a strut. Converting the strut to a real edge does not
  change anything else in the embedding, so the only change to the
  embedding needed when inserting $(x,y)$ is the actual insertion of
  the edge. 
\end{proof}

\begin{lemma}\label{lem:gen-embgood-mpath}
  For any planar graph $G$ with $n$ vertices, and any two pairs of
  distinct vertices $x_1,y_1$ and $x_2,y_2$: $\flipdist(\EmbGood(G;
  x_1,y_1), \EmbGood(G; x_2,y_2))\in \OO(\log n)$.
\end{lemma}
\begin{proof}
  The minimal flip-distance is determined by the total number of edge
  changes when changing from a pre-split BC tree that is $u_1$-exposed
  with root $r_1$ where $r_1\cdots u_1=m(x_1,y_1)$ to one that is
  $u_2$-exposed with root $r_2$ where $r_2\cdots u_2=m(x_2,y_2)$ is
  bounded by the number of edges changing between solid and dashed in
  the BC and SPQR trees, and the number of merges and splits.  By
  Lemma~\ref{lem:dyntree-internal-telescope}   this number
  is $\OO(\log n)$. 
  Thus, given any embedding $\emb_1\in \EmbGood(G;x_1,y_1)$, performing those $\OO(\log n)$ local changes will yield a local embedding in $\EmbGood(G;x_2,y_2)$ (and a symmetric argument applies to any $\emb_2\in \EmbGood(G;x_2,y_2)$). 
\end{proof}

\begin{theorem}\label{thm:good-emb-bc}
  For any planar graph $G$ with $n$ vertices, and any edge $e\in G$:
  $\flipdist(\EmbGood(G-e),\EmbGood(G))\in\OO(\log n)$
\end{theorem}
\begin{proof}
  Let $e=(x,y)$. Then by the triangle inequality
  \begin{align*}
    \flipdist(\EmbGood(G-e),\EmbGood(G))
    &\leq \flipdist(\EmbGood(G-e),\EmbGood(G-e;x,y))
    \\
    &\quad + \flipdist(\EmbGood(G-e;x,y),\EmbGood(G;x,y))
    \\
    &\quad + \flipdist(\EmbGood(G;x,y),\EmbGood(G))
  \end{align*}
  By Lemma~\ref{lem:gen-embgood-mpath} the first and last of these
  terms is $\OO(\log n)$ and by Lemma~\ref{lem:gen-embgood-insert}
  the middle term is $1$.
\end{proof}

\subsection{Glue-corners and flip-parity}

Similarly to Section~\ref{sec:flipp}, the embeddings we construct from BC-trees are recursively ``glued'' by combining smaller embedded graphs $\emb(G_1),\emb(G_1)$ at their \emph{glue corners} $c_1,c_2$. Each time we do this, we have the choice of which corners to glue in, and that of whether or not to flip the child before gluing.  We want to make these choices consistently over time, even as parts of the graph $G_i$ containing $c_i$ changes. In particular, whenever possible, we want the glue corner to stay fixed, and, if the local subgraph of $G_i$ surrounding $c_i$ gets flipped, we want that the choice of whether to flip $\emb(G_i)$ before gluing is reversed.  We can formalize this by introducing the concept of a \emph{glue-corner function}, and defining a  \emph{flip-parity} bit for each corner $G_i$.

A \emph{glue-corner function} is a function which for every plane embedded multigraph $G$ and every vertex $v\in G$ assigns a corner $\gluecorner(G,v)\in G$ such that:
\begin{itemize}
\item If $G$ is connected and $G'\neq G$ is the mirror embeding of
  $G$, then for every vertex $v\in G$, $\gluecorner(G,v)=\gluecorner(G',v)$.

  That is, flipping a connected graph with at least one vertex of
  degree $\geq 3$ does not change any glue corners.

\item If $G'$ is a connected component of $G$, then for every vertex
  $v\in G'$, $\gluecorner(G,v)=\gluecorner(G',v)$.

  That is, the glue corners are determined independently for each
  connected component.

\item If $G=G'\setminus e'$ where $G'$ is a plane embedded multigraph,
  then for every vertex $v\in G$, exept vertices on the local path of
  $e'$ in $G'$ or vertices that are internal to the at most $4$ local
  paths that are its neighbors in $G'$,
  $\gluecorner(G,v)=\gluecorner(G',v)$.

  Furthermore, for every two vertices $v_1,v_2$ of degree $2$ that are
  on the same local path,
  \begin{align*}
    \gluecorner(G,v_1)=\gluecorner(G',v_1)
    \iff
    \gluecorner(G,v_2)=\gluecorner(G',v_2)
  \end{align*}

  That is, deleting an edge only changes the glue-corners on its local
  path and its neighbors, and if one vertex of degree $2$ changes its
  glue corner, then all other degree $2$ vertices on the same local
  path changes their glue-corner as well.
\end{itemize}

In addition to the glue-corner function we also need to define a corresponding \emph{flip-parity function for corners}, as a function which for every plane embedded multigraph $G$ and every corner $c\in G$ assigns a value $\flipparity(G,c)\in\set{0,1}$ such that:
\begin{itemize}

\item
  \begin{sloppypar}
    If $G$ is connected and $G'\neq G$ is the mirror embedding of
    $G$, then for every corner $c\in G'$,
    $\flipparity(G',c)=1-\flipparity(G,c)$.

    That is, flipping a connected graph with at least one vertex of
    degree $\geq3$ toggles the flip parity of every corner.
  \end{sloppypar}

\item
  \begin{sloppypar}
    If $G'$ is a connected component of $G$, then for every corner
    $c\in G'$, $\flipparity(G',c)=\flipparity(G,c)$.

    That is, the flip
    parity is determined independently for each connected component.
  \end{sloppypar}

\item If $G=G'\setminus e'$ where $G'$ is a plane embedded multigraph,
  then for every vertex $v\in G$ of degree $2$, we have
  \begin{align*}
    \gluecorner(G,v)&\neq\gluecorner(G',v)
    \\
    &\Updownarrow
    \\
    \flipparity(G,\gluecorner(G,v))&=1-\flipparity(G',\gluecorner(G',v))
  \end{align*}

  For all other vertices, exept vertices on the local path of
  $e'$ in $G'$ or vertices that are at the end of the at most $4$ local
  paths that are its neighbors in $G'$,
  $\flipparity(G,\gluecorner(G,v))=\flipparity(G',\gluecorner(G',v))$.

  That is, when deleting an edge, the flips caused by the changing
  glue corners are grouped such that whole local paths can be
  flipped at the same time.
\end{itemize}

\begin{observation}
  If the plane multigraphs $G_1$ and $G_2$ are joined into a new graph
  $G'$ by gluing them together at corners $c_1\in G_1$ and $c_2\in
  G_2$ at some common vertex, then for all vertices $v\in G'$ with the
  exception of vertices on at most $4$ local paths,
  \begin{align*}
    \gluecorner(G',v) &= 
    \begin{cases}
      \gluecorner(G_1,v) & \text{if $v\in G_1$}
      \\
      \gluecorner(G_2,v) & \text{if $v\in G_2$}
    \end{cases}
  \end{align*}
\end{observation}

\begin{observation}
  If an articulation flip in $G$ flips a subgraph $A$ 
  to create a new graph $G'$, then for all vertices $v\in G'$ with
  the exception of vertices on at most $4$ local paths,
  \begin{align*}
    \gluecorner(G',v) &= \gluecorner(G,v)
  \end{align*}
\end{observation}

The point is that given such functions, we can use
$\flipparity\paren[\Big]{\emb\paren[\big]{\Gamma(b)},\,
  \gluecorner\paren[\big]{\emb\paren[\big]{\Gamma(b)},c} }$ as the flip-bit for any edge $(b,c)$ where $b$ is a $B$ node and $c$ is a $C$ node.

    To show that these functions do exist, we here give an example of one possible definition:

Given a vertex $v$ in a plane embedded graph $G$, define
\begin{align*}
  \gluecorner(G,v) &:=
  \begin{cases}
    c^- & \text{if $\flipparity(G,e_0)=0$ \qquad (see Section~\ref{sec:flipp})}
    \\
    c^+ & \text{otherwise}
  \end{cases}
  \\
  \text{where}&
  \\
  e_0 &= \text{the minimum-id edge incident to $v$ in $G$}
  \\
  c^- &= \text{the corner clockwise preceding $e_0$ around $v$ in $G$}
  \\
  c^+ &= \text{the corner clockwise succeeding $e_0$ around $v$ in $G$}
\end{align*}

Given a corner $c$ between clockwise consecutive edges $e^-,e^+$ (not necessarily distinct) in a plane embedded graph $G$, define
\begin{align*}
  \flipparity(G,c) :=
  \begin{cases}
    0 & \text{if $\operatorname{id}(e^-)\leq\operatorname{id}(e^+)$}
    \\
    1 & \text{otherwise}
  \end{cases}
\end{align*}

By construction, these definitions of flip-parity and glue-corner have the properties listed above.

\subsection{The canonical embedding}\label{sec:canonBC}

We will use a slightly modified version of the technique from
Section~\ref{sec:BCembedding} to assign a unique embedding $\emb(G;
x,y)$ to every connected planar graph $G=(V,E)$ and 	every pair of
vertices $x,y\in V$.

Define the \emph{compressed} BC tree as the tree obtained from the
pre-split BC tree by contracting each solid path into a single node.

For each node $v$ in the compressed tree, select the local embedding
of $\Gamma(v)$ as follows:
\begin{itemize}
\item For each node in the compressed tree that is either an isolated
  $B$ node or corresponds to a path $M$, we can define the local
  embedding as the embedding of $\Gamma(M)$ from
  step~\ref{sec:EmbGoodBiconn}.

\item For each $C$ node, let $e_1,\ldots,e_k$ be the neighboring
  edges, ordered by increasing $\operatorname{id}(e_i)$, where the id
  of a virtual edge corresponding to the parent is counted as
  $-\infty$, and the id an edge corresponding to a child $c$
  is $\min_{v\in T_c}\operatorname{id}(v)$.

  Choose the local order of the children such that the clockwise
  cyclic order around $v$ is $c_1,\ldots,c_k$.
\end{itemize}

For any edge $(b,c)$ where $b$ is a $B$ node and $c$ is a $C$ node, we
can now define
\begin{align*}
  \operatorname{flip-bit}(c,b) := \flipparity\paren[\Big]{\emb\paren[\big]{\Gamma(b)},\,
  \gluecorner\paren[\big]{\emb\paren[\big]{\Gamma(b)},c} }
\end{align*}
thereby forcing (most of) the children to follow the embedding of the
parent even as that graph changes.

Finally define $\emb(G; x,y)$ using the simple recursive
definition on the compressed tree.

With these definitions in place, the proof of the following lemma goes along the same lines as in Section~\ref{sec:canonical}.

  \begin{lemma}
    For any planar graph $G$ with $n$ nodes, and any two pairs of
    distinct vertices $x_1,y_1$ and $x_2,y_2$: $\flipdist(\emb(G;
    x_1,y_1), \emb(G; x_2,y_2))\in\OO(\log n)$.\todo{proof?}
  \end{lemma}

\begin{proof}[Proof of Theorem~\ref{thm:canonical-emb}]
For any planar graph $G$ with $n$ nodes, and any edge $e\in G$:
$\flipdist(\emb(G-e), \emb(G)-e)\in\OO(\log n)$.
Thus, $\emb$ serves as a canonical embedding with the desired properties.
\end{proof} 

\subsection{Incremental planarity testing in $O(\log^3 n)$ insertion time}\label{sec:incrPlanrBicon}
We now have a structure for maintaining a planar embedding of an incremental graph and answering queries to whether an edge is compatible with planarity (not just with the embedding) of the graph:

\begin{proof}[Proof of Thm.~\ref{thm:incrplanarity}]
Let $G$ be a graph. Embed each component of $G$ with the canonical embedding, and build the data structure $M$ for maintaining alterable planar embeddings~\cite{DBLP:journals/mst/HolmR17}.
Let $(a,b)$ be an edge to be inserted or queried.
Inserting $(a,b)$
corresponds to at most $\OO(\log n)$ flip-bit changes. Changing the flip-bit of the SPQR-edge that corresponds to the separation pair $c,d$ yields a separation-flip in the pair $c,d$, which the data structure $M$ handles in $O(\log^2 n)$ time. 
Changing the embedding information of an articulation point corresponds to an articulation flip which $M$ handles in $O(\log^2 n)$ time.
Finally, query the data structure $M$ for whether the edge $(a,b)$ may be inserted across a face. 
Now if this is an update, then if the answer is yes, insert the edge, if no, report that the graph is no longer planar. If this is a query, report the answer, and reverse all operations.

Since there were $O(\log n)$ updates to the data structure $M$, and $M$ processes each update in $O(\log ^2 n)$ time, the total update or query time is $O(\log^3 n)$.
\end{proof}

\section{Implementation details}\label{sec:implementation}

We are now ready to provide the final details that complete the proofs of Theorem~\ref{thm:incr3con}, Theorem~\ref{thm:incrplanarity}, and Lemma~\ref{lem:intro-biaseddyntree}.
\subsection{Handling child lists}

We use globally biased search trees to track the dashed children of
each node.  Depending on the type of split we want to enable for a
given node, we need different versions of this structure.

For $C$ nodes in the BC tree, and $P$ nodes in the SPQR tree, we just want to split off one single child from the children's list. For $S$ nodes, on the other hand, we have a natural circular ordering of its neighbours, that divides its light children into two (possibly trivial) groups, namely those on either side of the heavy path. For $S$ nodes, we need to split off these two groups of light children separately.

The original dynamic trees~\cite{DBLP:journals/jcss/SleatorT83} used a
single globally biased search tree for all children of each node,
ordered left to right by their weight.  Adding or deleting a dashed
child $c$ of $v$ in this structure then costs
$\OO(\log\frac{w(T_v)-w(v)}{w(T_c)})$, which is exactly what we need for the $P$ and $C$ nodes.
However, this structure only supports a type of split where a single
child is separated from the rest.

A simple variation is to have all children ordered, and use one tree
for all children that are to the left of the solid edge, and another
tree for those that are to the right (using just a single tree when
there are no solid children).  Changing the solid/dashed state for a
child $c$ of $v$ then corresponds to a split/join on these trees,
which again costs $\OO(\log\frac{w(T_v)-w(v)}{w(T_c)})$.  This
structure only supports a different type of split, where (some
consequtive subsequence of) the left and/or right children are moved
to a new child. This is precicely the type of split we need for the
$S$ nodes.

\subsection{Handling solid paths}

Following~\cite{DBLP:journals/jcss/SleatorT83}, we want to use a
biased search tree for each solid path.  This works really well for
our biased dynamic trees (except for a tiny matter of handling zero
weights). However, when maintaining heavy paths for a tree
decomposition, our system of weights depend on the choice of root, so
we can not use them (directly) for balancing a biased tree over each
solid path.  But in order to maintain the correct set of solid paths
we need the be able to find the correct light edges, based on the
correct weights.

Fortunately, these two things do not need to be tied together. We can
use one algorithm, with one system of weights, to define and maintain
the shape of the biased search tree, and a different algorithm and
system of weights to determine the light edges.  We will do this, even
for the basic biased dynamic trees.

Given a solid path $H$ with bottom node $b$, let
$I(H)=V[H]\setminus\set{b,r}$, and let $t=\abs{I(H)}$, $u_0=b$ and
$u_{i+1}=p(u_i)$ for $i\in\set{0,\ldots,t}$.  In particular,
$I(H)=\set{u_1,\ldots,u_t}$ and the neighbors $u_0$ and $u_{t+1}$ are
well-defined.

We will keep our biased tree over $u_1\cdots u_t$ rather than over all
of $H$, because this vastly reduces the number of expensive update
operations we need to make on biased trees and makes reweight$(r,i)$
and reverse$(r)$ simpler.

\paragraph{Biased balancing for general biased dynamic trees}
To handle zero-weight leaves in the biased tree, we can replace every
zero weight with $\frac{1}{k}$. By definition of $k$-positive weights
(Definiton~\ref{def:k-positive}) this at most doubles the total weight
of the biased tree, so it adds at most a constant to the depth of every
node.  In particular, the time for a split or $3$-way join at $v$ is
$\OO\paren*{1+\frac{W}{\max\set{w(v),\frac{1}{k}}}}$, which is
$\OO\paren*{1+\frac{W}{w(v)}}$ whenever $w(v)\geq 1$. Thus the usual
telescoping sum still works, and we get the desired result.

However, we have to keep track of the \emph{actual} weight as well, in
order to find the correct light edges.

\begin{proof}[Proof of Lemma~\ref{lem:dyntree-internal-optimes}]
  This follows directly from
  using~\cite{DBLP:journals/jcss/SleatorT83} with our modified
  weights, and with the small change of keeping the bottom and root
  nodes out of the biased search tree for each solid path.
\end{proof}

\bibliography{references}

\end{document}